\documentclass[aps,amsmath,onecolumn,amssymb]{revtex4}
\usepackage{amssymb}
\usepackage{amsthm}
\usepackage{amsfonts}
\usepackage{complexity}
\usepackage{graphicx}
\usepackage{dcolumn}
\usepackage{bm}
\usepackage{hyperref}
\usepackage{enumerate}
\usepackage{algorithm}
\usepackage{algpseudocode}

\usepackage{multirow}

\newcommand{\sm}{appendix}

\newtheorem{theorem}{Theorem}
\newtheorem{lemma}{Lemma}
\newtheorem{definition}{Definition}
\newtheorem{corollary}{Corollary}

\usepackage[usenames,dvipsnames]{xcolor}
\hypersetup{
    colorlinks=true,       
    linkcolor=Maroon,          
    citecolor=OliveGreen,        
    filecolor=magenta,      
    urlcolor=Blue           
}

\begin{document}

\def\ket#1{\left|#1\right\rangle}
\def\bra#1{\langle#1|}
\newcommand{\ketbra}[2]{|#1\rangle\!\langle#2|}
\newcommand{\braket}[2]{\langle#1|#2\rangle}
\newcommand{\prob}[1]{{\rm Pr}\left(#1 \right)}
\newcommand{\expect}[2]{{\mathbb{E}_{#2}}\!\left\{#1 \right\}}
\newcommand{\var}[2]{{\mathbb{V}_{#2}}\!\left\{#1 \right\}}


\newcommand{\sde}{\mathrm{sde}}
\newcommand{\Z}{\mathbb{Z}}
\newcommand{\w}{\omega}
\newcommand{\Kap}{\kappa}

\newcommand{\Tchar}{$T$}
\newcommand{\T}{\Tchar~}
\newcommand{\ClT}{\{{\rm Clifford}, \Tchar\}~}
\newcommand{\Tcount}{\Tchar--count~}
\newcommand{\Tcountper}{\Tchar--count}
\newcommand{\Tcounts}{\Tchar--counts~}
\newcommand{\Tdepth}{\Tchar--depth~}
\newcommand{\Zr}{\Z[i,1/\sqrt{2}]}
\newcommand{\ve}{\varepsilon}

\newcommand{\eq}[1]{\hyperref[eq:#1]{(\ref*{eq:#1})}}
\renewcommand{\sec}[1]{\hyperref[sec:#1]{Section~\ref*{sec:#1}}}
\newcommand{\app}[1]{\hyperref[app:#1]{Appendix~\ref*{app:#1}}}
\newcommand{\fig}[1]{\hyperref[fig:#1]{Figure~\ref*{fig:#1}}}
\newcommand{\thm}[1]{\hyperref[thm:#1]{Theorem~\ref*{thm:#1}}}
\newcommand{\lem}[1]{\hyperref[lem:#1]{Lemma~\ref*{lem:#1}}}
\newcommand{\tab}[1]{\hyperref[tab:#1]{Table~\ref*{tab:#1}}}
\newcommand{\cor}[1]{\hyperref[cor:#1]{Corollary~\ref*{cor:#1}}}
\newcommand{\alg}[1]{\hyperref[alg:#1]{Algorithm~\ref*{alg:#1}}}
\newcommand{\defn}[1]{\hyperref[def:#1]{Definition~\ref*{def:#1}}}

\newcommand{\targfix}{\qw {\xy {<0em,0em> \ar @{ - } +<.39em,0em>
\ar @{ - } -<.39em,0em> \ar @{ - } +
<0em,.39em> \ar @{ - }
-<0em,.39em>},<0em,0em>*{\rule{.01em}{.01em}}*+<.8em>\frm{o}
\endxy}}

\newenvironment{proofof}[1]{\begin{trivlist}\item[]{\flushleft\it
Proof of~#1.}}
{\qed\end{trivlist}}

\newcommand{\cu}[1]{{\textcolor{red}{#1}}}
\newcommand{\tout}[1]{{}}
\newcommand{\beq}{\begin{equation}}
\newcommand{\eeq}{\end{equation}}
\newcommand{\beqa}{\begin{eqnarray}}
\newcommand{\good}{{\rm good}}
\newcommand{\bad}{{\rm bad}}
\newcommand{\eeqa}{\end{eqnarray}}

\newcommand{\id}{\openone}
\title{Quantum Deep Learning}
\author{Nathan Wiebe}
\author{Ashish Kapoor}
\author{Krysta M.~Svore}

\affiliation{Microsoft Research, Redmond, WA (USA)}

\begin{abstract}
In recent years, deep learning has had a profound impact on machine learning and artificial intelligence. At the same time, algorithms for quantum computers have been shown to efficiently solve some problems that are intractable on conventional, classical computers. We show that quantum computing not only reduces the time required to train a deep restricted Boltzmann machine, but also provides a richer and more comprehensive framework for deep learning than classical computing and leads to significant improvements in the optimization of the underlying objective function. Our quantum methods also permit efficient training of full Boltzmann machines and multi–layer, fully connected models and do not have well known classical counterparts.
\end{abstract}

\maketitle

\section*{Introduction}
We present quantum algorithms to perform deep learning that outperform conventional, state-of-the-art classical algorithms in terms of both training efficiency and model quality. Deep learning is a recent technique used in machine learning that has substantially impacted the way in which classification, inference, and artificial intelligence (AI) tasks are modeled~\cite{HOT06,CW08,Ben09,LYK+10}.  It is based on the premise that to perform sophisticated AI tasks, such as speech and visual recognition, it may be necessary to allow a machine to learn a model that contains several layers of abstractions of the raw input data.  For example, a model trained to detect a car might first accept a raw image, in pixels, as input.  In a subsequent layer, it may abstract the data into simple shapes. In the next layer, the elementary shapes may be abstracted further into aggregate forms, such as bumpers or wheels.  At even higher layers, the shapes may be tagged with words like ``tire" or ``hood".  Deep networks therefore automatically learn a complex, nested representation of raw data similar to layers of neuron processing in our brain, where ideally the learned hierarchy of concepts is (humanly) understandable.
In general, deep networks may contain many levels of abstraction encoded into a highly connected, complex graphical network; training such graphical networks falls under the umbrella of deep learning.

Boltzmann machines (BMs) are one such class of deep networks, which formally are a class recurrent neural nets with undirected edges and thus provide a generative model for the data.  From a physical perspective, Boltzmann machines model the training data with an Ising model that is in thermal equilibrium. These spins are called units in the machine learning literature and encode features and concepts while the edges in the Ising model's interaction graph represent the statistical dependencies of the features. The set of nodes that encode the observed data and the output are called the visible units ($v$), whereas the nodes used to model the latent concept and feature space are called the hidden units ($h$).  Two important classes of BMs are the restricted Boltzmann machine (RBM) which takes the underlying graph to be a complete bipartite graph, and the deep restricted Boltzmann machine which is composed of many layers of RBMs  (see Figure~\ref{fig:BMdiag}).  
For the purposes of discussion, we assume that the visible and hidden units are binary.

A Boltzmann machine models the probability of a given configuration of visible and hidden units by the Gibbs distribution (with inverse temperature $1$):
\begin{equation}
P(v,h) = e^{-E(v,h)}/Z,~\label{eq:rbmformula}
\end{equation}
where $Z$ is a normalizing factor known as the \emph{partition function} and the energy $E(v,h)$ of a given configuration $(v, h)$ of visible and hidden units is given by
\begin{equation}
E(v,h)=-\sum_i v_i b_i - \sum_j h_j d_j - \sum_{i,j} w^{vh}_{ij}v_ih_j- \sum_{i,j} w^{v}_{ i,j}v_iv_j-\sum_{i,j} w^{h}_{ i,j}h_ih_j.\label{eq:E}
\end{equation}
Here the vectors $b$ and $d$ are \emph{biases} that provide an energy penalty for a unit taking the value $1$ and $w^{v,h}_{i,j}$, $w^v_{i,j}$, and $w^h_{i,j}$ are \emph{weights} which assign an energy penalty if the visible and hidden units both take value $1$. We denote $w=[w^{v,h},w^v,w^h]$ and let $n_v$ and $n_h$ be the numbers of visible and hidden units, respectively.

Given some a priori observed data, referred to as the training set, {\em learning} for these models proceeds by modifying the strengths of the interactions in the graph to maximize the likelihood of the Boltzmann machine producing the given observations.  Consequently, the training process uses gradient descent to find weights and biases that optimize the maximum--likelihood objective (ML--objective):
\begin{equation}
O_{\rm ML} := \frac{1}{N_{\rm train}}\sum_{v\in x_{\rm train}} \log\left(\sum_{h=1}^{n_h} P(v,h)\right) - \frac{\lambda}{2}w^T w,
\end{equation}
where $N_{train}$ is the size of the training set, $x_{\rm train}$ is the set of training vectors, and $\lambda$ is an $L2$--regularization term to combat overfitting.  The derivative of $O_{\rm ML}$ with respect to the weights is
\begin{eqnarray}
\frac{\partial O_{\rm ML}}{\partial w_{i,j}}&=&\left\langle v_ih_j \right\rangle_{\rm data}-\left\langle v_ih_j \right\rangle_{\rm model}-\lambda w_{i,j},\label{eq:logderiv}
\end{eqnarray}
where the brackets denote the expectation values over the data and model for the BM.
The remaining derivatives take a similar form~\cite{Hin02}.

Computing these gradients directly from (\ref{eq:rbmformula}) and (\ref{eq:logderiv}) is exponentially hard in $n_v$ and $n_h$; thus, classical approaches resort to approximations such as contrastive divergence~\cite{Hin02,SMH07,Tie08,SH09,Ben09}.
Unfortunately, contrastive divergence does not provide the gradient of any true objective function~\cite{ST10}, it
is known to lead to suboptimal solutions ~\cite{TH09,BD07,FI11}, it is not guaranteed to converge in the presence of certain regularization functions~\cite{ST10}, and it cannot be used directly to train a full Boltzmann machine.
We show that quantum computation provides a much better framework for deep learning and illustrate this by providing efficient alternatives to these methods that are elementary to analyze, accelerate the learning process and lead to better models for the training data.


\begin{figure}
\centering
\includegraphics[width=0.85\linewidth]{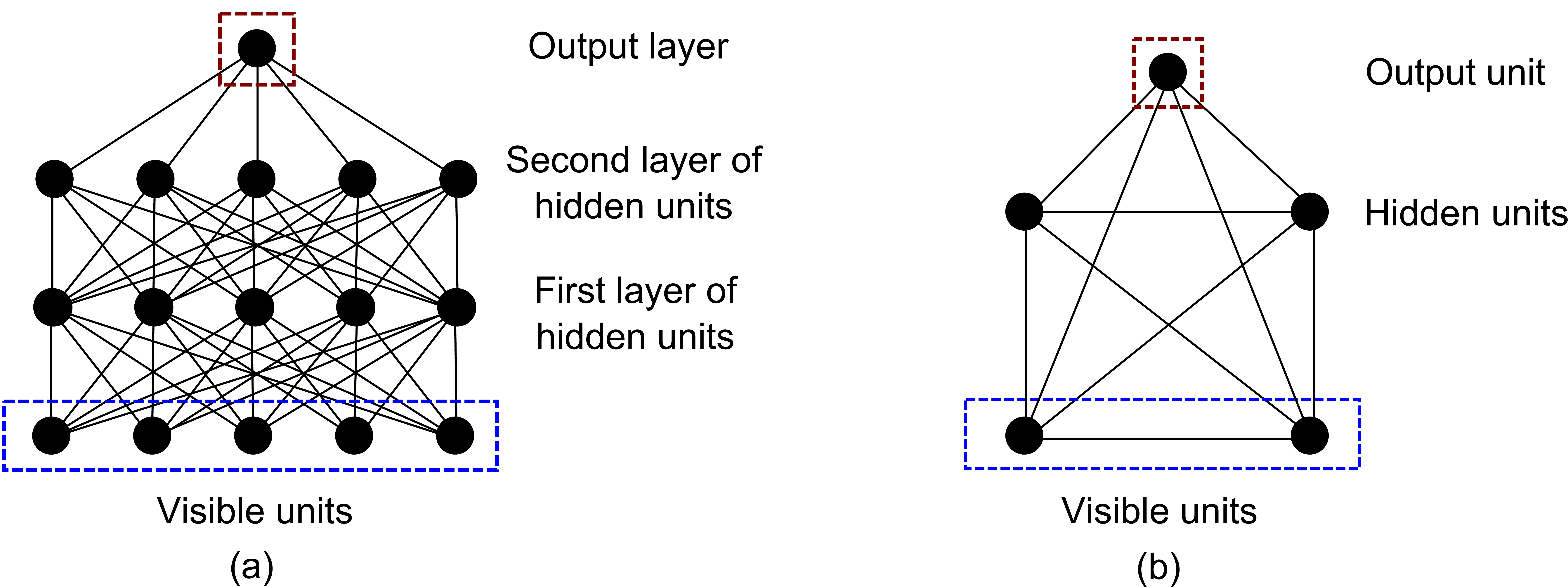}
\caption{A graphical representation of two types of Boltzmann machines.   (a) A $4$--layer deep restricted Boltzmann machine (dRBM) where each black circle represents a hidden or visible unit and each edge represents a non--zero weight for the corresponding interaction.  The output layer is often treated as a visible layer to provide a classification of the data input in the visible units at the bottom of the graph.  (b) An example of a $5$--unit full Boltzmann machine (BM).  The hidden and visible units no longer occupy distinct layers due to connections between units of the same type.}
\label{fig:BMdiag}
\end{figure}

%

\section*{GEQS Algorithm}
We propose two quantum algorithms: Gradient Estimation via Quantum Sampling (GEQS) and Gradient Estimation via Quantum Ampitude Estimation (GEQAE).
These algorithms prepare a coherent analog of the Gibbs state for Boltzmann machines and then draw samples from the
resultant state to compute the expectation values in~(\ref{eq:logderiv}).  Formal descriptions of the algorithms are given in the \sm.
Existing algorithms for preparing these states~\cite{LB97,TD98,PW09,DF11,ORR13} tend not to be efficient for machine learning applications or do not offer clear evidence of a quantum speedup.  The inefficiency of~\cite{PW09,ORR13} is a consequence of the uniform initial state having small overlap with the Gibbs state.
The complexities of these prior algorithms, along with our own, are given in~\tab{cost}

Our algorithms address this problem by using a non-uniform prior distribution for the probabilities of each configuration, which is motivated by the fact that we know a priori from the weights and the biases that certain configurations will be less likely than others.  We obtain this distribution by using a mean--field (MF) approximation to the configuration probabilities.  This approximation is classically efficient and typically provides a good approximation to the Gibbs states observed in practical machine learning problems~\cite{Jor99,WH02,Tie08}.  Our algorithms exploit this prior knowledge to refine the Gibbs state from copies of the MF state.  This allows the Gibbs distribution to be prepared efficiently and exactly if the two states are sufficiently close.

The MF approximation, $Q(v,h)$, is defined to be the product distribution that
minimizes the Kullback--Leibler divergence ${\rm KL}(Q||P)$.  The fact that it is a product distribution means that it can be efficiently computed and also
can be used to find a classically tractable estimate of the partition function $Z$:
$$Z_{Q}:=\sum_{v,h} Q(v,h) \log\left(\frac{e^{-E(v,h)}}{Q(v,h)}\right).$$
Here $Z_{Q}\le Z$ and equality is achieved if and only if ${\rm KL}(Q||P)=0$~\cite{Jor99}.  Here $Q(v,h)$ does not need to be the MF approximation.  The same formula also applies if $Q(v,h)$ is replaced by another efficient approximation, such as a structured mean--field theory calculation~\cite{Xin02}.

Let us assume that a constant $\kappa$ is known such that
\begin{equation}
 P(v,h)\le \frac{e^{-E(v,h)}}{Z_{Q}} \le \kappa\, Q(v,h),\label{eq:kappadef}
\end{equation}
and define the following ``normalized'' probability of a configuration as
\begin{equation}
\mathcal{P}(v,h) := \frac{e^{-E(v,h)}}{\kappa\, Z_{Q}\, Q(v,h)}.
\end{equation}
Note that
\begin{equation}
Q(v,h)\mathcal{P}(v,h) \propto P(v,h),
\end{equation}
which means that if the state
\begin{equation}
\sum_{v,h} \sqrt{Q(v,h)}|v\rangle |h\rangle,
\end{equation}
is prepared and each of the amplitudes are multiplied by $\sqrt{\mathcal{P}(v,h)}$ then the result will be proportional to the desired state.

The above process can be made operational by adding an additional quantum register to compute $\mathcal{P}(v,h)$ and using quantum superpostion to prepare the state
\begin{equation}
\sum_{v,h} \sqrt{Q(v,h)}|v\rangle |{h}\rangle |{\mathcal{P}(v,h)}\rangle\left(\sqrt{1-\mathcal{P}(v,h)}|{0}\rangle+\sqrt{\mathcal{P}(v,h)}|{1}\rangle\right).\label{eq:stateprep}
\end{equation}
The target Gibbs state is obtained if the right--most qubit is measured to be $1$.  Preparing~(\ref{eq:stateprep}) is efficient because $e^{-E(v,h)}$ and $Q(v,h)$ can be calculated in time that
is polynomial in the number of visible and hidden units.
The success probability of preparing the state in this manner is
\begin{equation}
P_{\rm success}  =\frac{Z}{\kappa \, Z_{Q} }\ge \frac{1}{\kappa}.\label{eq:Psucc}
\end{equation}
In practice, our algorithm uses quantum amplitude amplification~\cite{BHM+00} to quadratically boost the probability of success if (\ref{eq:Psucc}) is small.

The complexity of the algorithm is determined by the number of quantum operations needed in the gradient calculation.  Since the evaluation of the energy requires a number of operations that, up to logarithmic factors, scales linearly with the total number of edges in the model the combined cost of estimating the gradient is
\begin{equation}
\tilde{O}\left(N_{\rm train} E (\sqrt{\kappa} + \max_{x\in x_{\rm train}} \sqrt{\kappa_x})\right),
\end{equation}
here $\kappa_x$ is the value of $\kappa$ corresponding to the case where the visible units are constrained to be $x$.	The cost of estimating $Q(v,h)$ and $Z_{\rm MF}$ is $\tilde{O}(E)$ (see~\sm) and thus does not asymptotically contribute to the cost.
In contrast, the number of operations required to classically estimate the gradient using greedy layer--by--layer optimization~\cite{Ben09} scales as
\begin{equation}
\tilde{O}(N_{\rm train} \ell E),
\end{equation}
where $\ell$ is the number of layers in the dRBM and $E$ is the number of connections in the BM.  Assuming that $\kappa$ is a constant, the quantum sampling approach provides an asymptotic advantage for training deep networks.  We provide numerical evidence in the \sm showing that $\kappa$ can often be made constant by increasing $n_h$ and the regularization parameter $\lambda$.

The number of qubits required by our algorithm is minimal compared to existing quantum machine learning algorithms~\cite{ABG06,LMR13,RML13,QKS15}.  This is because the training data does not need to be stored in a quantum database, which would otherwise require $\tilde{O}(N_{\rm train})$ logical qubits~\cite{NC00,GLM08}.  Rather, if $\mathcal{P}(v,h)$ is computed with $\lceil\log(1/\mathcal{E})\rceil$ bits of precision and can be accessed as an oracle then only $$O(n_h + n_v + \log(1/\mathcal{E}))$$
logical qubits are needed for the GEQS algorithm. The number of qubits required will increase if $\mathcal{P}(v,h)$ is computed using reversible operations, but recent developments in quantum arithmetic can substantially reduce such costs~\cite{WR14}.

Furthermore, the exact value of $\kappa$ need not be known.  If a value of $\kappa$ is chosen that does not satisfy (\ref{eq:kappadef}) for all configurations then our algorithm will still be able to approximate the gradient  if $\mathcal{P}(v,h)$ is clipped to the interval $[0,1]$.  The algorithm can therefore always be made efficient, at the price of introducing errors in the resultant probability distribution, by holding $\kappa$ fixed as the size of the BM increases.  These errors emerge because the state preparation algorithm will under--estimate the relative probability of configurations that violate~(\ref{eq:kappadef}); however, if the sum of the probabilities of these violations is small then a simple continuity argument reveals that the fidelity of the approximate Gibbs state and the correct state is high.  In particular, if we define ``bad'' to be the set of configurations that violate~(\ref{eq:kappadef}) then the continuity argument shows that if $$\sum_{(v,h)\in {\rm bad}} P(v,h)\le \epsilon,$$ then the fidelity of the resultant state with the Gibbs state is at least $1-\epsilon$.  This is formalized in the \sm.

Our algorithms are not expected to be both exact and efficient~\emph{for all BMs}.
If they were then they could be used to learn ground--state energies of non--planar Ising models, implying that ${\rm NP} \subseteq {\rm BQP}$, which is widely believed to be false.  Therefore BMs exist for which our algorithm will fail to be efficient and exact, modulo complexity theoretic assumptions.  It is unknown how common these hard examples are in practice; however, they are unlikely to be commonplace because of the observed efficacy of the MF approximation for trained BMs~\cite{Jor99,WH02,Tie08,SH09} and because the weights used in trained models tend to be small.

\begin{table}[t!]
\begin{tabular}{|c|c|c|c|}
\hline
 & Operations & Qubits & Exact  \\
\hline
ML & $\tilde{O}(N_{\rm train}2^{n_v + n_h})$ & $0$ & Y \\
CD--k & $\tilde{O}(N_{\rm train} \ell Ek)$ & $0$ & N  \\
GEQS & $\tilde{O}(N_{\rm train} E(\sqrt{\kappa} + \max_x \sqrt{\kappa_x}))$ & $O(n_h+n_v+ \log(1/\mathcal{E}))$ & Y  \\
GEQAE & $\tilde{O}(\sqrt{N_{\rm train}} E^2(\sqrt{\kappa} + \max_x \sqrt{\kappa_x}))$ & $O(n_h+n_v+ \log(1/\mathcal{E}))$ & Y  \\
GEQAE (QRAM) & $\tilde{O}(\sqrt{N_{\rm train}} E^2(\sqrt{\kappa} + \max_x \sqrt{\kappa_x}))$ & $O(N_{\rm train}+n_h+n_v+ \log(1/\mathcal{E}))$ & Y  \\
\hline
\end{tabular}
\caption{The resource scalings of our quantum algorithms for a dRBM with $\ell$ layers, $E$ edges, and $N_{\rm train}$ training vectors.  An algorithm is exact if sampling is the only source of error.  GEQS and GEQAE are not exact if (\ref{eq:kappadef}) is violated.  We assume QRAM allows operations on different qubits to be executed simultaneously at unit cost. \label{tab:cost}}
\end{table}

\section*{GEQAE Algorithm}
New forms of training, such as our GEQAE algorithm, are possible in cases where the training data is provided via a quantum oracle, allowing access to the training data in superposition rather than sequentially.  The idea behind the GEQAE algorithm is to leverage the data superposition by amplitude estimation~\cite{BHM+00}, which leads to a quadratic reduction in the variance in the estimated gradients over the GEQS algorithm.  GEQAE consequently leads to substantial performance improvements for large training sets.  Also, allowing the training data to be accessed quantumly allows it to be pre--processed using quantum clustering and data processing algorithms~\cite{ABG06,ABG07,LMR13,RML13,QKS15}.

The quantum oracle used in GEQAE abstracts the access model for the training data.  The oracle can be thought of as a stand--in for a quantum database or an efficient quantum subroutine that generates the training data (such as a pre--trained quantum Boltzmann machine or quantum simulator).  As the training data must be directly (or indirectly) stored in the quantum computer, GEQAE typically requires more qubits than GEQS; however, this is mitigated by the fact that quantum superposition allows training over the entire set of training vectors in one step as opposed to learning over each training example sequentially.  This allows the gradient to be accurately estimated while accessing the training data at most $O(\sqrt{N_{\rm train}})$ times.

Let $U_O$ be a quantum oracle that for any index $i$ has the action
\begin{equation}
U_O|i\rangle |{y}\rangle := |i\rangle |{y}\oplus {x}_i\rangle,
\end{equation}
where $x_i$ is a training vector.  This oracle can be used to prepare the visible units in the state of the $i^{\rm th}$ training vector.
A single query to this oracle then suffices to prepare a uniform superposition over all of the training vectors, which then can be converted into
\begin{equation}
\frac{1}{\sqrt{N_{\rm train}}}\sum_{i,h} \sqrt{Q(X_i,h)}|i\rangle|x_i\rangle |{h}\rangle \left(\sqrt{1-\mathcal{P}(x_i,h)}|{0}\rangle+\sqrt{\mathcal{P}(x_i,h)}|{1}\rangle\right),\label{eq:stateprep2}
\end{equation}
 by repeating the state preparation
method given in (\ref{eq:stateprep}).

GEQAE computes expectations such as $\langle v_i h_j\rangle$ over the data and model by estimating (a) $P(1)$, the probability of measuring the right--most qubit in~(\ref{eq:stateprep2}) to be $1$ and (b) $P(11)$, the probability of measuring $v_i = h_j=1$ and the right--most qubit in~(\ref{eq:stateprep2}) to be $1$.  It then follows that
\begin{equation}
\langle v_i h_j \rangle  = \frac{P(11)}{P(1)}.
\end{equation}
These two probabilities can be estimated by sampling, but a more efficient method is to learn them using amplitude estimation~\cite{BHM+00} --- a quantum algorithm that uses phase estimation on Grover's algorithm to directly output these probabilities in a qubit string.  If we demand the sampling error to scale as $1/\sqrt{N_{\rm train}}$ (in rough analogy to the previous case) then the query complexity of GEQAE is
\begin{equation}
\tilde{O}\left(\sqrt{N_{\rm train}} E({\kappa+\max_x \kappa_x}) \right).\label{eq:queryAE}
\end{equation}
Each energy calculation requires $\tilde{O}(E)$ arithmetic operations, therefore (\ref{eq:queryAE}) gives that the number of non--query operations scales as
\begin{equation}
\tilde{O}\left(\sqrt{N_{\rm train}} E^2({\kappa+\max_x \kappa_x}) \right).
\end{equation}

If the success probability is known to within a constant factor then amplitude amplification~\cite{BHM+00} can be used to boost the success probability prior to estimating it.  The original success probability is then computed from the amplified probability.    This reduces the query complexity of GEQAE to
\begin{equation}
\tilde{O}\left(\sqrt{N_{\rm train}} E(\sqrt{\kappa}+\max_x \sqrt{\kappa_x}) \right).\label{eq:kappaAEAA}
\end{equation}
  GEQAE is therefore preferable to GEQS if $\sqrt{N_{\rm train}}\gg E$.



\section*{Parallelizing Algorithms}
Greedy CD--$k$ training is embarassingly parallel, meaning that almost all parts of the algorithm can be distributed over parallel processing nodes.  However, the $k$ rounds of sampling used to train each layer in CD--$k$ cannot be easily parallelized.  This means that simple but easily parallelizable models (such as GMMs) can be preferable in some circumstances~\cite{HDY+12}.  In contrast, GEQS and GEQAE can leverage the parallelism that is anticipated in a fault--tolerant quantum computer to train a dRBM much more effectively.  To see this, note that the energy  is the sum of the energies of each layer, which can be computed in depth $\log(M)=O(\max(n_v,n_h) \log(\max(n_v,n_h)))$ (see Figure~\ref{fig:BMdiag}) and summed in depth $O(\log(\ell))$.  The $O(\sqrt{\kappa+\max_x \kappa_x})$ MF state preparations can be executed simultaneously and the correct sample(s) located via a log--depth calculation.  The depth of GEQS is therefore
\begin{equation}
O\left(\log([\kappa+\max_x \kappa_x] M\ell N_{\rm train})\right).
\end{equation}
Since each of the derivatives output by GEQAE can be computed independently, the depth of GEQAE is
\begin{equation}
O\left(\sqrt{N_{\rm train}[\kappa+\max_x \kappa_x]}\log(M\ell) \right).
\end{equation}
The depth can be reduced, at the price of increased circuit size, by dividing the training set into mini--batches and averaging the resultant derivatives.

Training using $k$--step contrastive divergence (CD-$k$) requires depth
\begin{equation}
O(k\ell^2 \log(MN_{\rm train})).
\end{equation}
The $O(\ell^2)$ scaling arises because CD-$k$ is a feed--forward algorithm, whereas GEQS and GEQAE are not.

\section*{Numerical Results}
We address the following questions regarding the behavior of our algorithms:
\begin{enumerate}
\item What are typical values of $\kappa$?
\item How do models trained using CD-$1$ differ from those trained with GEQS and GEQAE?
\item Do full BMs yield substantially better models than dRBMs?
\end{enumerate}

To answer these questions, $P(v,h)$ and $O_{\rm ML}$ need to be computed classically, requiring time that grows exponentially with $\max\{n_v,n_h\}$.
Computational restrictions therefore severely limit the size of the models that we can study through numerical experiments.
In practice, we are computationally limited to models with at most $20$ units.
We train the following dRBMs with $\ell \in \{2,3\}$ layers, $n_h \in \{2,\ldots,8\}$ hidden units, and $n_v\in \{4,\ldots, 12 \}$ visible units.

\begin{figure}
\includegraphics[width=0.9\linewidth]{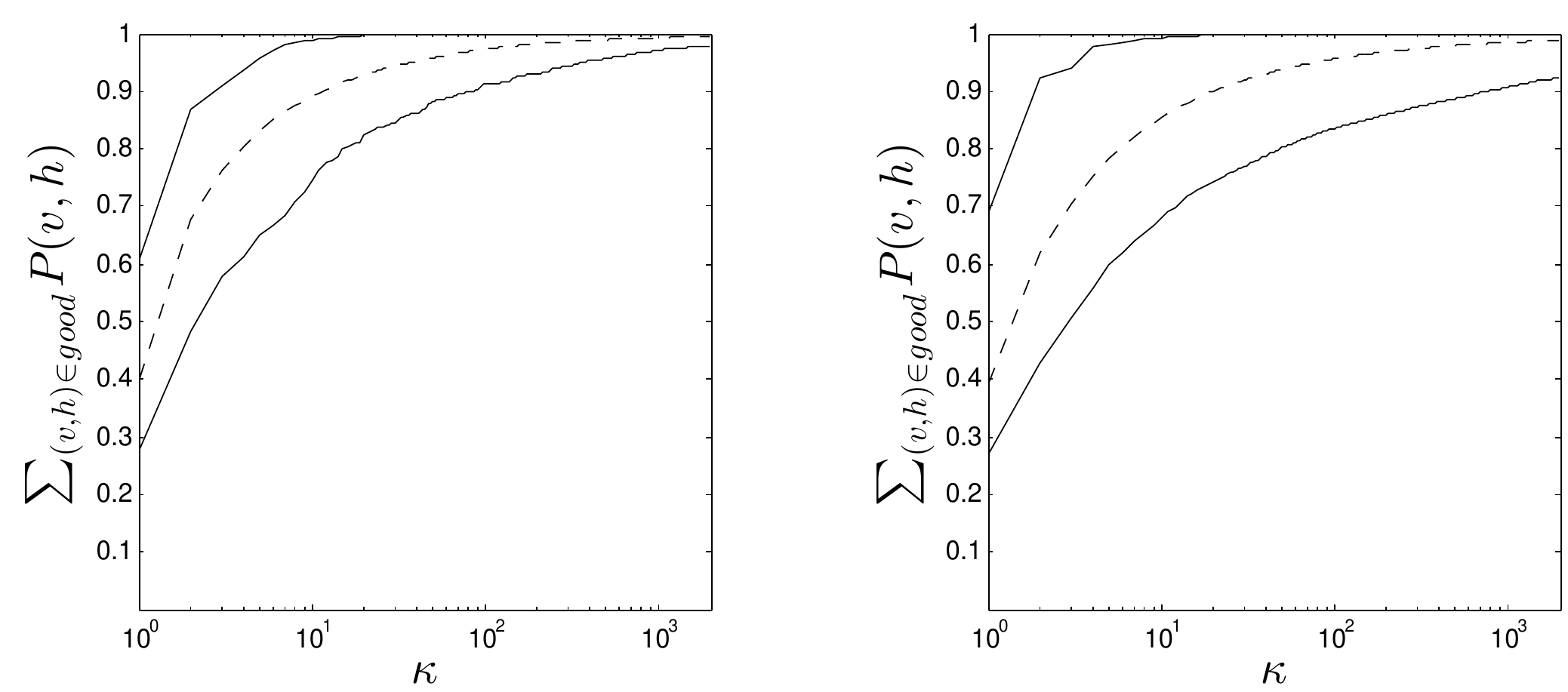}
\caption{Probability mass such that $\mathcal{P}(v,h)\le 1$ vs $\kappa$ for RBMs trained on~(\ref{eq:4datavectors}) with $n_h=8$ and $n_v =6$ (left) and $n_v=12$ (right).  Dashed lines give the mean value; solid lines give a $95\%$ confidence interval.\label{fig:kappacomb}}
\end{figure}

Large-scale traditional data sets for benchmarking machine learning, such as MNIST~\cite{lecun1998mnist}, are impractical here due to computational limitations.  Consequently, we focus on synthetic training data consisting of four distinct functions:
\begin{eqnarray}
{[x_1]_j} &=& 1 \mbox{ if } j \le n_v/2  \mbox{ else } 0 \nonumber\\
{[x_2]_j} & = & j \mbox{ mod } 2\label{eq:4datavectors},
\end{eqnarray}
as well as their bitwise negations.
We add Bernoulli noise $\mathcal{N}\in [0,0.5]$ to each of the bits in the bit string to increase the size of the training sets.  In particular, we take each of the four patterns in~(\ref{eq:4datavectors}) and flip each bit with probability $\mathcal{N}$.  We use $10,000$ training examples in each of our numerical experiments; each vector contains $4,\ldots,12$ binary features.
Our task is to infer a generative model for these four vectors.
We provide numerical experiments on sub--sampled MNIST digit data in the \sm.  The results are qualitatively similar.

\begin{figure}[t!]
\centering
\includegraphics[width=0.5\linewidth]{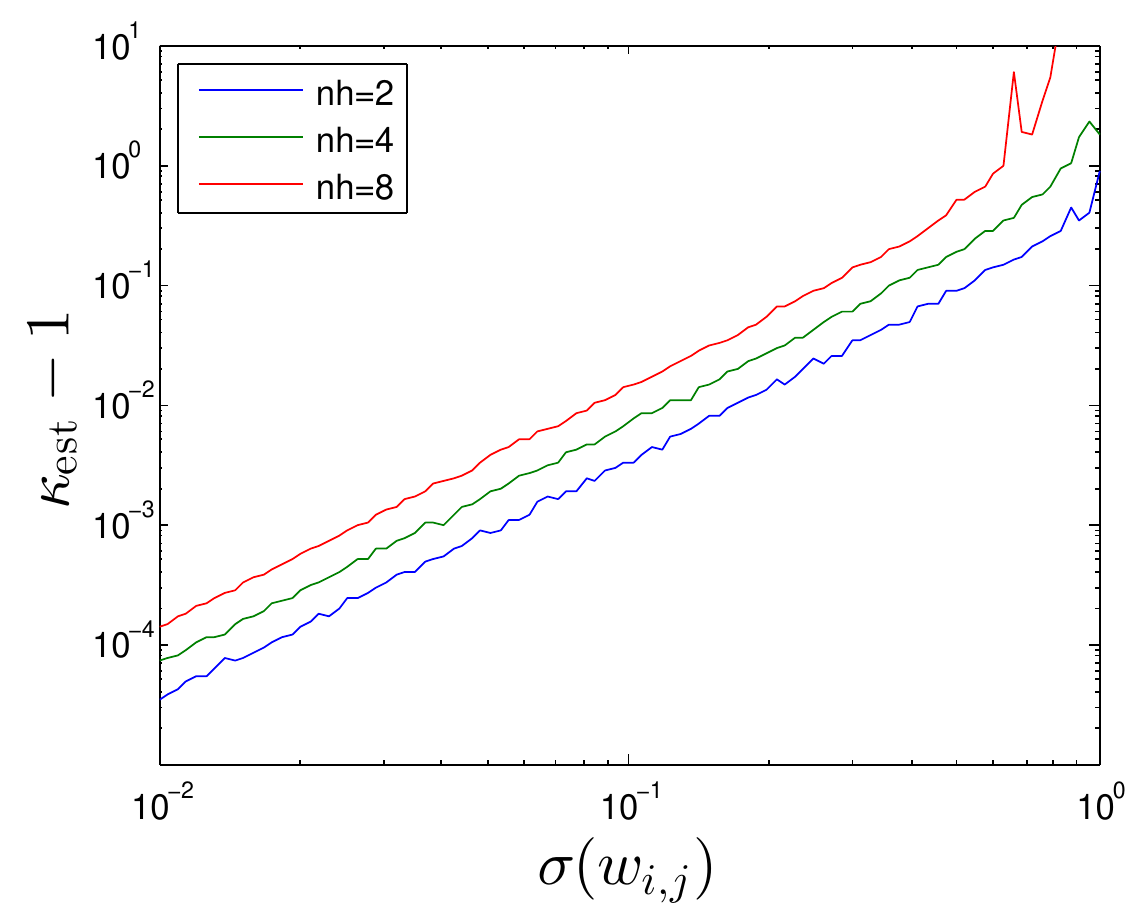}
\caption{Scaling of estimated $\kappa$ as a function of the standard deviation of weights in synthetic RBMs with $n_v=4$ and Gaussian weights with zero mean and variance $\sigma^2(w_{i,j})$.  The biases are set to be drawn from a Gaussian with zero mean and unit variance. \label{fig:kappa_scale}}
\end{figure}

Figure~\ref{fig:kappacomb} shows that doubling the number of visible units does not substantially increase $\kappa$  for this data set (with $\mathcal{N}=0$), despite the fact that their Hilbert space dimensions differ by a factor of $2^6$.  This illustrates that $\kappa$ primarily depends  on  the quality of the MF approximation, rather than $n_v$ and $n_h$.  Similar behavior is observed for full BMs, as shown in the \sm.  Furthermore, $\kappa\approx 1000$ typically results in a close approximation to the true Gibbs state. Although $\kappa=1000$ is not excessive, we introduce ``hedging strategies'' in the \sm that can reduce $\kappa$ to roughly $50$.


We further examine the scaling of $\kappa$ for random (untrained) RBMs via
\begin{equation}
\kappa_{\rm est} = \sum_{v,h} P^2(v,h)/Q(v,h).
\end{equation}
Figure~\ref{fig:kappa_scale} shows that for small random RBMs,  $\kappa -1\in O(\sigma^2(w_{i,j})E)$ for $\sigma^2( w_{i,j}) E\ll 1$.

This leads to the second issue: determining the distribution of weights for actual Boltzmann machines.  Figure~\ref{fig:sigma2_simple} shows that for large RBMs trained using contrastive divergence have weights that tend to rapidly shrink as $n_h$ is increased.  For $\mathcal{N}=0$, the empirical scaling is $\sigma^2 \in O(E^{-1})$ which suggests that $\kappa-1$ will not diverge as $n_h$ grows.  Although taking $\mathcal{N}=0.2$ reduces $\sigma^2$ considerably, the scaling is also reduced.  This may be a result of regularization having different effects for the two training sets.  In either case, these results coupled with those of Figure~\ref{fig:kappa_scale} suggest that $\kappa$ should be manageable for large networks.

\begin{figure}[t!]
\centering
\includegraphics[width=0.5\linewidth]{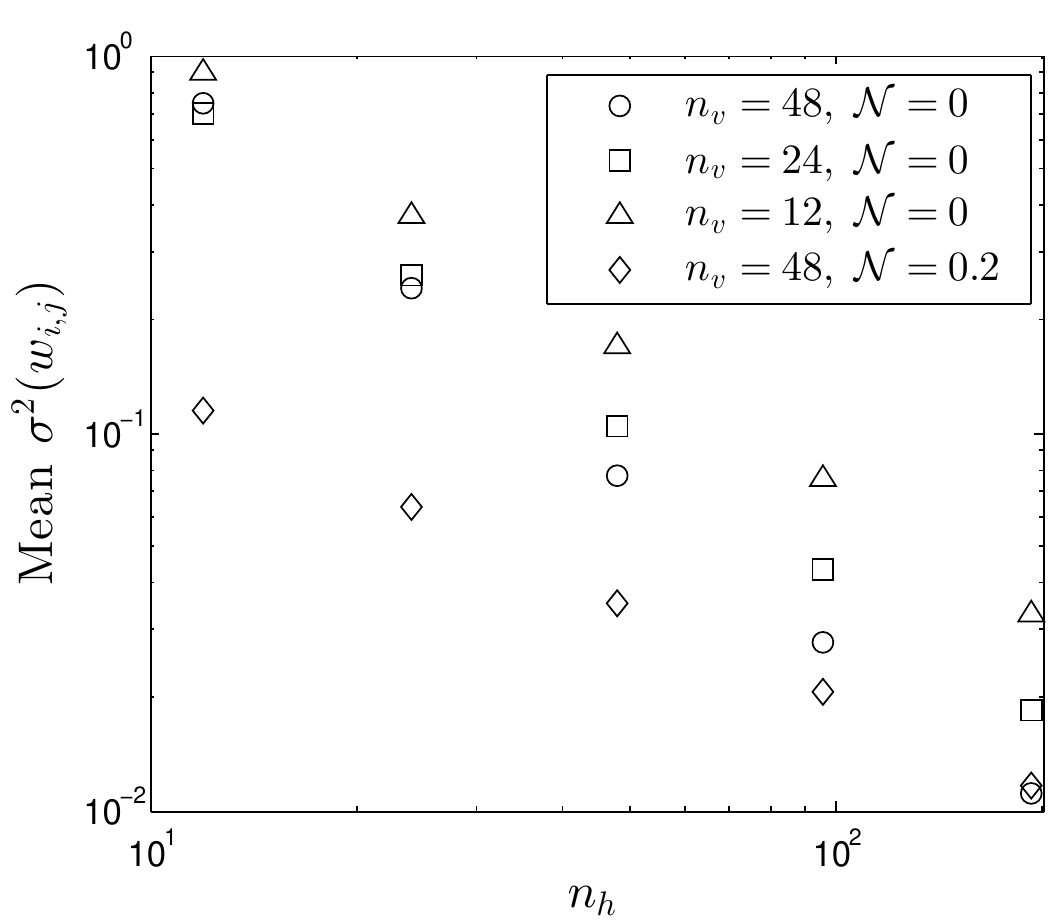}
\caption{Standard deviation of weights for large RBMs trained on~(\ref{eq:4datavectors}) using contrastive divergence with $\lambda=0.01$. \label{fig:sigma2_simple}}
\end{figure}


We assess the benefits of GEQS and GEQAE by comparing the average values of $O_{\rm ML}$ found under contrastive divergence and under our quantum algorithm, for dRBMs.  The difference between the optima found is significant for even small RBMs; the differences for deep networks can be on the order of $10\%$.  The data in Table~\ref{tab:3layer} shows that ML-training leads to substantial improvements in the quality of the resultant models.  We also observe that contrastive divergence can outperform gradient descent on the ML objective in highly constrained cases.  This is because the stochastic nature of the contrastive divergence approximation makes it less sensitive to local minima.
\begin{table}[tp]
\centering
\begin{tabular}{|c|c|c|c|c|c|}
\hline
$n_v$ & $n_{h1}$ & $n_{h2}$ & CD & ML&$\%$ Improvement \\
\hline
$6$ & $2$ &$2$ &$-2.7623$ & $-2.7125$&$1.80$\\
$6$ & $4$ &$4$ &$-2.4585$ & $-2.3541$&$4.25$\\
$6$ & $6$ &$6$ &$-2.4180$ & $-2.1968$&$9.15$\\
\hline
$8$ & $2$ &$2$ &$-2.8503$ & $-3.5125$&$-23.23$\\
$8$ & $4$ &$4$ &$-2.8503$ & $-2.6505$&$7.01$\\
$8$ & $6$ &$4$ &$-2.7656$ & $-2.4204$&$12.5$\\
\hline
$10$&$2$&$2$&$-3.8267$ &$-4.0625$&$-6.16$\\
$10$&$4$&$4$&$-3.3329$ &$-2.9537$&$11.38$\\
$10$&$6$&$4$&$-2.9997$ &$-2.5978$&$13.40$\\
\hline
\end{tabular}
\caption{Mean values of $O_{\rm ML}$ found by greedy training using contrastive divergence and gradient descent optimization of $O_{\rm ML}$ for $3$--layer dRBM with $\mathcal{N}=0$.\label{tab:3layer}}
\end{table}

The modeling capacity of a full BM can significantly outperform a dRBM in terms of the quality of the objective function.  In fact, we show in the \sm that a full Boltzmann machine with $n_v=6$ and $n_h=4$ can achieve $O_{\rm ML} \approx -1.84$.  dRBMs with comparable numbers of edges result in $O_{\rm ML}\approx -2.3$ (see Table~\ref{tab:3layer}), which is $25\%$ less than the full BM.  Since our quantum algorithms can efficiently train full BMs in addition to dRBMs, the quantum framework enables forms of machine learning that are not only richer than classically tractable methods but also potentially lead to improved models of the data.

\section*{Conclusions}
A fundamental result of our work is that training Boltzmann machines can be reduced to a problem of quantum state preparation.  This state preparation process notably does not require the use of contrastive divergence approximation or assumptions about the topology of the graphical model. We show that the quantum algorithms not only lead to significantly improved models of the data, but also provide a more elegant framework in which to think about training BMs.  This framework enables the wealth of knowledge developed in quantum information and condensed matter physics on preparing and approximating Gibbs states to be leveraged during training.

Our quantum deep learning framework enables the refining of MF approximations into states that are close (or equivalent to) the desired Gibbs state.  This state preparation method allows a BM to be trained using a number of operations that does not explicitly depend on the number of layers in a dRBM.  It also allows a quadratic reduction in the number of times the training data must be accessed and enables full Boltzmann machines to be trained.  Our algorithms can also be better parallelized over multiple quantum processors, addressing a major shortcoming of deep learning~\cite{HDY+12}.

While numerical results on small examples are encouraging in advance of having a scalable quantum computer, future experimental studies using quantum hardware will be needed to assess the generalization performance of our algorithm.  Given our algorithm's ability to provide better gradients than contrastive divergence, it is natural to expect that it will perform well in that setting by using the same methodologies currently used to train deep Boltzmann machines~\cite{Ben09}.  Regardless, the myriad advantages offered by quantum computing to deep learning not only suggests a significant near-term application of a quantum computer but also underscores the value of thinking about machine learning from a quantum perspective. 

\appendix

\section{Quantum algorithm for state preparation}\label{sec:QSTATE}

We begin by showing how quantum computers can draw unbiased samples from the Gibbs distribution, thereby allowing the probabilities $P(v,h)$ to be computed by sampling (or by quantum sampling). The idea behind our approach is to prepare a quantum distribution that approximates the ideal probability distribution over the model or data.  This approximate distribution is then refined using rejection sampling into a quantum distribution that is, to within numerical error, the target probability distribution~\cite{ORR13}.  
If we begin with a uniform prior over the amplitudes of the Gibbs state, then preparing the state via quantum rejection sampling is likely to be inefficient.  This is because the success probability depends on the ratio of the partition functions of the initial state and the Gibbs state~\cite{PW09}, which in practice is exponentially small for machine learning problems.  Instead, our algorithm uses a mean--field approximation, rather than a uniform prior, over the joint probabilities in the Gibbs state.  We show numerically that this extra information can be used to boost the probability of success to acceptable levels.
The required expectation values can then be found by sampling from the quantum distribution.  We show that the number of samples needed to achieve a fixed sampling error can be quadratically reduced by using a quantum algorithm known as amplitude estimation~\cite{BHM+00}.  

We first discuss the process by which the initial quantum distribution is refined into a quantum coherent Gibbs state (often called a coherent thermal state or CTS).  We then discuss how mean--field theory, or generalizations thereof, can be used to provide suitable initial states for the quantum computer to refine into the CTS.  We assume in the following that all units in the Boltzmann machine are binary valued. Other valued units, such as Gaussian units, can be approximated within this framework by forming a single unit out of a string of several qubits.

First, let us define the mean--field approximation to the joint probability distribution to be $Q(v,h)$.  For more details on the mean--field approximation, see~\sec{meanfield}.  We also use the mean--field distribution to compute a variational approximation to the partition functions needed for our algorithm.  These approximations can be efficiently calculated (because the probability distribution factorizes) and are defined below.
\begin{definition}
Let $Q$ be the mean--field approximation to the Gibbs distribution $P=e^{-E}/Z$ then
$$Z_{Q}:=\sum_{v,h} Q(v,h) \log\left(\frac{e^{-E(v,h)}}{Q(v,h)}\right).$$
Furthermore for any $x\in x_{\rm train}$ let $Q_x$ be the mean--field approximation to the Gibbs distribution found for a Boltzmann machine with the visible units clamped to $x$, then
$$Z_{x,Q}:=\sum_{h} Q_x(x,h) \log\left(\frac{e^{-E(x,h)}}{Q_x(x,h)}\right).$$ \label{def:Zmf}
\end{definition}

In order to use our quantum algorithm to prepare $P$ from $Q$ we need to know an upper bound, $\kappa$, on the ratio of the approximation $P(v,h)\approx e^{-E(v,h)}/Z_{Q}$ to $Q(v,h)$.  We formally define this below.
\begin{definition}
Let $\kappa>0$ be a constant that is promised to satisfy for all visible and hidden configurations $(v,h)$
\begin{equation}
\frac{e^{-E(v,h)}}{Z_{Q}} \le \kappa Q(v,h),\label{eq:kappacond}
\end{equation}
where $Z_{Q}$ is the approximation to the partition function given in~\defn{Zmf}.
\label{def:kappacond}
\end{definition}
We also define an analogous quantity appropriate for the case where the visible units are clamped to one of the training vectors.
\begin{definition}
Let $\kappa_x>0$ be a constant that is promised to satisfy for $x\in x_{\rm train}$ and all hidden configurations $h$
\begin{equation}
\frac{e^{-E(x,h)}}{Z_{{x},Q}} \le \kappa Q_x(x,h),\label{eq:kappacondx}
\end{equation}
where $Z_{{x},Q}$ is the approximation to the partition function given in~\defn{Zmf}.
\label{def:kappacondx}
\end{definition}

\begin{lemma}
Let $Q(v,h)$ be the mean--field probability distribution for a Boltzmann machine, then for all configurations of hidden and visible units we have 
\begin{equation*}
P(v,h)\le \frac{e^{-E(v,h)}}{Z_{Q}} \le \kappa Q(v,h).
\end{equation*}
\label{lem:kappacond}
\end{lemma}
\begin{proof}
The mean--field approximation can also be used to provide a lower bound for the log--partition function.  For example, Jensen's inequality shows that
\begin{align}
\log\left(Z\right) &= \log\left(\sum_{v,h} \frac{Q(v,h)e^{-E(v,h)}}{Q(v,h)}\right),\nonumber\\
& \ge \sum_{v,h} Q(v,h) \log\left(\frac{e^{-E(v,h)}}{Q(v,h)}\right)=\log(Z_{Q}).\label{eq:Zbd}
\end{align}
This shows that $Z_{Q} \le Z$ and hence 
\begin{equation}
P(v,h)\le  e^{-E(v,h)}/Z_{Q},\label{eq:kappaLeftBound}
\end{equation}
where $Z_{Q}$ is the approximation to $Z$ that arises from using the mean--field distribution.
The result then follows from~\eq{kappaLeftBound} and~\defn{kappacond}.
\end{proof}
The result of~\lem{kappacond} allows us to prove the following lemma, which gives the success probability for preparing the Gibbs state from the mean--field state.

\begin{lemma}
A coherent analog of the Gibbs state for a Boltzmann machine can be prepared with a probability of success of $\frac{Z}{\kappa Z_{Q}}$.  Similarly, the Gibbs state corresponding to the
visible units being clamped to a configuration $x$ can be prepared with success probability $\frac{Z_x}{\kappa_x Z_{x, {Q}}}$.  \label{lem:succ}
\end{lemma}
\begin{proof}
The first step in the algorithm is to compute the mean--field parameters $\mu_i$ and $\nu_j$ using~\eq{update}.  These parameters uniquely specify the mean--field distribution $Q$.
Next the mean--field parameters are used to approximate the partition functions $Z$ and $Z_x$.
These mean--field parameters are then used to prepare a coherent analog of $Q(v,h)$, denoted as $\ket{\psi_{Q}}$, by performing a series of single--qubit rotations:
\begin{equation}
\ket{\psi_{Q}} := \prod_i R_y(2\arcsin(\sqrt{\mu_i}))\ket{0} \prod_j R_y(2\arcsin(\sqrt{\nu_j}))\ket{0}=\sum_{v,h} \sqrt{Q(v,h)}\ket{v}\ket{h} .\label{eq:psimf}
\end{equation}
The remaining steps use rejection sampling to refine this crude approximation to $\sum_{v,h} \sqrt{P(v,h)} \ket{v}\ket{h} $.

For compactness we define
\begin{equation}
\mathcal{P}(v,h) := \frac{e^{-E(v,h)}}{\kappa Z_{Q} Q(v,h)}.
\end{equation}
$\mathcal{P}(v,h)$ can be computed efficiently from the mean--field parameters and so an efficient quantum algorithm (quantum circuit) also exists to compute $\mathcal{P}(v,h)$.  \lem{kappacond} also guarantees that $0\le \mathcal{P}(v,h)\le 1$.

Since quantum operations (with the exception of measurement) are linear, if we apply the algorithm to a state $\sum_v\sum_h \sqrt{Q(v,h)}\ket{v}\ket{h}\ket{0}$ we obtain $\sum_v\sum_h \sqrt{Q(v,h)}\ket{v}\ket{h}\ket{\mathcal{P}(v,h)}$.  We then add an additional quantum bit, called an ancilla qubit, and perform a controlled rotation of the form $R_y(2 \sin^{-1}(\mathcal{P}(v,h)))$ on this qubit to enact the following transformation:
\begin{equation}
\sum_{v,h} \sqrt{Q(v,h)}\ket{v}\ket{h}\ket{\mathcal{P}(v,h)}\ket{0}\mapsto\sum_{v,h} \sqrt{Q(v,h)}\ket{v}\ket{h}\ket{\mathcal{P}(v,h)}\left(\sqrt{1-\mathcal{P}(v,h)}\ket{0}+\sqrt{\mathcal{P}(v,h)}\ket{1}\right).\label{eq:sampledist0}
\end{equation}
The quantum register that contains the qubit string $\mathcal{P}(v,h)$ is then reverted to the $\ket{0}$ state by applying the same operations used to prepare $\mathcal{P}(v,h)$ in reverse.  This process is possible because all quantum operations, save measurement, are reversible.
Since $\mathcal{P}(v,h)\in [0,1]$, then \eq{sampledist0} is a properly normalized quantum state and in turn its square is a valid probability distribution.

If the rightmost quantum bit in~\eq{sampledist0} is measured and a result of $1$ is obtained (recall that projective measurements always result in a unit vector) then the remainder of the state will be proportional to
\begin{equation}
\sum_{v,h} \sqrt{Q(v,h)\mathcal{P}(v,h)}  = \sqrt{\frac{Z}{\kappa Z_{Q}}}\sum_{v,h} \sqrt{\frac{e^{-E(v,h)}}{Z}}\ket{v}\ket{h}=\sqrt{\frac{Z}{\kappa Z_{Q}}}\sum_{v,h} \sqrt{P(v,h)} \ket{v}\ket{h},\label{eq:coeffeq}
\end{equation}
which is the desired state up to a normalizing factor.  The probability of measuring $1$ is the square of this constant of proportionality
\begin{equation}
P(1|\kappa, Z_{Q})=\frac{Z}{\kappa Z_{Q} }.
\end{equation}
Note that this is a valid probability because \lem{kappacond} gives that $\sum_{v,h} \kappa Z_{Q} Q(v,h) \ge \sum_{v,h} e^{-E(v,h)} \Rightarrow \kappa Z_Q \ge Z$.

Preparing a quantum state that can be used to estimate the expectation values over the data requires a slight modification to this algorithm.  First, for each $x\in x_{\rm train}$ needed for the expectation values, we replace $Q(v,h)$ with the constrained mean--field distribution $Q_x(x,h)$.  Then using this data the quantum state
\begin{equation}
\sum_h \sqrt{Q_x(x,h)} \ket{x} \ket{h},
\end{equation}
can be prepared.  We then follow the exact same protocol using $Q_x$ in place of $Q$, $Z_x$ in place of $Z$, and $Z_{x,{Q}}$ in place of $Z_{Q}$.  The success probability of this algorithm is 
\begin{equation}
P(1|\kappa, Z_{x,{Q}})=\frac{Z_{x}}{\kappa_x Z_{x,{Q}} }.
\end{equation}  
\end{proof}

\begin{algorithm}[t!]
\rule{\linewidth}{1pt}
\begin{algorithmic}
\Require Model weights $w$, visible biases $b$, hidden biases $d$, edge set $E$ and $\kappa$.
\Ensure Quantum state that can be measured to obtain the correct Gibbs state for a (deep) Boltzmann machine.
\vskip0.2em
\hrule
\vskip0.2em
\Function{qGenModelState}{$w,b,d,E,\kappa$}
\State Compute vectors of mean-field parameters $\mu$ and $\nu$ from $w$, $b$ and $d$.
\State Compute the mean-field partition function $Z_{Q}$.
\State Prepare state $\sum_{v,h} \sqrt{Q(v,h)} \ket{v}\ket{h} := \left(\prod_{i=1}^{n_v} e^{-i \sqrt{\mu_i} Y} \ket{0} \right)\left(\prod_{j=1}^{n_h} e^{-i \sqrt{\nu_j} Y} \ket{0} \right)$
\State Add qubit register to store energy values and initialize to zero: $\sum_{v,h} \sqrt{Q(v,h)} \ket{v}\ket{h} \rightarrow \sum_{v,h} \sqrt{Q(v,h)} \ket{v}\ket{h}\ket{0}$
\For{$i=1:n_v$}
\State $\sum_{v,h} \sqrt{Q(v,h)} \ket{v}\ket{h}\ket{E(v,h)} \rightarrow \sum_{v,h} \sqrt{Q(v,h)} \ket{v}\ket{h}\ket{E(v,h)+v_ib_i+\ln\left(\mu_i^{v_i}(1-\mu_i)^{1-v_i} \right)}$.
\State \Comment{Here an energy penalty is included to handle the bias and the visible units contribution to $Q(v,h)^{-1}$.}
\EndFor
\For{$j=1:n_h$}
\State $\sum_{v,h} \sqrt{Q(v,h)} \ket{v}\ket{h}\ket{E(v,h)} \rightarrow \sum_{v,h} \sqrt{Q(v,h)} \ket{v}\ket{h}\ket{E(v,h)+h_jd_j+\ln\left(\nu_j^{h_j}(1-\nu_j)^{1-h_j} \right)}$.
\EndFor
\For{$(i,j)\in E$}
\State $\sum_{v,h} \sqrt{Q(v,h)} \ket{v}\ket{h}\ket{E(v,h)} \rightarrow \sum_{v,h} \sqrt{Q(v,h)} \ket{v}\ket{h}\ket{E(v,h)+v_ih_jw_{i,j} }$.
\EndFor
\State $\sum_{v,h} \sqrt{Q(v,h)} \ket{v}\ket{h}\ket{E(v,h)}\rightarrow\sum_{v,h} \sqrt{Q(v,h)} \ket{v}\ket{h}\ket{E(v,h)}\left(\sqrt{\frac{e^{-E(v,h)}}{Z_{Q} \kappa}}\ket{1} +\sqrt{1-\frac{e^{-E(v,h)}}{Z_{Q} \kappa}}\ket{0} \right).$
\EndFunction
\end{algorithmic}
\rule{\linewidth}{1pt}
\caption{\label{alg:model} Quantum algorithm for generating states that can be measured to estimate the expectation values over the model.}
\end{algorithm}

\begin{algorithm}
\rule{\linewidth}{1pt}
\begin{algorithmic}
\Require Model weights $w$, visible biases $b$, hidden biases $d$, edge set $E$ and $\kappa_x$, training vector $x$.
\Ensure Quantum state that can be measured to obtain the correct Gibbs state for a (deep) Boltzmann machine with the visible units clamped to $x$.
\vskip0.2em
\hrule
\vskip0.2em
\Function{qGenDataState}{$w,b,d,E,\kappa_x,x$}
\State Compute vectors of mean-field parameters $\nu$ from $w$, $b$ and $d$ with visible units clamped to $x$.
\State Compute the mean-field partition function $Z_{x,{Q}}$.
\For{$i=1:n_v$}
\State $\sum_{h} \sqrt{Q(x,h)} \ket{x}\ket{h}\ket{E(x,h)} \rightarrow \sum_{h} \sqrt{Q(x,h)} \ket{x}\ket{h}\ket{E(x,h)+x_ib_i}$.
\EndFor
\For{$j=1:n_h$}
\State $\sum_{h} \sqrt{Q(x,h)} \ket{x}\ket{h}\ket{E(x,h)} \rightarrow \sum_{h} \sqrt{Q(x,h)} \ket{x}\ket{h}\ket{E(x,h)+h_jd_j+\ln\left(\nu_j^{h_j}(1-\nu_j)^{1-h_j} \right)}$.
\EndFor
\For{$(i,j)\in E$}
\State $\sum_{h} \sqrt{Q(x,h)} \ket{x}\ket{h}\ket{E(x,h)} \rightarrow \sum_{h} \sqrt{Q(x,h)} \ket{x}\ket{h}\ket{E(x,h)+x_ih_jw_{i,j} }$.
\EndFor
\State $\sum_{h} \sqrt{Q(x,h)} \ket{x}\ket{h}\ket{E(x,h)}\rightarrow\sum_{h} \sqrt{Q(x,h)} \ket{x}\ket{h}\ket{E(x,h)}\left(\sqrt{\frac{e^{-E(x,h)}}{Z_{x,{Q}} \kappa_x}}\ket{1} +\sqrt{1-\frac{e^{-E(x,h)}}{Z_{x,{Q}} \kappa_x}}\ket{0} \right).$
\EndFunction
\end{algorithmic}
\rule{\linewidth}{1pt}
\caption{\label{alg:data} Quantum algorithm for generating states that can be measured to estimate the expectation value over the data.}
\end{algorithm}

The approach to the state preparation problem used in~\lem{succ} is similar to that of~\cite{PW09}, with the exception that we use a mean-field approximation rather than the infinite temperature Gibbs state as our initial state.  This choice of initial state is important because the success probability of the state preparation process depends on the distance between the initial state and the target state. For machine learning applications, the inner product between the Gibbs state and the infinite temperature Gibbs state is often exponentially small; whereas we find in~\sec{kappa} that the mean--field and the Gibbs states typically have large overlaps.  

The following lemma is a more general version of~\lem{succ} that shows that if a insufficiently large value of $\kappa$ is used then the state preparation algorithm can still be employed, but at the price of reduced fidelity with the ideal coherent Gibbs state.

\begin{lemma}
If we relax the assumptions of~\lem{succ} such that $\kappa Q(v,h)\ge e^{-E(v,h)}/Z_{Q}$ for all $(v,h)\in \good$ and $\kappa Q(v,h) < e^{-E(v,h)}/Z_{Q}$ for all $j\in \bad$ and $\sum_{(v,h)\in \bad} \left(e^{-E(v,h)}-Z_{Q} \kappa Q(v,h)\right) \le \epsilon Z$, then a state can be prepared that has fidelity at least $1-\epsilon$ with the target Gibbs state with probability at least $Z(1-\epsilon)/(\kappa Z_{Q})$.\label{lem:kappa}
\end{lemma}
\begin{proof}
Let our protocol be that used in~\lem{succ} with the modification that the rotation is only applied if $e^{-E(v,h)}/Z_{Q} \kappa \le 1$.
This means that prior to the measurement of the register that projects the state onto the success or failure branch, the state is
\begin{align}
&\sum_{(v,h)\in{\rm good}} \sqrt{Q(v,h)}\ket{v}\ket{h}\left(\sqrt{\frac{{e^{-E(v,h)}}}{Z_{Q} \kappa Q(v,h)}}\ket{1} + \sqrt{1-\frac{{e^{-E(v,h)}}}{Z_{Q} \kappa Q(v,h)}}\ket{0} \right)+ \sum_{(v,h)\in{\rm bad}} \sqrt{Q(v,h)}\ket{v}\ket{h}\ket{1}
\end{align}
The probability of successfully preparing the approximation to the state is then
\begin{equation}
\sum_{(v,h)\in {\rm good}} \frac{e^{-E(v,h)}}{\kappa Z_{Q}}+\sum_{(v,h)\in {\rm bad}} Q(v,h) =\frac{Z- (\sum_{(v,h)\in{\rm bad}}e^{-E(v,h)}-\sum_{(v,h)\in {\rm bad}} \kappa Z_{Q} Q(v,h))}{ \kappa Z_{Q}}\ge \frac{Z(1-\epsilon)}{\kappa Z_{Q}}.
\end{equation}

The fidelity of the resultant state with the ideal state $\sum_{v,h} \sqrt{e^{-E(v,h)}/Z}\ket{v}\ket{h}$ is
\begin{equation}
\frac{\sum_{(v,h)\in {\rm good}}e^{-E(v,h)}+\sum_{(v,h)\in {\rm bad}} \sqrt{Q(v,h)Z_{Q} \kappa e^{-E(v,h)}}}{\sqrt{Z(\sum_{(v,h)\in {\rm good}} e^{-E(v,h)} + \sum_{(v,h)\in {\rm bad}} Z_{Q} \kappa Q(v,h))}}\ge \frac{\sum_{(v,h)\in {\rm good}}e^{-E(v,h)}+\sum_{(v,h)\in {\rm bad}} {Z_{Q} \kappa Q(v,h)}}{\sqrt{Z(\sum_{(v,h)\in {\rm good}} e^{-E(v,h)} + \sum_{(v,h)\in {\rm bad}} Z_{Q} \kappa Q(v,h))}},\label{eq:trick}
\end{equation}
since $Q(v,h) Z_{Q} \kappa \le e^{-E(v,h)}$ for all $(v,h)\in {\rm bad}$.  Now using the same trick employed in~\eq{trick} and the assumption that $\sum_{(v,h)\in \bad} \left(e^{-E(v,h)}-Z_{Q} \kappa Q(v,h)\right) \le \epsilon Z$, we have that the fidelity is bounded below by
\begin{equation}
\frac{Z(1-\epsilon)}{Z\sqrt{1-\epsilon}} =\sqrt{1-\epsilon} \ge 1-\epsilon.
\end{equation}
\end{proof}

The corresonding algorithms are outlined in \alg{model} and \alg{data}, for preparing the state required to compute the model expectation and the data expectation, respectively.

\section{Gradient calculation by sampling}\label{sec:QSAMP}
Our first algorithm for estimating the gradients of $O_{\rm ML}$ involves preparing the Gibbs state from the mean--field state and then drawing samples from the resultant distribution in order to estimate the expectation values required in the expression for the gradient.  We refer to this algorithm as GEQS (Gradient Estimation via Quantum Sampling) in the main body. We also optimize GEQS algorithm by utilizing a quantum algorithm known as amplitude amplification~\cite{BHM+00} (a generalization of Grover's search algorithm~\cite{Gro96}) which quadratically reduces the mean number of repetitions needed to draw a sample from the Gibbs distribution using the approach in~\lem{succ} or~\lem{kappa}.  

t is important to see that the distributions that this algorithm prepares are not directly related to the mean--field distribution.  The mean field distribution is chosen because it is an efficiently computable distribution that is close to the true Gibbs distribution and thereby gives a shortcut to preparing the state.  Alternative choices, such as the uniform distribution, will ideally result in the same final distribution but may require many more operations than would be required if the mean--field approximation were used as a starting point.
We state the performance of the GEQS algorithm in the following theorem.

\begin{theorem}
There exists a quantum algorithm that can estimate the gradient of $O_{\rm ML}$ using $N_{\rm train}$ samples for a Boltzmann machine on a connected graph with $E$ edges.  The mean number of quantum operations required by algorithm to compute the gradient is
$$\tilde{O}\left(N_{\rm train}E\left(\sqrt{\kappa} +\sqrt{\max_v \kappa_v}\right)\right),$$
where $\kappa_v$ is the value of $\kappa$ that corresponds to the Gibbs distribution when the visible units are clamped to $v$ and $f\in\tilde{O}(g)$ implies $f\in O(g)$ up to polylogarithmic factors.\label{thm:sampalg}
\end{theorem}
\begin{proof}
We use \alg{deriv} to compute the required gradients.  It is straightforward to see from~\lem{succ} that \alg{deriv} draws $N_{\rm train}$ samples from the Boltzmann machine and then estimates the expectation values needed to compute the gradient of the log--likelihood by drawing samples from these states.  The subroutines that generate these states, \texttt{qGenModelState} and \texttt{qGenDataState}, given in \alg{model} and \alg{data}, represent the only quantum processing in this algorithm.  The number of times the subroutines must be called on average before a success is observed is given by the mean of a geometric distribution with success probability given by~\lem{succ} that is at least
\begin{equation}
\min\left\{\frac{Z}{\kappa Z_{Q}},\min_x \frac{Z_x}{\kappa_x Z_{x,Q}}\right\}.
\end{equation}
\lem{kappacond} gives us that $Z>Z_{Q}$ and hence the probability of success satisfies
\begin{equation}
\min\left\{\frac{Z}{\kappa Z_{Q}},\min_v \frac{Z_x}{\kappa_x Z_{x,Q}}\right\}\ge\frac{1}{{\kappa + \max_v \kappa_v}}.\label{eq:succp}
\end{equation}
Normally, \eq{succp} implies that preparation of the Gibbs state would require $O(\kappa +\max_v \kappa_v)$ calls to~\alg{model} and~\alg{data}  on average, but the quantum amplitude amplification algorithm~\cite{BHM+00} reduces the average number of repetitions needed before a success is obtained to $O(\sqrt{\kappa+\max_v \kappa_v})$.
\alg{deriv} therefore requires an average number of calls to~\texttt{qGenModelState} and \texttt{qGenDataState} that scale as $O(N_{\rm train}\sqrt{\kappa +\max_v \kappa_v})$.

\alg{model} and~\alg{data} require preparing the mean-field state, computing the energy of a configuration $(v,h)$, and performing a controlled rotation.  Assuming that the graph is connected, the number of hidden and visible units are $O(E)$.  Since the cost of synthesizing single qubit rotations to within error $\epsilon$ is $O(\log(E/\epsilon))$~\cite{KMM+13,RS14,BRS14} and the cost of computing the energy is $O(E~{\rm polylog}(E/\epsilon))$ it follows that the cost of these algorithms is $\tilde{O}(E)$.  Thus the expected cost of~\alg{deriv} is $\tilde{O}(N_{\rm train}E\sqrt{\kappa +\max_v \kappa_v})\in \tilde{O}(N_{\rm train}E(\sqrt{\kappa} +\sqrt{\max_v \kappa_v}))$ as claimed.
\end{proof}

In contrast, the number of operations and queries to $U_O$ required to estimate the gradients using greedy layer--by--layer optimization scales as~\cite{Ben09}
\begin{equation}
\tilde{O}(N_{\rm train} \ell E),
\end{equation}
where $\ell$ is the number of layers in the deep Boltzmann machine.  Assuming that $\kappa$ is a constant, it follows that the quantum sampling approach provides an asymptotic advantage for training deep networks.  In practice, the two approaches are difficult to directly compare because they both optimize different objective functions and thus the qualities of the resultant trained models will differ.  It is reasonable to expect, however, that the quantum approach will tend to find superior models because it optimizes the maximum-likelihood objective function up to sampling error due to taking finite $N_{\rm train}$.  

Note that~\alg{deriv} has an important advantage over many existing quantum machine learning algorithms~\cite{ABG06,LMR13,RML13,QKS15}: \emph{it does not require that the training vectors are stored in quantum memory}.  It requires only $n_h+n_v +1 + \lceil\log_2(1/\mathcal{E})\rceil$ qubits if a numerical precision of $\mathcal{E}$ is needed  in the evaluation of the $E(v,h) - \log(Q(v,h))$.  This means that a demonstration of this algorithm that would not be classically simulatable could be performed with fewer than $100$ qubits, assuming that $32$ bits of precision suffices for the energy.  In practice though, additional qubits will likely be required to implement the required arithmetic on a quantum computer.  Recent developments in quantum rotation synthesis could, however, be used to remove the requirement that the energy is explicitly stored as a qubit string~\cite{WR14}, which may substantially reduce the space requirements of this algorithm.  Below we consider the opposite case: the quantum computer can coherently access the database of training data via an oracle.  The algorithm requires more qubits (space), however it can quadratically reduce the number of samples required for learning in some settings.

\begin{algorithm}[t!]
\rule{\linewidth}{1pt}
\begin{algorithmic}
\Require Initial model weights $w$, visible biases $b$, hidden biases $d$, edge set $E$ and $\kappa$, a set of training vectors $x_{\rm train}$, a regularization term $\lambda$, and a learning rate $r$.
\Ensure Three arrays containing gradients of weights, hidden biases and visible biases: $\texttt{gradMLw},\texttt{gradMLb},\texttt{gradMLd}$.
\vskip0.2em
\hrule
\vskip0.2em

\For{$i=1:N_{\rm train}$}
\State ${\texttt{success}}\gets 0$
\While{$\texttt{success}=0$}
\State $\ket{\psi}\gets\texttt{qGenModelState}(w,b,d,E,\kappa)$
\State $\texttt{success}\gets$ result of measuring last qubit in $\ket{\psi}$
\EndWhile
\State $\texttt{modelVUnits}[i] \gets$ result of measuring visible qubit register in $\ket{\psi}$.
\State $\texttt{modelHUnits}[i] \gets$ result of measuring hidden unit register in $\ket{\psi}$ using amplitude amplification.
\State ${\texttt{success}}\gets 0$
\While{$\texttt{success}=0$}
\State $\ket{\psi}\gets\texttt{qGenDataState}(w,b,d,E,\kappa,x_{\rm train}[i])$.
\State $\texttt{success}\gets$ result of measuring last qubit in $\ket{\psi}$ using amplitude amplification.
\EndWhile
\State $\texttt{dataVUnits}[i] \gets$ result of measuring visible qubit register in $\ket{\psi}$.
\State $\texttt{dataHUnits}[i] \gets$ result of measuring hidden unit register in $\ket{\psi}$.
\EndFor
\For{each visible unit $i$ and hidden unit $j$}
\State $\texttt{gradMLw}[i,j] \gets r\left(\frac{1}{N_{\rm train}}\sum_{k=1}^{N_{\rm train}}\left(\texttt{dataVUnits}[k,i]\texttt{dataHUnits}[k,j]-\texttt{modelVUnits}[k,i]\texttt{modelHUnits}[k,j]\right)-\lambda w_{i,j}\right)$.
\State $\texttt{gradMLb}[i] \gets r\left(\frac{1}{N_{\rm train}}\sum_{k=1}^{N_{\rm train}}\left(\texttt{dataVUnits}[k,i]-\texttt{modelVUnits}[k,i]\right)\right)$.
\State $\texttt{gradMLd}[j] \gets r\left(\frac{1}{N_{\rm train}}\sum_{k=1}^{N_{\rm train}}\left(\texttt{dataHUnits}[k,j]-\texttt{modelHUnits}[k,j]\right)\right)$.
\EndFor
\end{algorithmic}
\rule{\linewidth}{1pt}
\caption{\label{alg:deriv} GEQS algorithm for estimating the gradient of $O_{\rm ML}$.}
\end{algorithm}

\section{Training via quantum amplitude estimation}\label{sec:QAE}
We now consider a different learning environment, one in which the user has access to the training data via a quantum oracle which could represent either an efficient
quantum algorithm that provides the training data (such as another Boltzmann machine used as a generative model) or a quantum database that stores the memory via a binary access tree~\cite{NC00,GLM08}, such as a quantum Random Access Memory (qRAM)~\cite{GLM08}.

If we denote the training set as $\{{x}_i| i=1,\ldots,N_{\rm train}\}$, then the oracle is defined as a unitary operation as follows:
\begin{definition}
 $U_O$ is a unitary operation that performs for any computational basis state $\ket{i}$ and any $y\in \mathbb{Z}_2^{n_v}$
\begin{equation*}
U_O\ket{i} \ket{{y}} := \ket{i}\ket{{y}\oplus {x}_i},
\end{equation*}
where $\{{x}_i| i=1,\ldots,N_{\rm train}\}$ is the training set and $x_i \in \mathbb{Z}_2^{n_v}$.
\label{def:UO}
\end{definition}
A single quantum access to $U_O$ is sufficient to prepare a uniform distribution over all the training data
\begin{equation}
U_O\left(\frac{1}{N_{\rm train}}\sum_{i=1}^{N_{\rm train}}  \ket{i}\ket{0}\right) = \frac{1}{N_{\rm train}}\sum_{i=1}^{N_{\rm train}}  \ket{i}\ket{x_i}.
\end{equation}
The state $\frac{1}{N_{\rm train}}\sum_{i=1}^{N_{\rm train}}  \ket{i}\ket{0}$ can be efficiently prepared using quantum techniques~\cite{QKS15} and so the entire procedure is efficient.

At first glance, the ability to prepare a superposition over all data in the training set seems to be a powerful resource.  However, a similar probability distribution can also be generated classically using one query by picking a random training vector.  More sophisticated approaches are needed if we wish to leverage such quantum superpositions of the training data.  
\alg{derivAE} utilizes such superpositions to provide advantages, under certain circumstances, for computing the gradient.  The performance of this algorithm is given in the following theorem.
\begin{theorem}
There exists a quantum algorithm that can compute $r\frac{\partial O_{\rm ML}}{\partial w_{ij}}$, $r\frac{\partial O_{\rm ML}}{\partial b_{i}}$ or $r\frac{\partial O_{\rm ML}}{\partial d_j}$ for any $(i,j)$ corresponding to visible/hidden unit pairs for a Boltzmann machine on a connected graph with $E$ edges to within error $\delta$ using an expected number of queries to $U_O$ that scales as
$$
\tilde{O}\left(\frac{\kappa+\max_v \kappa_v}{\delta} \right),
$$
and a number of quantum operations that scales as 
$$
\tilde{O}\left(\frac{E(\kappa+\max_v \kappa_v)}{\delta} \right),
$$
for constant learning rate $r$.\label{thm:AEtrain}\label{thm:AEtrain}
\end{theorem}

\alg{derivAE} requires the use of the amplitude estimation~\cite{BHM+00} algorithm, which provides a quadratic reduction in the number of samples needed to learn the probability of an event occurring, as stated in the following theorem.
\begin{theorem}[Brassard, H\o yer, Mosca and Tapp]
For any positive integer $L$, the amplitude estimation algorithm of takes as input a quantum algorithm
that does not use measurement and with success probability $a$ and outputs $\tilde{a}$ $(0 \le \tilde a \le 1)$ such that
$$
|\tilde{a}-a|\le \frac{\pi(\pi+1)}{L}
$$
with probability at least $8/\pi^2$.  
It uses exactly $L$ iterations of Grover's algorithm.  
If $a=0$ then $\tilde{a}=0$ with certainty, and if $a=1$ and $L$ is even, then $\tilde{a}=1$ with certainty.\label{thm:AE}
\end{theorem}
This result is central to the proof of~\thm{AEtrain} which we give below.

\begin{proofof}{\thm{AEtrain}}
\alg{derivAE} computes the derivative of $O_{\rm ML}$ with respect to the weights.  The algorithm can be trivially adapted to compute the derivatives with respect to the biases.  The first step in the algorithm prepares a uniform superposition of all training data and then applies $U_O$ to it.  The result of this is 
\begin{equation}
\frac{1}{\sqrt{N_{\rm train}}}\sum_{p=1}^{N_{\rm train}} \ket{p}\ket{x_p},\label{eq:superstate}
\end{equation}
as claimed.

Any quantum algorithm that does not use measurement is linear and hence applying \texttt{qGenDataState} (\alg{data}) to~\eq{superstate} yields
\begin{align}
&\frac{1}{\sqrt{N_{\rm train}}}\sum_{p=1}^{N_{\rm train}}\ket{p} \ket{x_p}\sum_{h} \sqrt{Q(x_p,h)}\ket{h}\ket{\mathcal{P}(x_p,h)}\left(\sqrt{1-\mathcal{P}(x_p,h)}\ket{0}+\sqrt{\mathcal{P}(x_p,h)}\ket{1}\right)\nonumber\\
&\qquad :=\frac{1}{\sqrt{N_{\rm train}}}\sum_{p=1}^{N_{\rm train}}\ket{p} \ket{x_p}\sum_{h} \sqrt{Q(x_p,h)}\ket{h}\ket{\mathcal{P}(x_p,h)}\ket{\chi(x_p,h)}. \label{eq:state2}
\end{align}
If we consider measuring $\chi=1$ to be success then~\thm{AE} gives us that $\tilde{O}((\kappa +\max_v \kappa_v)/\Delta)$ preparations of~\eq{state2} are needed to learn $P(\texttt{success})=P(\chi=1)$ to within relative error $\Delta/8$ with high probability.  This is because $P(\texttt{success}) \ge 1/(\kappa+\max_v \kappa_v)$.  Similarly, we can also consider success to be the event where the $i^{\rm th}$ visible unit is $1$ and the $j^{\rm th}$ hidden unit is one and a successful state preparation is measured.  This marking process is exactly the same as the previous case, but requires a Toffoli gate.  Thus $P(v_i=h_j=\chi=1)$ can be learned within relative error $\Delta/8$ using $\tilde{O}((\kappa+\max_v \kappa_v)/\Delta)$ preparations.
It then follows from the laws of conditional probability that
\begin{equation}
\langle v_i h_j\rangle_{\rm data}= P([x_p]_i=h_j=1|\chi=1)=\frac{P([x_p]_i=h_j=\chi=1)}{P(\chi=1)},\label{eq:quotient}
\end{equation}
can be calculated from these values as claimed.

In order to ensure that the total error in $\langle v_i h_j\rangle_{\rm data}$ is at most $\Delta$, we need to bound the error in the quotient in~\eq{quotient}.  It can be seen that for $\Delta<1/2$,
\begin{equation}
\left|\frac{P([x_i]_j=h_k=\chi=1)(1 \pm \Delta/8)}{P(\chi=1)(1\pm \Delta/8)}- \frac{P([x_i]_j=h_k=\chi=1)}{P(\chi=1)}\right|\le  \frac{\Delta P([x_i]_j=h_k=\chi=1)}{P(\chi=1)}\le \Delta.
\end{equation}
Therefore the algorithm gives $\langle v_i h_j\rangle_{\rm data}$ within error $\Delta$.

The exact same steps can be repeated using~\alg{model} instead of~\alg{data} as the state preparation subroutine used in amplitude estimation.  This allows us to compute $\langle v_i h_j\rangle_{\rm data}$ within error $\Delta$ using $\tilde{O}(1/\Delta)$ state preparations.  The triangle inequality shows that the maximum error incurred from approximating $\langle v_i h_j\rangle_{\rm data}-\langle v_i h_j\rangle_{\rm model}$ is at most $2\Delta$.  Therefore, since a learning rate of $r$ is used in~\alg{derivAE}, the overall error in the derivative is at most $2\Delta r$.  If we pick $\Delta = \delta/(2r)$ then we see that the overall algorithm requires $\tilde{O}(1/\delta)$ state preparations for constant $r$.

Each state preparation requires one query to $U_O$ and $\tilde{O}(E)$ operations assuming that the graph that underlies the Boltzmann machine is connected.  This means that the expected query complexity of the algorithm is $\tilde{O}((\kappa+\max_v \kappa_v)/\delta)$ and the number of circuit elements required is $\tilde{O}((\kappa+\max_v \kappa_v)E/\delta)$ as claimed.
\end{proofof}

\begin{algorithm}[t!]
\rule{\linewidth}{1pt}
\begin{algorithmic}
\Require Initial model weights $w$, visible biases $b$, hidden biases $d$, edge set $E$ and $\kappa$, a set of training vectors $x_{\rm train}$, a regularization term $\lambda$, $1/2 \ge \Delta>0$, a learning rate $r$, and a specification of edge $(i,j)$.
\Ensure $r\frac{\partial O_{\rm ML}}{\partial w_{ij}}$ calculated to within error $2r\Delta$.
\vskip0.2em
\hrule
\vskip0.2em
\State Call $U_O$ once to prepare state $\ket{\psi} \gets \frac{1}{\sqrt{N_{\rm train}}}\sum_{p \in x_{\rm train}} \ket{p} \ket{x_p}$.
\vskip0.2em
\State $\ket{\psi} \gets \texttt{qGenDataState}({w,b,d,E,\kappa,\ket{\psi}})$. \Comment{Apply~\alg{data} using a superposition over $x_p$ rather than a single value.}
\vskip0.2em
\State Use amplitude estimation on state preparation process for $\ket{\psi}$ to learn $P([x_p]_i = h_j = \texttt{success}=1)$ within error $\Delta/8$.
\vskip0.2em
\State Use amplitude estimation on state preparation process for $\ket{\psi}$ to learn $P(\texttt{success}=1)$ within error $\Delta/8$.
\vskip0.2em
\State $\left\langle v_i h_j \right\rangle_{\rm data} \gets \frac{P([x_p]_i = h_j = \texttt{success}=1)}{P(\texttt{success}=1)}$.
\vskip0.2em
\State Use amplitude estimation in exact same fashion on $\texttt{qGenModelState}({w,b,d,E,\kappa})$ to learn $\left\langle v_i h_j \right\rangle_{\rm data}$.
\vskip0.2em
\State $\frac{\partial O_{\rm ML}}{\partial w_{ij}} \gets r\left(\left\langle v_i h_j \right\rangle_{\rm data} -\left\langle v_i h_j \right\rangle_{\rm model} \right)$
\end{algorithmic}
\rule{\linewidth}{1pt}
\caption{\label{alg:derivAE} GEQAE algorithm for estimating the gradient of $O_{\rm ML}$.}
\end{algorithm}
There are two qualitative differences between this approach and that of~\alg{deriv}.  The first is that the algorithm provides detailed information about one component of the gradient, whereas the samples from~\alg{deriv} provide limited information about every direction.  It may seem reasonable that if amplitude estimation is used to learn $\langle v_i h_j \rangle$ then $\langle v_k h_\ell \rangle$ could be estimated by sampling from the remaining qubits.  The problem is that amplitude estimation creates biases in these measurements and so the correct way to use this evidence to update the user's confidence about the remaining components of the gradient is unclear.

In general, it is not straight forward to blindly use amplitude amplification to provide a quadratic improvement to the scaling with $\kappa$ because a randomized algorithm is typically employed in cases where the success probability is unknown.  Since the randomized algorithm uses measurements, it cannot be used in concert with amplitude estimation.  However, if an upper bound on the success probability is known then amplitude amplification can be used deterministically to lead to amplify the success probability for the system.  Amplitude estimation can then be employed on the amplified version of the original circuit and the inferred success probability can then be backtracked to find the success probability in absentia of the amplification step.  This process is explained in the following corollary.
\begin{corollary}
Assume that $P_u:1\ge P_u>\frac{Z}{\kappa Z_{Q}}$ is known where $P_u\in \Theta(\frac{Z}{\kappa Z_{Q}})$ then there exists a quantum algorithm that can compute $r\frac{\partial O_{\rm ML}}{\partial w_{ij}}$, $r\frac{\partial O_{\rm ML}}{\partial b_{i}}$ or $r\frac{\partial O_{\rm ML}}{\partial d_j}$ for any $(i,j)$ corresponding to visible/hidden unit pairs for a Boltzmann machine on a connected graph with $E$ edges to within error $\delta\le P_u$ using an expected number of queries to $U_O$ that scales as
$$
\tilde{O}\left(\frac{\sqrt{\kappa}+\max_v \sqrt{\kappa_v}}{\delta/P_u} \right),
$$
and a number of quantum operations that scales as 
$$
\tilde{O}\left(\frac{E(\sqrt{\kappa}+\max_v \sqrt{\kappa_v})}{\delta/P_u} \right),
$$
for constant learning rate $r$.\label{cor:AEtrain}
\end{corollary}
\begin{proof}
The idea behind this corollary is to apply $m$ iterations of Grover's search to boost the probability of success before estimating the amplitude and then use the resultant probability of success to work backwards to find $P(11):= P([x_p]_i = h_j =\chi=1)$ or $P(1):=P(\chi=1)$.

Let us focus on learning $P(1)$.  The case of learning $P(11)$ is identical.  Applying $m$ steps of amplitude amplification (without measuring) causes the success probability to become
\begin{equation}
P_s = \sin^2([2m+1] \sin^{-1} (\sqrt{P(1)}) ).
\end{equation}
This equation cannot be inverted to find $P_s$ unless $[2m+1] \sin^{-1} (P(1)) \le \pi/2$.  However if $m$ is chosen such that $[2m+1] \sin^{-1} (P_u) \le \pi/2$ then we are guaranteed that an inverse exists.  Under this assumption
\begin{equation}
P(1) = \sin^2\left(\frac{\sin^{-1}(\sqrt{ P_s}) }{2m+1}\right).\label{eq:P1}
\end{equation}

Now we need to show that small errors in estimating $P_s$ do not propagate to large errors in $P(1)$.  In general, large errors can propagate because the derivative of~\eq{P1} with respect to $P_s$ diverges at $P_s=1$.  Let us assume that $m$ is chosen such that $P_s\le 1/4$.  If such an $m$ does not exist, then the success probability is already $O(1)$ and hence the complexity is independent of $\kappa$, so we can safely assume this as it is a worst case assumption.  The derivative of $P(1)$ then exists and is

\begin{equation}
\frac{\partial P(1)}{\partial P_s} = \frac{\sin\left(\frac{2\sin^{-1}(\sqrt{P_s})}{2m+1} \right)}{2\sqrt{P_s}\sqrt{1-P_s}(2m+1)}.
\end{equation}
Using the fact that $\sin(x) \le x$ this equation becomes
\begin{equation}
\frac{\sin\left(\frac{2\sin^{-1}(\sqrt{P_s})}{2m+1} \right)}{2\sqrt{P_s}\sqrt{1-P_s}(2m+1)} \le \frac{\sin^{-1}(\sqrt{P_s})}{\sqrt{P_s}\sqrt{1-P_s}(2m+1)^2}.\label{eq:derivbd}
\end{equation}
Since the Taylor series of $\sin^{-1}(x)$ has only positive terms and $\sin^{-1}(x) /x =1 +O(x^2)$, $\sin^{-1}(\sqrt{P_s})/\sqrt{P_s} \ge 1$.  Ergo~\eq{derivbd} is a monotonically increasing function of $P_s$ on $(0,1)$.  Thus the extreme value theorem implies
\begin{equation}
\frac{\partial P(1)}{\partial P_s} \le \frac{2\pi}{3^{3/2} (2m+1)^2}.
\end{equation}
Taylor's remainder theorem then shows that if phase estimation is used with precision $\Delta_0:-P_s \le \Delta_0\le 1/4 -P_s$ then
\begin{equation}
\left| \frac{\sin\left(\frac{2\sin^{-1}(\sqrt{P_s})}{2m+1} \right)}{2\sqrt{P_s}\sqrt{1-P_s}(2m+1)}  - \frac{\sin\left(\frac{2\sin^{-1}(\sqrt{P_s+\Delta_0})}{2m+1} \right)}{2\sqrt{P_s+\Delta_0}\sqrt{1-P_s-\Delta_0}(2m+1)}\right| \le \frac{2\pi\Delta_0}{3^{3/2} (2m+1)^2}.  
\end{equation}
Thus the resultant error is $O(\Delta_0 / m^2)$.  Hence if overall error $\Delta/8$ is required then it suffices to take $\Delta_0 \in O(m^2 \Delta)$.  Here $m\in \Theta(\sqrt{1/P_u})$ which means that $\Delta_0 \in O(\Delta/P_u)$.  Thus amplitude estimation requires $O(P_u/\Delta)$ repetitions of the amplified circuit used to produce $P(1)$ and $P(11)$.  Amplitude amplification results in an overhead of $O(\sqrt{\kappa}+\max_v \sqrt{\kappa_v})$ assuming $P_u\in \Theta(\frac{Z}{\kappa Z_{Q}})$. The total cost is then simply the product of these two costs resulting in the claimed complexities.
\end{proof}

The above processes can be repeated for each of the components of the gradient vector in order to perform an update of the weights and biases of the Boltzmann machine.
\begin{corollary}
Let $N_{\rm op}$ be the number of quantum operations and oracle calls needed to compute a component the gradient of $O_{\rm ML}$ using~\alg{derivAE} or~\cor{AEtrain} for a Boltzmann machine on a connected graph  scales as $$\tilde{O}\left( E N_{\rm op} \right),$$
if the learning rate $r$ is a constant.\label{cor:AE}
\end{corollary}
\begin{proof}
The proof is a trivial consequence of using the result of~\thm{AEtrain} $O(E)$ times to compute each of the components of the gradient vector.
\end{proof}

Unlike the prior algorithm, it is difficult to meaningfully compare the costs in~\cor{AE} to those incurred with training under contrastive divergence.  This is because~\alg{derivAE} uses quantum superposition to compute the relevant expectation values using all the training data simultaneously.  Thus each component of the derivative operator is computed using the entire set of data, and it is better to instead consider the run time as a function of the estimation error rather than the number of training vectors.  A natural metric for comparison is to imagine that the training set is drawn from a larger set of training data that could be considered.  In that case there is inherent sampling error in the expectation values computed over the training data of $O(1/\sqrt{N_{\rm train}})$.  Thus taking $\delta \in O(1/\sqrt{N_{\rm train}})$ is a reasonable choice to compare the two methods, but this is by no means the only way the two costs could be meaningfully compared.

Although the query complexity is independent of the number of training vectors, in order to assess the cost of this algorithm in practical examples we also need to include the costs of instantiating the oracle.  We consider three cases.  If each oracle implements an efficiently computable function then the space-- and time--complexities of implementing the oracle is polylogarithmic in $N_{\rm train}$.  On the other hand, if the data can only be accessed via a lookup table (as is true in most machine learning problems) then a quantum computer that allows parallel execution can implement the oracle in time $O({\rm polylog}(N_{\rm train}))$ using memory $O(N_{\rm train})$.  If on the other hand the quantum computer only can process information serially then $\Theta(N_{\rm train})$ space and time are required to implement an oracle query using a database of training vectors stored as a qubit string in the quantum computer.  The lower bound follows from lower bounds on the parity function that show $\Theta(N)$ queries to the bits in this database are required to determine the parity of a $N$ qubit string.  This shows that the dependence on the number of training vectors re--emerges depending on problem-- and architecture--specific issues.  

The quadratic scaling with $E$ means that~\alg{derivAE} may not be preferable to~\alg{deriv} for learning all of the weights.  On the other hand,~\alg{derivAE} can be used to improve gradients estimated using the prior method.  The idea is to begin with a preliminary gradient estimation step using the direct gradient estimation method while using $O({\sqrt{N_{\rm train}}})$ randomly selected training vectors.  Then the gradient is estimated by breaking the results into smaller groups and computing the mean and the variance of each component of the gradient vector over each of the subgroups.  The components of the gradients with the largest uncertainty can then be learned with error that is comparable to the sampling error incurred by only using $N_{\rm train}$ training examples in contrastive divergence training by using~\alg{derivAE} with $\delta\sim 1/\sqrt{N_{\rm train}}$ to estimate them.  Since the two costs are asymptotically comparable, this approach allows the benefits of both approaches to be used in cases where the majority of the uncertainty in the gradient comes from a small number of components.

\section{Hedging strategies}
Large values of $\kappa$ can be an obstacle facing exact state preparation.  This problem emerges because $Q(v,h)$ may assign orders of magnitude less probability to configurations than $P(v,h)$.  For example, we find examples of Boltzmann machines that require values of $\kappa$ in excess of $10^{20}$ are to \emph{exactly} prepare the Gibbs state even for small BMs.  The root of this problem is that taking $Q(v,h)$ to be the MF distribution does not always adequately reflect the uncertainty we have in $P(v,h)\approx Q(v,h)$.  We introduce ``hedging strategies'' to address this problem.  Our strategy is to introduce a hedging parameter $\alpha:1\ge\alpha\ge 0$ that can be adjusted to reduce bias towards the mean--field state.  In particular, if $\mu_i$ is the mean--field expectation value of $v_i$ in the absence of hedging then choosing $\alpha<1$ results in $\mu_i \rightarrow \alpha \mu_i + (1-\alpha)/2$ and similarly for the hidden units.  $Z_{Q}$ retains the same value regardless of $\alpha$.  

This strategy can also be thought of from a Bayesian perspective as parameterizing a prior distribution that transitions from complete confidence in the MF ansatz to a uniform prior.  Preparing the Gibbs state from this state then corresponds to an update of this prior wherein $\mathcal{P}(v,h)$ is the likelihood function in this language.

State preparation is essentially unchanged by the use of hedging, since the only difference is that the mean--field parameters that specify $Q(v,h)$ are altered.  Similar effects can also be obtained by using the method of~\cite{CW12} to prepare linear combinations of the mean--field state and the uniform distribution, but we focus on the former method for simplicity.  We see below that hedging does not substantially increase the overhead of the algorithm but can substantially reduce the complexity of GEQS and GEQAE when the MF approximation begins to break down.  
\section{Numerical experiments}\label{sec:NUM}
In this section we quantify the differences between training Boltzmann machines using contrastive divergence (see~\sec{CD} for a brief review of contrastive divergence) and training them by optimizing $O_{\rm ML}$ using \alg{deriv} or \alg{derivAE}.  Here we present additional data that examines the performance of our algorithms in a wider range of conditions than those considered in the main body.  In particular, we present detailed investigation into the quality of $O_{\rm ML}$ and $\kappa$ for both RBMs as well as full BMs.  We also present evidence showing that the results in the main body, which involved training on synthetic data sets, are comparable to those attained by training using sub--sampled handwritten digits taken from the MNIST database.

\subsection{Data and Methodology}\label{sec:data}


We train the dRBM model using gradient ascent with (a) contrastive divergence to approximate the gradients and (b) ML--objective optimization (ML) using techniques of \alg{deriv} or \alg{derivAE}.
The objective function in both cases is $O_{\rm ML}$.
Since different approximations to the gradient result in learning different local optima, even if the same initial conditions are used in both instances, it is potentially unfair to directly compare the optima vectors found using the two training techniques.
We consider two methods of comparison.
First, we verify that our results lie approximately in an optima of $O_{\rm ML}$ with high probability by using the approach of Donmez, Svore and Burges~\cite{DSB09}. For each proposed optima, we perform many perturbations about the point, fixing all parameters while perturbing one, and repeating many times for each parameter, and then compare the difference in the value of the objective functions at the original and perturbed points.  This allows us to say, with fixed confidence, that the probability of the objective function decreasing in a randomly chosen direction is less than a cited value.  We repeat this process $459$ times with perturbations of size $10^{-3}$, which is sufficient to guarantee that the objective function will not increase in $99\%$ of all randomly chosen directions for steps of size $10^{-3}$.

Second, we use a methodology similar to that used in~\cite{CH05}.  We perform our experiments by running one algorithm until a local optima is found and then use this local optima as the initial configuration for the second algorithm.  In this way we can compare the locations and quality of analogous local optima.  We list the training combinations in~\tab{cdml}, which we denote CD--ML, ML--CD, and ML--ML corresponding to the optimizers used in the first and second steps of the comparison.
A subtle point in considering such comparisons is determination of convergence.  Contrastive divergence and other noisy gradient ascent algorithms (meaning gradient ascent where noise is added to the gradient calculation) do not converge to a single optima, but instead fluctuate about an approximate fixed point.  For this reason we consider an algorithm to have converged when the running average of the value of ${O_{\rm ML}}$ varies by less than $0.001\%$ after at least $10,000$ training epochs with a learning rate of $r=0.01$.  We apply this stopping condition not only in our contrastive divergence calculations, but also when we determine the effect of introducing sampling noise into the gradient of the ML objective function.  We typically optimize the ML objective function in the absence of noise using the Broyden--Fletcher--Goldfarb--Shanno algorithm (BFGS), but also use gradient ascent using discretized derivatives.  In both cases, we choose our stopping condition to occur when an absolute error of $10^{-7}$ in the ML objective function is obtained.
\begin{table}
\begin{tabular}{|c|c|c|}
\hline
&First Optimizer &Second Optimizer\\
\hline
CD--ML& CD$-1$ & Gradient ascent on ML objective\\
ML--CD& BFGS on ML objective & CD$-1$\\
ML--ML&BFGS on ML objective & Noisy gradient ascent on ML objective.\\
\hline

\end{tabular}
\caption{Numerical experiments.}\label{tab:cdml}
\end{table}

\subsection{Effect of noise in the gradient}
We first determine whether a sufficiently accurate estimate of the gradient can be obtained from a small number of samples from a quantum device, for example when training using  \alg{deriv} or \alg{derivAE}.  
Here, we train a single--layer RBM using ML--ML, with $6$ visible units and $4$ hidden units.  
We then proceed by computing the gradient of the ML objective function and add zero-mean Gaussian noise to the gradient vector.  
The training data, consisting of $10,000$ training examples.
A minimum of $10,000$ training epochs were used for each data point and typically fewer than $20,000$ epochs were required before the stopping condition (given in~\sec{data}) was met.

\fig{noise} shows that the mean error in the values of the objective function scales quadratically with the noise in the gradient.  This means that the gradient ascent algorithm is highly resilient to sampling noise in the components in the gradient and a relatively small value of $\delta$, such as $\delta=0.01$, will suffice to yield local optima that are close to those found when $\delta=0$.  Therefore, small sampling errors will not be catastrophic for~\alg{derivAE}.

If the sampling error is zero mean then the learning rate can always be adjusted in order to mitigate such sampling errors.  This means that if improved gradients are needed then it is not necessary to either increase the size of the training set in GEQS.  However, there is no guarantee in GEQAE that the errors are unbiased so reducing the learning rate may not always suffice for reducing such errors.  Regardless, multiple strategies exist to reduce such errors to the $10^{-2}$ threshold empirically needed to achieve negligible error in the quality of the trained model.

\begin{figure}
\begin{minipage}{0.45\linewidth}
\includegraphics{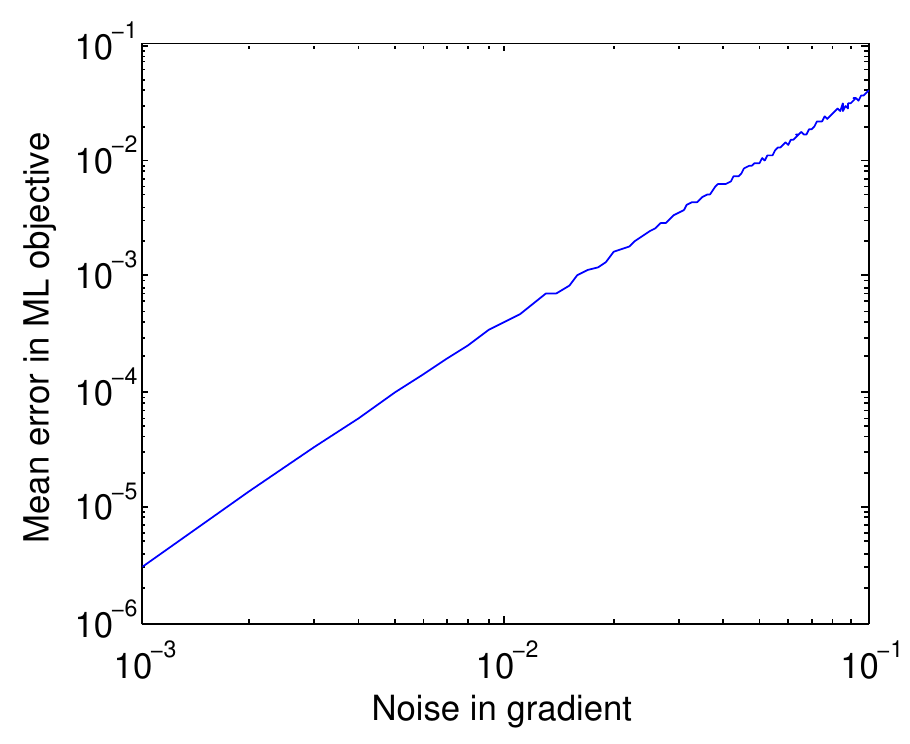}
(a)
\end{minipage}
\hspace{1mm}
\begin{minipage}{0.45\linewidth}
Fit to ${\rm Error} = a \times {\rm Noise}^b$.
\\\vspace{0.5cm}
\begin{tabular}{|c|c|c|}
\hline
Number of hidden units&a&b\\
\hline
$4$&$1.4044$&2.0161\\
$6$&$1.4299$&2.0121\\
$8$&$1.5274$&2.0086\\
\hline
\end{tabular}\\
(b)
\end{minipage}
\caption{(a) Mean discrepancy between ML--objective and ML found by initializing at ML--objective and introducing Gaussian noise to each component of the gradient computed.  Data taken for $6\times 4$ unit RBM with $\mathcal{N}=0$ and 100 samples used in the average for each data point. (b)  Results for $6 \times n_h$ unit RBMs for $n_h=6,8$ are qualitatively identical and strongly support a quadratic relationship between error and the noise in the gradient evaluation.\label{fig:noise}}
\end{figure}

\begin{figure}
\begin{minipage}{0.45\linewidth}
\includegraphics[width=\textwidth]{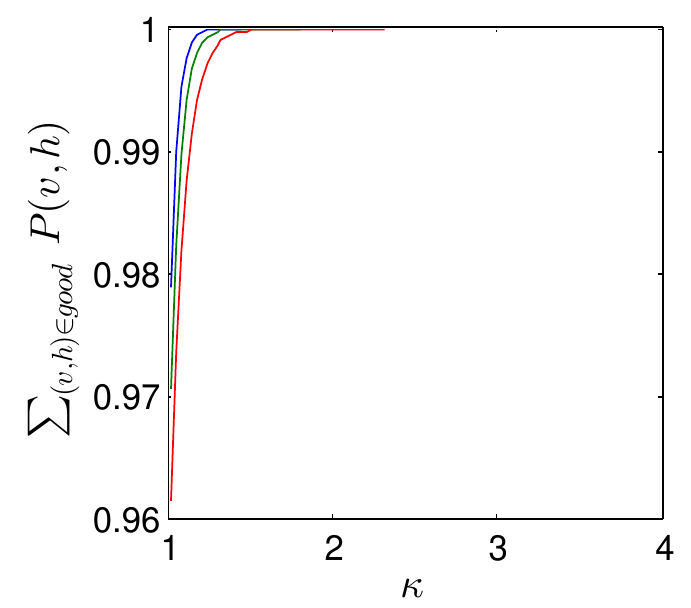}
(a) ${\rm Standard~deviation} = 0.1325$
\end{minipage}
\hspace{1mm}
\begin{minipage}{0.45\linewidth}
\includegraphics[width=\textwidth]{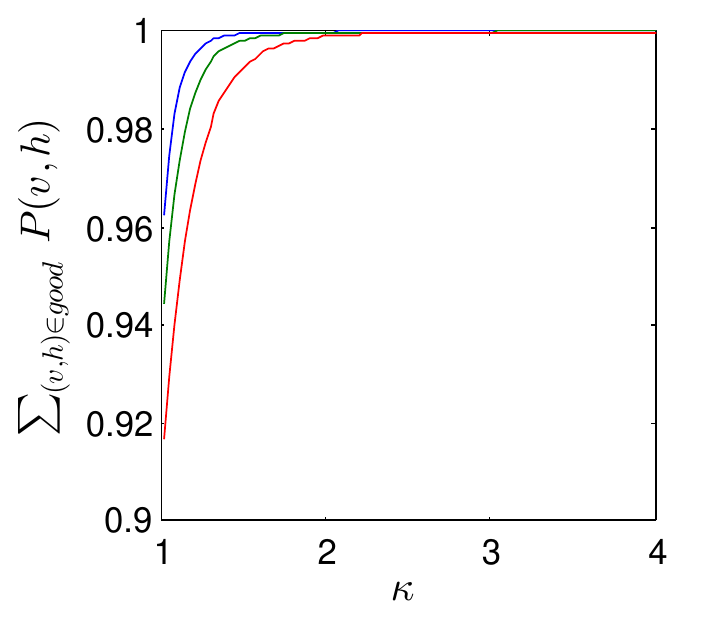}
(b) ${\rm Standard~deviation} = 0.2650$
\end{minipage}
\\
\begin{minipage}{0.45\linewidth}
\includegraphics[width=\textwidth]{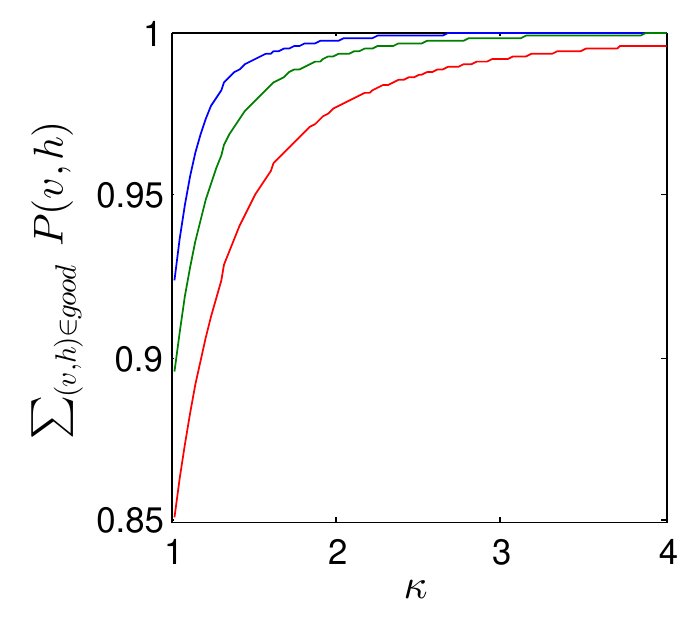}
(c) ${\rm Standard~deviation} = 0.5300$
\end{minipage}
\hspace{1mm}
\begin{minipage}{0.45\linewidth}
\begin{tabular}{|c|c|c|c|}
\hline
$\sigma(w_{i,j})$ & $n_v$ & Mean $ {\rm KL}(Q|| P)$&Mean $\ln(Z)$\\
\hline
0.1325 & 4 & 0.0031&$4.182$ \\
0.1325 & 6 & 0.0045&$5.550$\\
0.1325 & 8 & 0.0055&$6.942$\\
\hline
0.265 & 4 & 0.0115&4.211\\
0.265 & 6 & 0.0170&5.615\\
0.265 & 8 & 0.0245&7.020\\
\hline
0.53 & 4 &0.0423&4.359\\
0.53 & 6 &0.0626&5.839\\
0.53 & 8 &0.0827&7.295\\
\hline
\end{tabular}\\
\vskip1.2em
(d)
\end{minipage}
\caption{Fraction of probability for which $\mathcal{P}(v,h) \le 1$ vs $\kappa$ for synthetic $n_v \times 4$ RBMs with weights chosen with weights chosen randomly according to a normal distribution with varying standard deviation for the weights and $n=4,6,8$ from top to bottom in each graph. Table (d) shows the values of the KL--divergence between the mean--field distribution and the Gibbs distribution for the data shown in (a), (b) and (c).  Expectation values were found over $100$ random instances.}\label{fig:kappa}
\end{figure}

\subsection{Errors due to mean--field approximation and the scaling of $\kappa$}\label{sec:kappa}
Our quantum algorithm hinges upon the ability to prepare an approximation to the Gibbs state from a mean--field, or related, approximation.  \lem{succ} and \lem{kappa} shows that the success probability of the algorithm strongly depends on the value of $\kappa$ chosen and that the accuracy of the derivatives will suffer if a value of $\kappa$ is used that is too small or $Z_{Q}$ differs substantially from $Z$. We analyze the results for random single--layer RBMs with $4,6,8$ visible units and four hidden units.  The value of the standard deviation is an important issue because the quality of the mean--field approximation is known to degrade if stronger weights (i.e., stronger correlations) are introduced to the model~\cite{Wai05}.  We take the weights to be normally distributed in the Boltzmann machine with zero mean and standard deviation a multiple of $0.1325$, which we chose to match the standard deviation of the weight distribution empirically found for a $884$ unit RBM that was trained to perform facial recognition tasks using contrastive divergence.  The biases are randomly set according to a Gaussian distribution with zero mean and unit variance for all numerical experiments in this section.

\fig{kappa} shows that the value of $\kappa$ is not prohibitively large.  In fact, $\kappa<10$ suffices to produce the state with near-zero error for all of the cases considered.  Furthermore, we see that the value of $\kappa$ is a slowly increasing function of the number of visible units in the RBM and the standard deviation of the weights used in the synthetic models.  The fact that $\kappa$ is an increasing function of the standard deviation is not necessarily problematic, however, as regularization often causes the standard deviation of the weights to be less than $1$ in practical machine learning tasks.  It is difficult to extract the scaling of $\kappa$ from this data as the value chosen depends sensitively on the cutoff point chosen for the residual probability.


\begin{figure}
\includegraphics[width=0.49\linewidth]{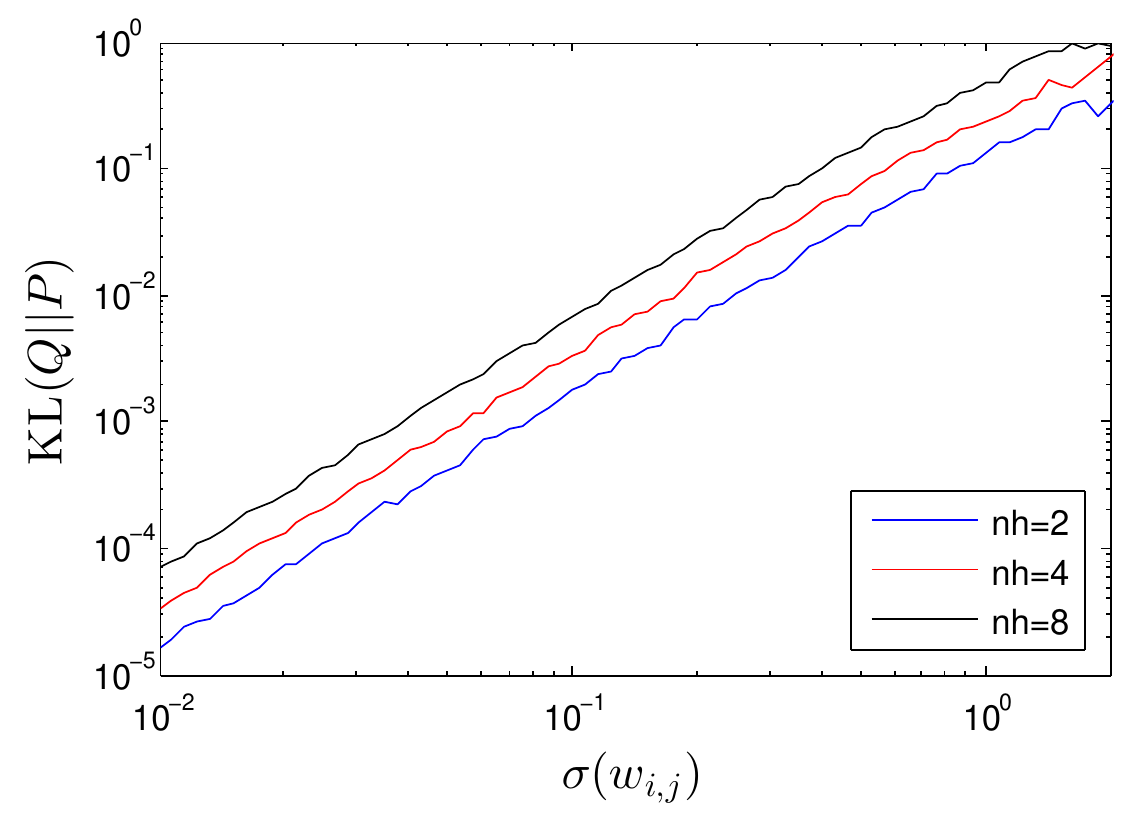}
\caption{Average KL--divergence as a function of the standard deviation of weights in synthetic RBMs with $n_v=4$ and Gaussian weights with zero mean and variance $\sigma^2(w_{i,j})$.  The biases were set to be drawn from a Gaussian with zero mean and unit variance.  Each data point is the average of $100$ random RBMs and the data is consistent with an $O(\sigma^2(w_{i,j})E)$ scaling. \label{fig:KL_nv=4}}
\end{figure}

\begin{figure}
\begin{minipage}{0.49\linewidth}
\includegraphics[width=\linewidth]{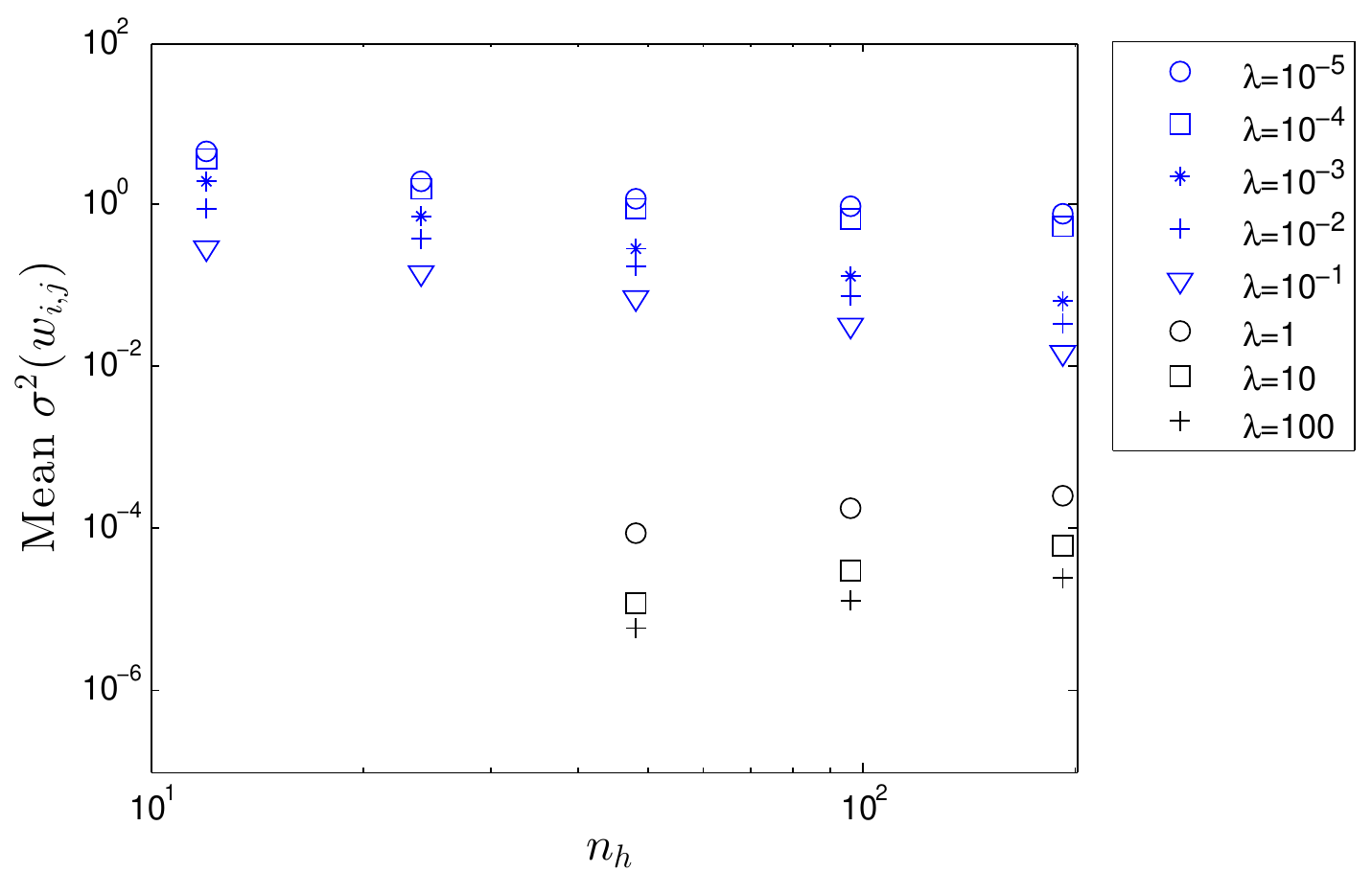}

(a)
\end{minipage}
\hspace{1mm}
\begin{minipage}{0.49\linewidth}
\includegraphics[width=\linewidth]{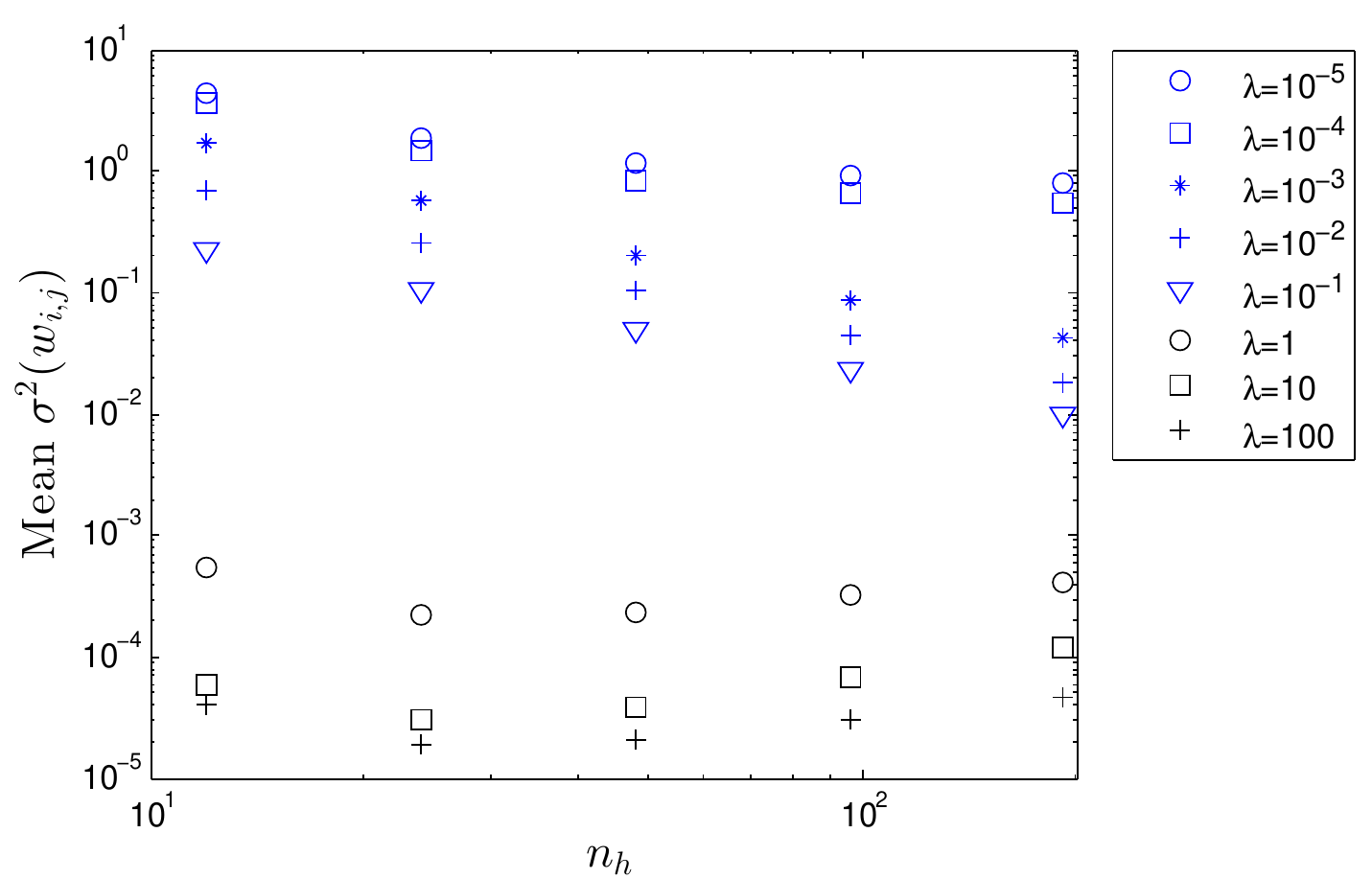}

(b)
\end{minipage}
\begin{minipage}{0.49\linewidth}
\includegraphics[width=\linewidth]{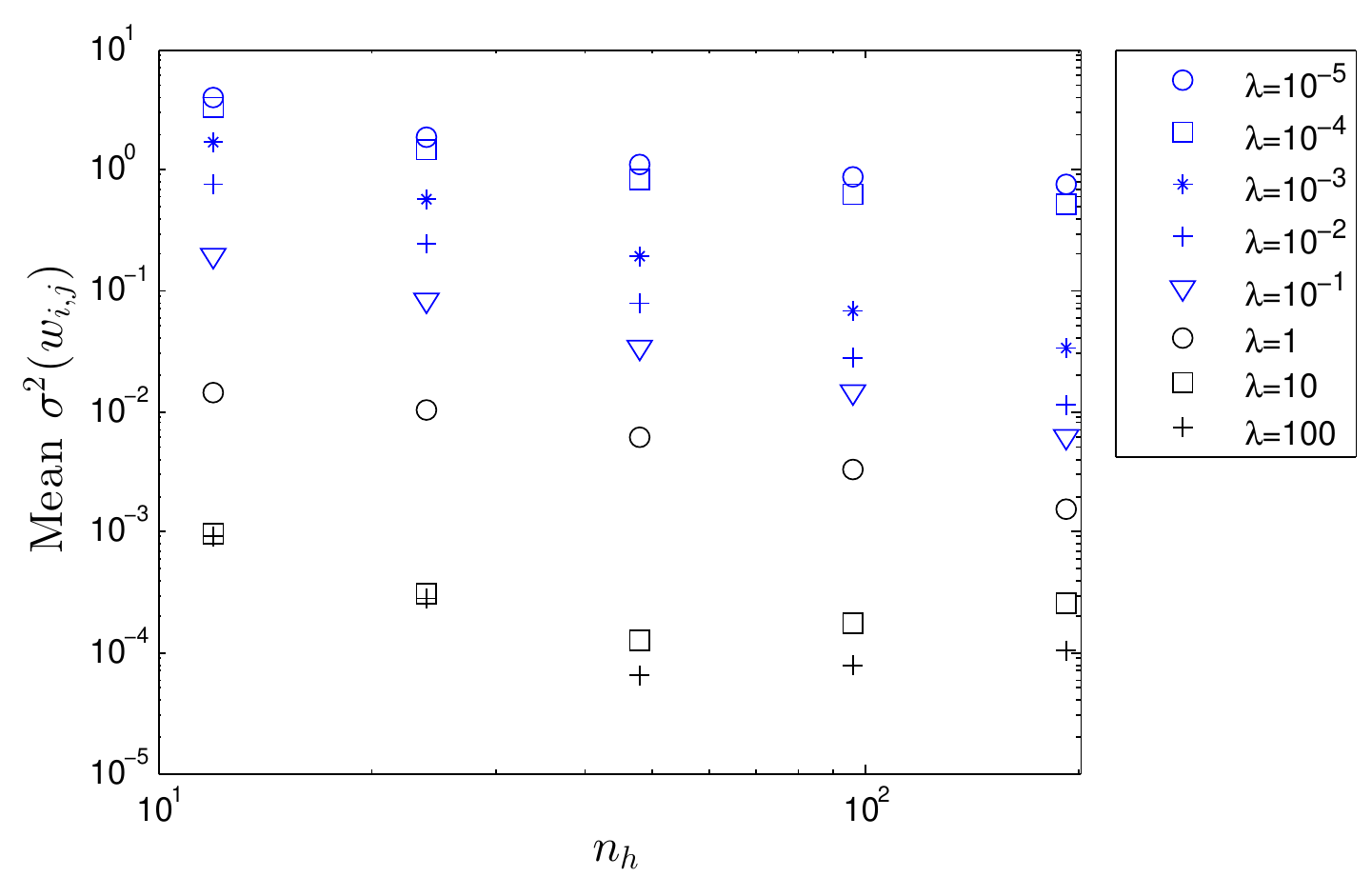}

(c)

\end{minipage}
\begin{minipage}{0.49\linewidth}
\includegraphics[width=\linewidth]{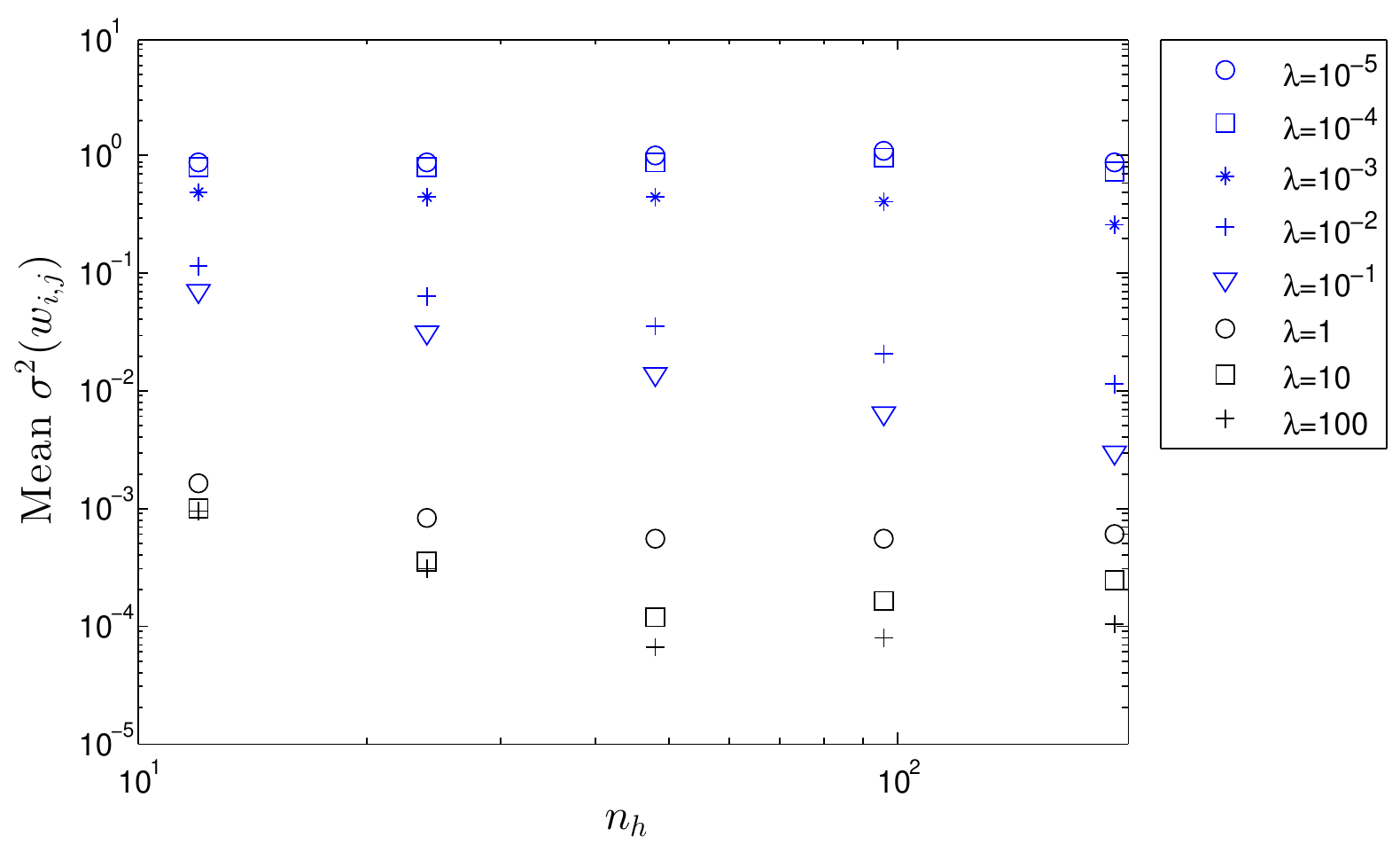}

(d)

\end{minipage}

\caption{Mean value of the variance of the weights for a RBM with (a) $n_v=12$, (b) $n_v=24$, (c) $n_v=48$ and (d) $n_v=48$ with $n_h=\{12,24,48,96,192\}$.  Each RBM was trained using $200,000$ epochs on our synthetic data set via CD$-1$ and learning rate $0.01$.  (a)--(c) use zero noise in the training set whereas (d) gives a $20\%$ chance of each bit flipping.  $100$ different local optima were averaged for each point in the plot.  \label{fig:sigma_CD_scale}}
\end{figure}

The quantity $\kappa$ is small because the mean--field approximation is very close to the true Gibbs state if the edge weights are small.  This can be seen from the Table in~\fig{kappa}(d) and~\fig{KL_nv=4}, which give the mean values of the  KL--divergence between the mean--field approximation and the true Gibbs state for the random RBMs.  
${\rm KL}(Q||P)$ tends to be less than $0.1$ for realistic weight distributions, implying that the mean--field approximation will often be very close to the actual distribution.  The KL--divergence is also the slack in the variational approximation to the log--partition function (see~\sec{meanfield}). This means that the data in~\fig{kappa} also shows that  $Z_{Q}$ will closely approximate $Z$ for these small synthetic models.  

There are two competing trends in the success probability.  As the mean--field approximation begins to fail, we expect that $\kappa$ will diverge.  On the other hand, we also expect $Z/Z_{Q}$ to increase as the ${\rm KL}$--divergence between $Q$ and $P$ increase.
We can better understand the scaling of the error by taking the scaling of $Z/Z_{Q}$ into consideration.  $Z_{Q}$ is a variational approximation to the partition function that obeys $\log(Z_{Q}) =\log(Z) - {\rm KL}(Q||P)$, which implies 
\begin{equation}
P_{\rm success} \ge \frac{Z}{Z_{Q} \kappa}= \frac{e^{{\rm KL}(Q||P)}}{\kappa}\ge \frac{1}{\kappa}.
\end{equation}
The data in~\fig{kappa} and~\fig{KL_nv=4} shows that ${\rm KL}(Q||P)$ empirically scales as $O(\sigma^2(w_{i,j}) E)$ ,where $E$ is the number of edges in the graph and $\sigma(w_{i,j})$ is the standard deviation in the weights of the synthetic models.  Thus we expect that (a) $P_{\rm success} \approx 1/\kappa$ if $\sigma^2(w_{i,j})\in O(1/E)$ and (b) $\kappa-1 \in O(\sigma^2(w_{i,j}) E)$ for $\sigma^2(w_{i,j}) E \ll 1$.  Thus our algorithms should be both exact and efficient if $\sigma^2(w_{i,j}) E$ is small for models that typically emerge in the training process.

We investigate this issue further, shown in~\fig{sigma_CD_scale}, by computing the typical distribution of weights for a RBM with $12,24$ and $48$ visible units and a variable number of hidden units.  This allows us to examine the scaling with the number of edges for a relatively large RBM trained using contrastive divergence.  Although the weights learned via contrastive divergence differ  from those learned using our quantum algorithm, we see in~\sec{compare} that these differences are often small for RBMs and so contrastive divergence gives us a good estimate of how $\sigma(w_{i,j})$ scales for large models that naturally arise through the training process.  We note from~\fig{sigma_CD_scale} that the standard deviation of the weights drops rapidly as more hidden units are added to the model.  This is because regularization (i.e., $\lambda>0$) provides a penalty for adding edges to the model.  The variance in the weights for $\lambda=0.1$ to $\lambda=0.001$  decays faster than the $\Theta(1/E)$ scaling that is expected to result in both high success probability and $Z_{Q} \approx Z$ for (a)-(c).  The results in (d) are qualitatively similar but the necessary scaling only holds for $\lambda=0.1$ to $\lambda=0.01$.  These values for regularization constants are typical of values used in practical ML algorithms and so we do not expect that the scaling of $\kappa$ nor the errors that arise from taking $Z\approx Z_{Q}$ will be an obstacle for applying our methods (or natural generalizations thereof) to practical machine learning problems.


\begin{figure}[t!]
\includegraphics[width=0.45\linewidth]{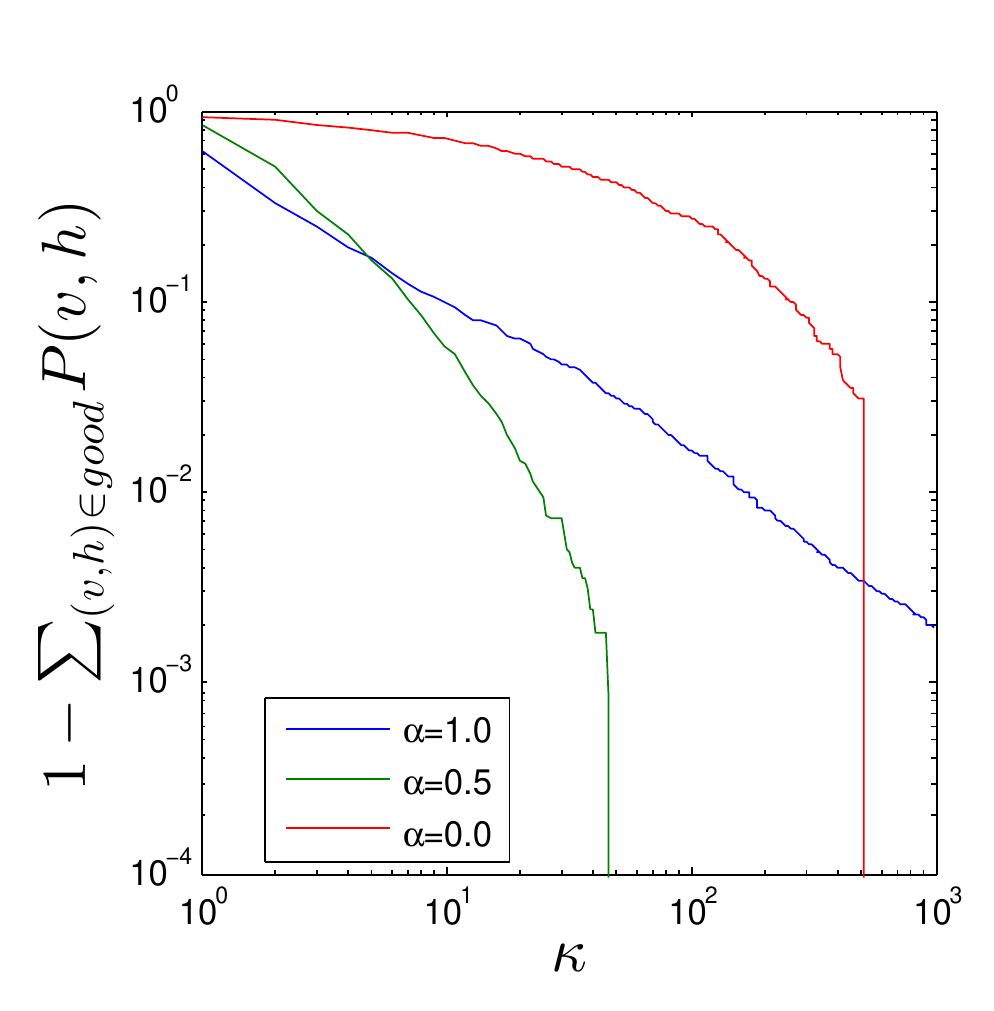}
\caption{Median probability mass of configurations such that $P(v,h)\ge 1$ for an RBM with $n_v=6$ and $n_h=8$ trained using the synthetic training set described in the main body with $\mathcal{N}=0$.}
\label{fig:hedge}
\end{figure}

Several strategies can be used to combat low success probability for the preparation of the Gibbs state.
In the event that the success probability is unacceptably low a more accurate estimate of the partition function than $Z_{Q}$ can be used in the algorithm~\cite{Xin02, OW01,SM08,SH09}.  \alg{deriv} can be used instead of~\alg{derivAE} to provide a quadratic advantage in the success probability.  The value of $\kappa$ chosen can also be decreased, as per \lem{kappa}.  In extreme cases, the regularization constant can also be adjusted to combat the emergence of large weights in the training process;  however, this runs the risk of producing a model that substantially underfits the data.

Hedging can also be used to address the issues posed by values of $\kappa$ that may be impractically large.  We observe this in~\fig{hedge} wherein the probability mass of states such that $\mathcal{P}(v,h)\ge 1$ is investigated.  We see in such cases that no hedging ($\alpha=1$) is superior if an accurate state is desired with a minimal value of $\kappa$.  A modest amount of hedging ($\alpha=0.5$) results in a substantial improvement in the accuracy of the state preparation:  $\kappa=50$ is sufficient to achieve perfect state preparation as opposed to $\kappa>1000$ without hedging.  The case of $\alpha=1$ is inferior to either $\alpha=0$ or $\alpha=0.5$ except for the fact that $Q(v,h)$ no longer needs to be computed coherently in those cases (the mean--field calculation is still needed classically to estimate $Z$).  However, as we see later hedging strategies tend not to be as successful for larger problems, whereas using the mean--field state as the initial distribution $(\alpha=1)$ tends to lead to nearly constant values of $\kappa$ as the size of the systems increases.
\begin{figure}[t!]
\begin{minipage}{0.45\linewidth}
\includegraphics[width=\linewidth]{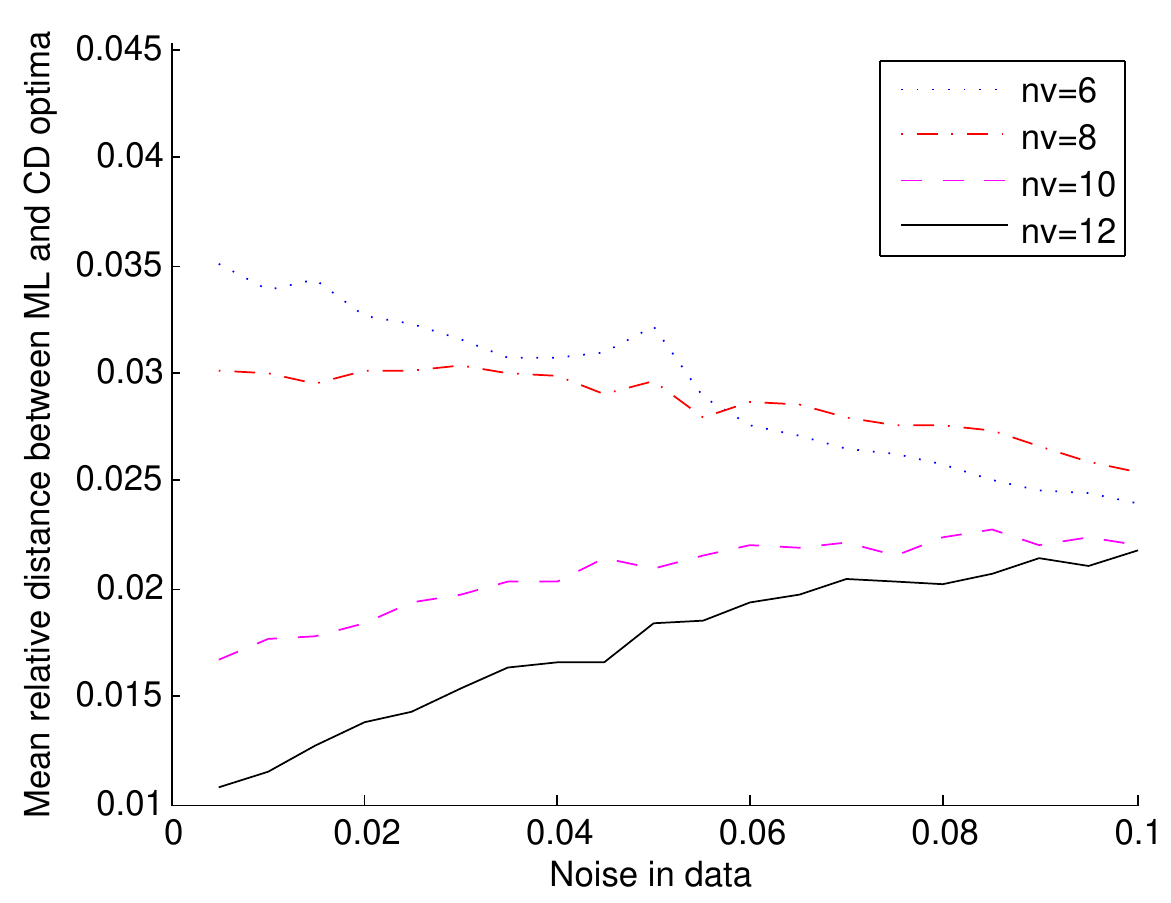}
\\
(a)
\end{minipage}
\hspace{1mm}
\begin{minipage}{0.45\linewidth}
\includegraphics[width=\linewidth]{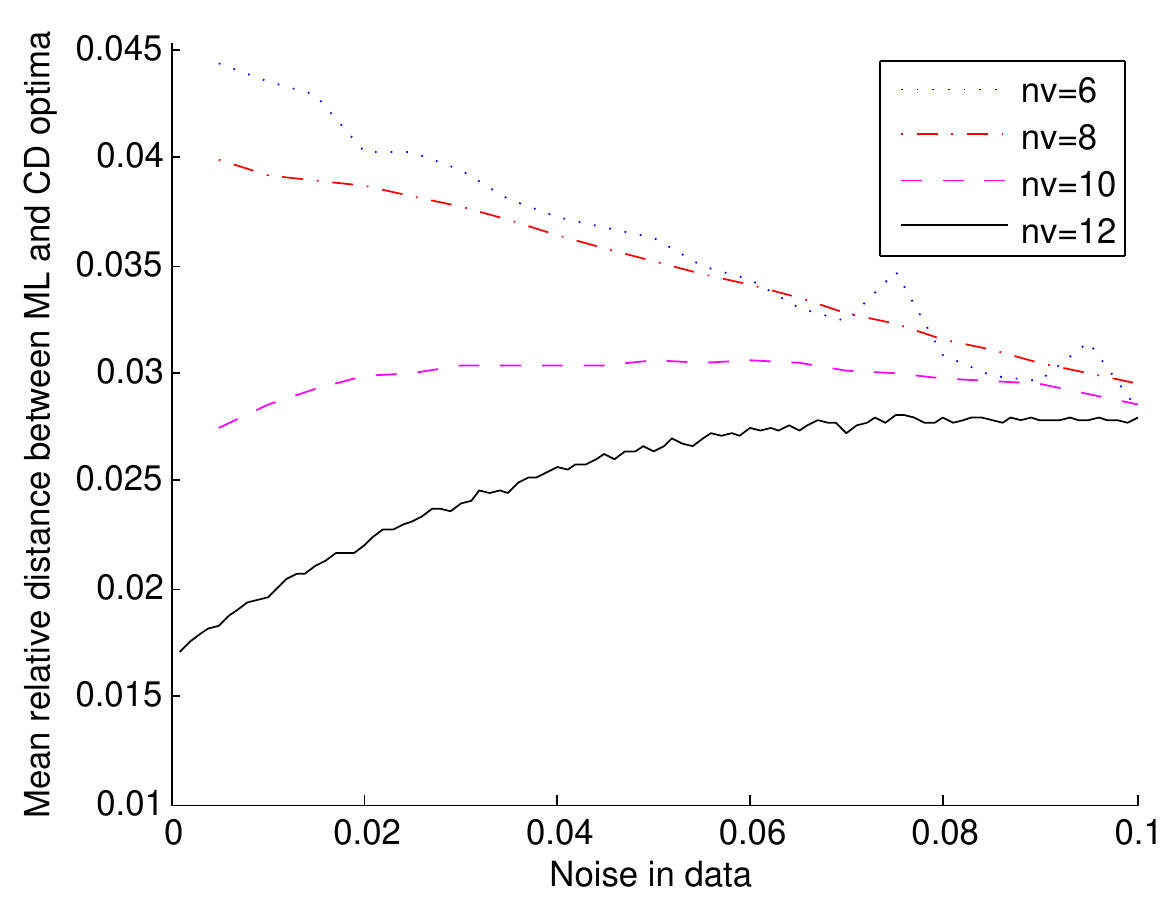}\\
(b)
\end{minipage}
\caption{Relative distances between the CD optima and ML optima for RBMs with $n_v$ visible units and (a) $4$ hidden units and (b) $6$ hidden units.\label{fig:CDMLdist}}
\end{figure}

\begin{figure}
\begin{minipage}{0.45\linewidth}
\includegraphics[width=\linewidth]{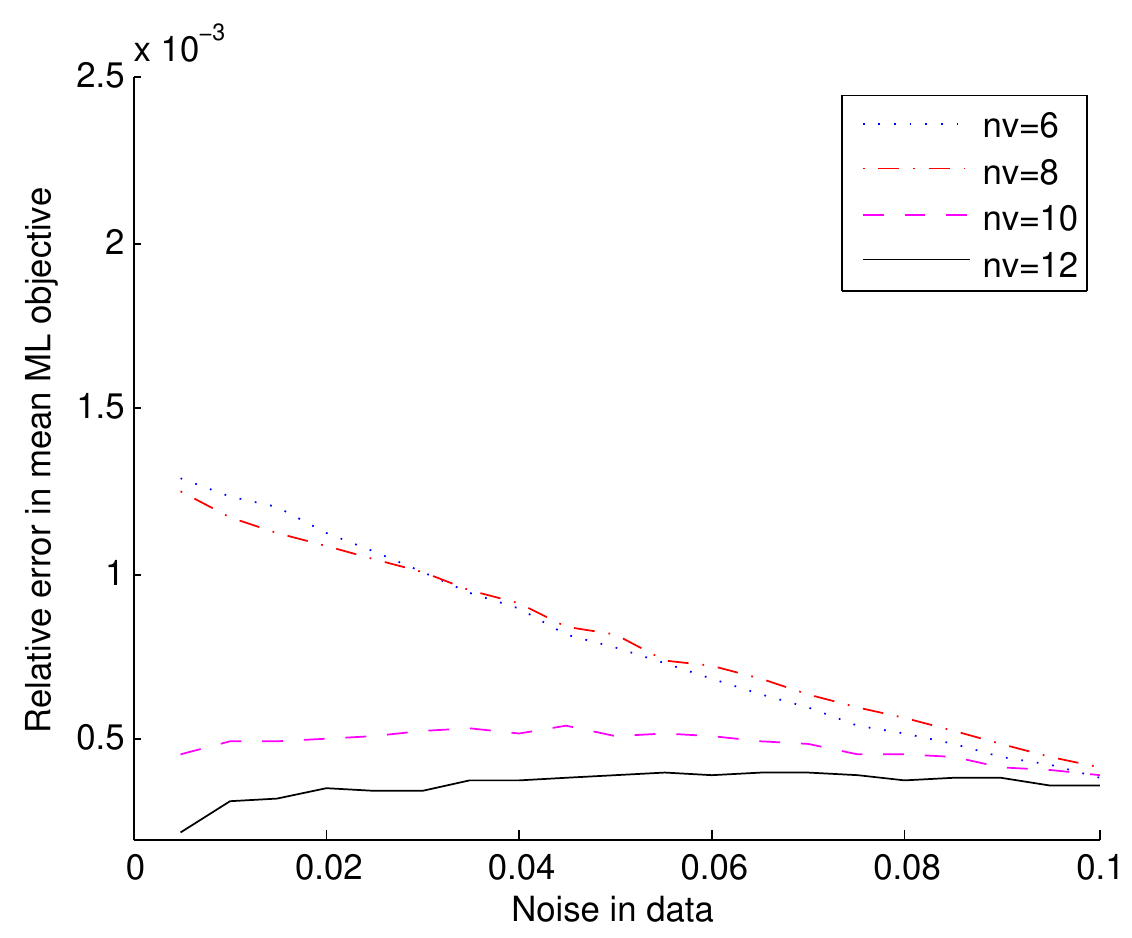}
\\
(a)
\end{minipage}
\hspace{1mm}
\begin{minipage}{0.45\linewidth}
\includegraphics[width=\linewidth]{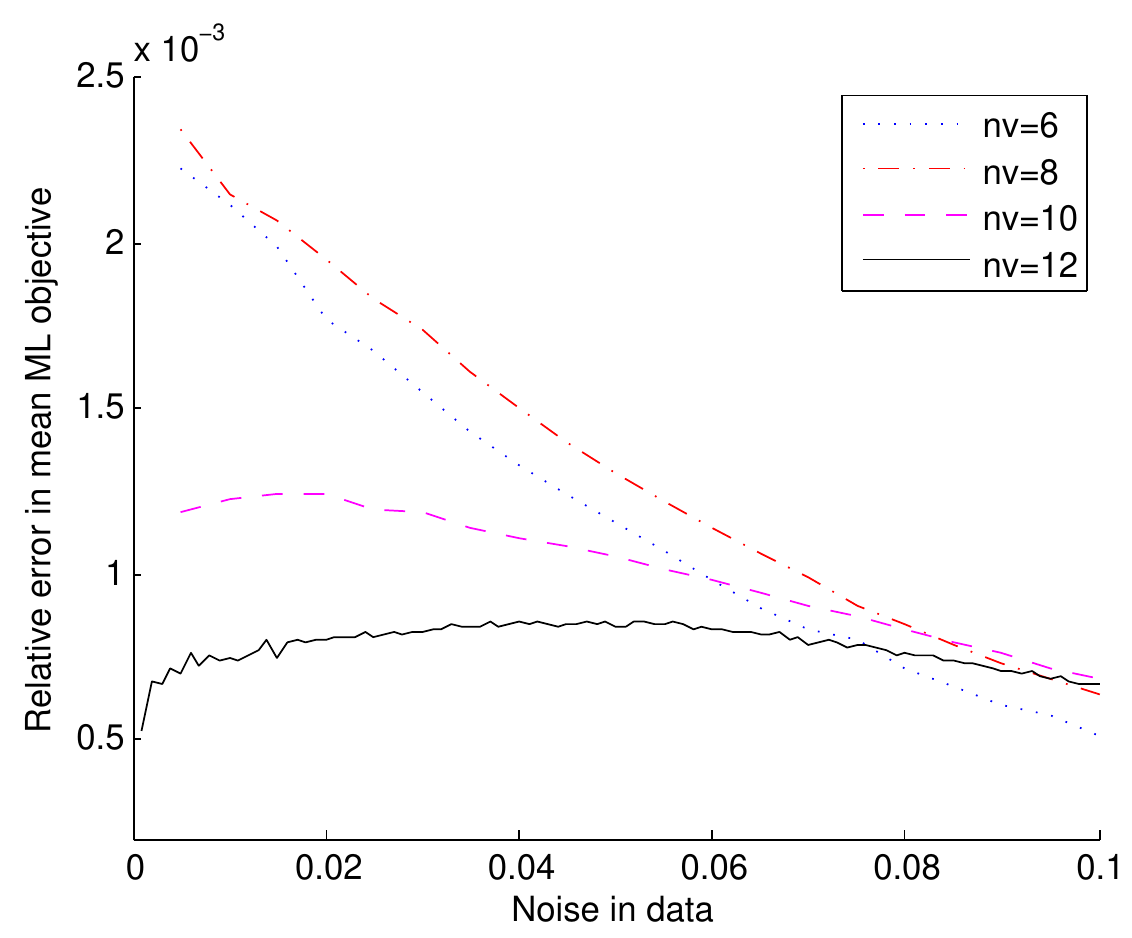}\\
(b)
\end{minipage}
\caption{Relative difference in the values of the ML objectives computed at the CD optima and the ML optima for RBMs with $n_v$ visible units and (a) $4$ hidden units and (b) $6$ hidden units.\label{fig:CDMLobj}}
\end{figure}

\subsection{Comparison of CD$-1$ to ML learning}\label{sec:compare}
An important advantage of our quantum algorithms is that they provide an alternative to contrastive divergence for training deep restricted Boltzmann machines.  We now investigate the question of whether substantial differences exist between the optima found by contrastive divergence and those found by optimization of the ML objective.
We train using CD--ML on single--layer RBMs with up to 12 visible units and up to 6 hidden units and compute the distance between the optima found after first training with CD$-1$ and then training with ML, starting from the optima found using CD.  
We find that the locations of the optima found using both methods differ substantially.

\fig{CDMLdist} illustrates that the distances between the contrastive divergence optima and the corresponding ML optima
are quite significant.  
The distances between the models found by ML training and CD training were found by flattening the weight matrix to a vector, concatenating the result with the bias vectors, and computing the Euclidean distance between the two vectors.  Differences on the order of a few percent are observed in the limit of no noise.  
This suggests that the models learned using CD optimization and ML optimization can differ substantially.  We see that these differences tend to increase as more hidden units are added to the model and that adding Bernoulli noise to the training data tends to cause the differences in the relative distances that arise from varying $n_v$ to shrink.

\fig{CDMLobj} shows that there are differences in the quality of the ML optima found as a function of the Bernoulli noise added to the training data, where quality is determined based on the value of $O_{\rm ML}$ at that point.  The relative errors observed tend to be on the order of $0.1$ percent for these examples, which is small but non--negligible given that differences in classification error on this order are significant in contemporary machine learning applications.  The discrepancies in the values of $O_{\rm ML}$ follow similar trends to the data in~\fig{CDMLdist}.  

These results show, even for small examples, that significant differences exist between the locations and qualities of the CD and ML optima.  Thus our quantum algorithms, which closely approximate ML training, are likely to lead to improved models over current state-of-the-art classical methods based on contrastive divergence if a reasonably small value of $\kappa$ suffices.  This point also is significant for classical machine learning approaches, wherein the use of more costly variants of contrastive divergence (such as CD$-k$ for $k>1$) may also lead to significant differences in the quality of models~\cite{Tie08}.

\subsection{Training full Boltzmann machines under $O_{\rm ML}$}
The prior examples considered the performance of ML--based learning on single-- and multi--layer restricted Boltzmann machines.  Here we examine the quality of ML optima found when training a full Boltzmann machine with arbitrary connections between any two units. 
While classical training using contrastive divergence requires learning over a layered bipartite graph (dRBMs), our quantum algorithms do not need to compute the conditional distributions and can therefore efficiently train full Boltzmann machines given that the mean--field approximation to the Gibbs state has only polynomially small overlap with the true Gibbs state.  
The main question remaining is whether there are advantages to using a quantum computer to train such complete graphical models, and if such models exhibit superior performance over dRBMs.

\fig{BM} shows that the ML objective function found by training a full Boltzmann machine slowly improves the quality of the optima learned as the number of visible units increases.  Although this increase is modest over the range of $n_h$ considered, it is important to note that the value of the mean ML objective attained via BFGS optimization on a single--layer RBM with six visible and four hidden is approximately $-2.33$.  Even the full Boltzmann machine with $7$ units and $21$ edges provided a much better model than an RBM with $10$ units and $24$ edges.  Although this numerical example is quite small, it demonstrates the benefit of introducing full connectivity, namely intra--layer connections, to a Boltzmann machine and therefore suggests that our quantum learning algorithm may lead to better models than those that can be efficiently learned using existing methods.

\begin{figure}[t!]
\includegraphics[width=0.4\linewidth]{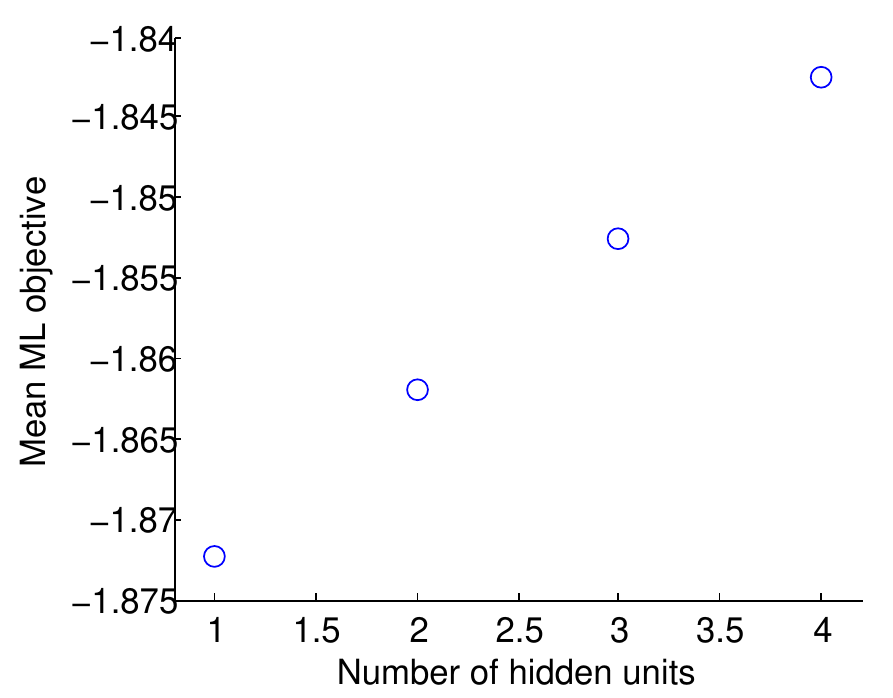}
\caption{$O_{\rm ML}$ for a fully connected Boltzmann machine with six visible units and one through four hidden units.\label{fig:BM}}
\end{figure}

\begin{figure}[t!]
\begin{minipage}{0.45\linewidth}
\includegraphics[width=\textwidth]{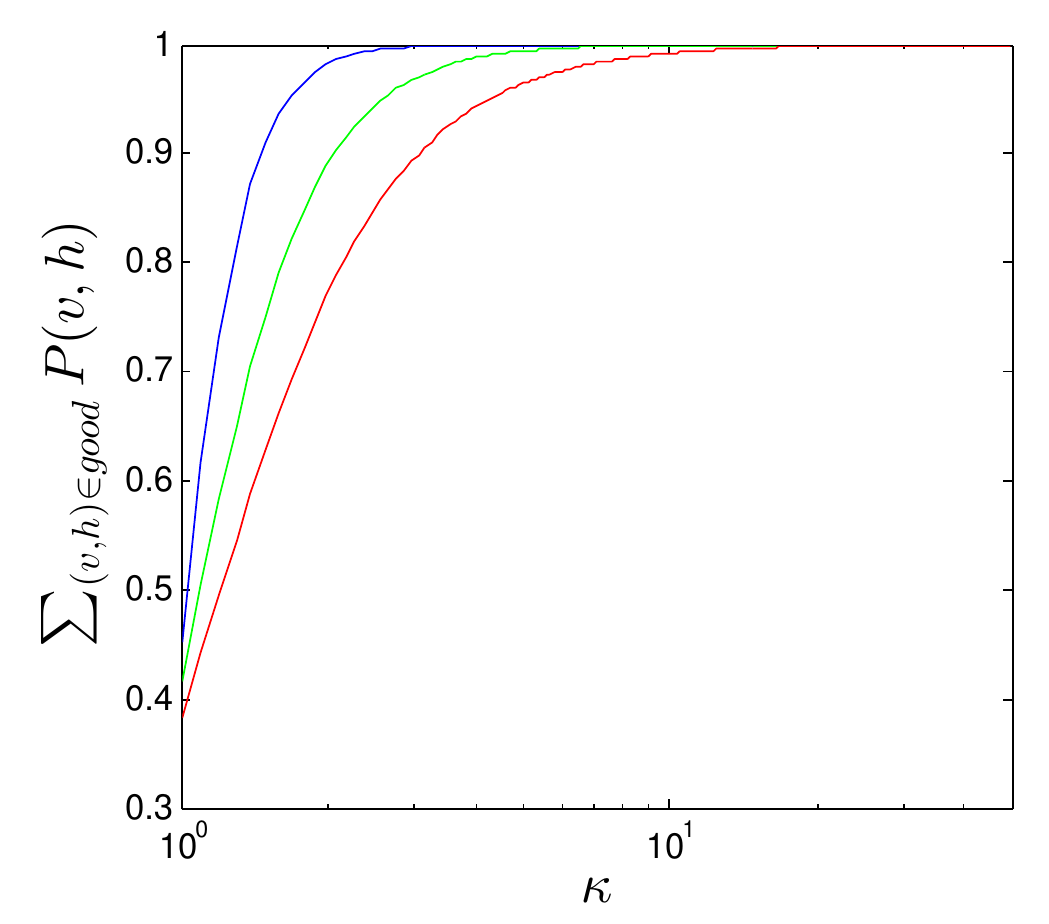}
(a) Standard deviation $=0.25$
\end{minipage}
\hspace{1mm}
\begin{minipage}{0.45\linewidth}
\includegraphics[width=\linewidth]{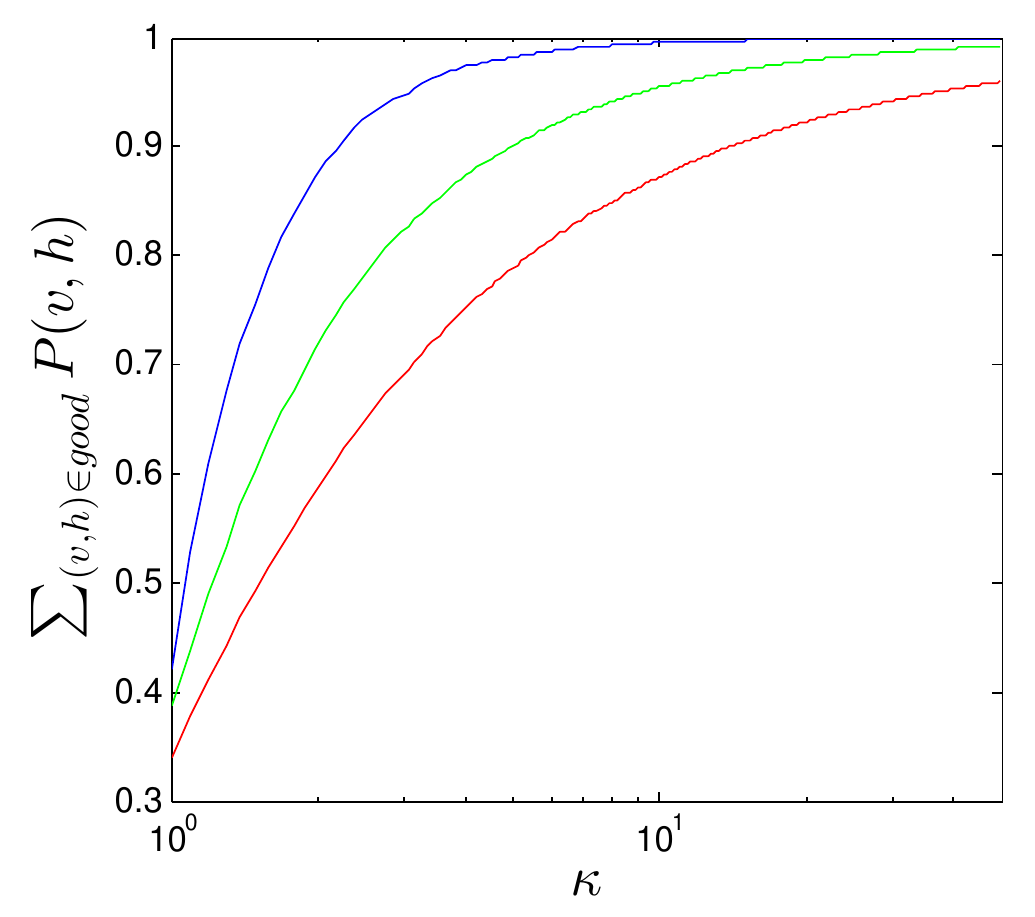}
(b) Standard deviation $=0.5$
\end{minipage}

\begin{minipage}{0.45\linewidth}
\includegraphics[width=\linewidth]{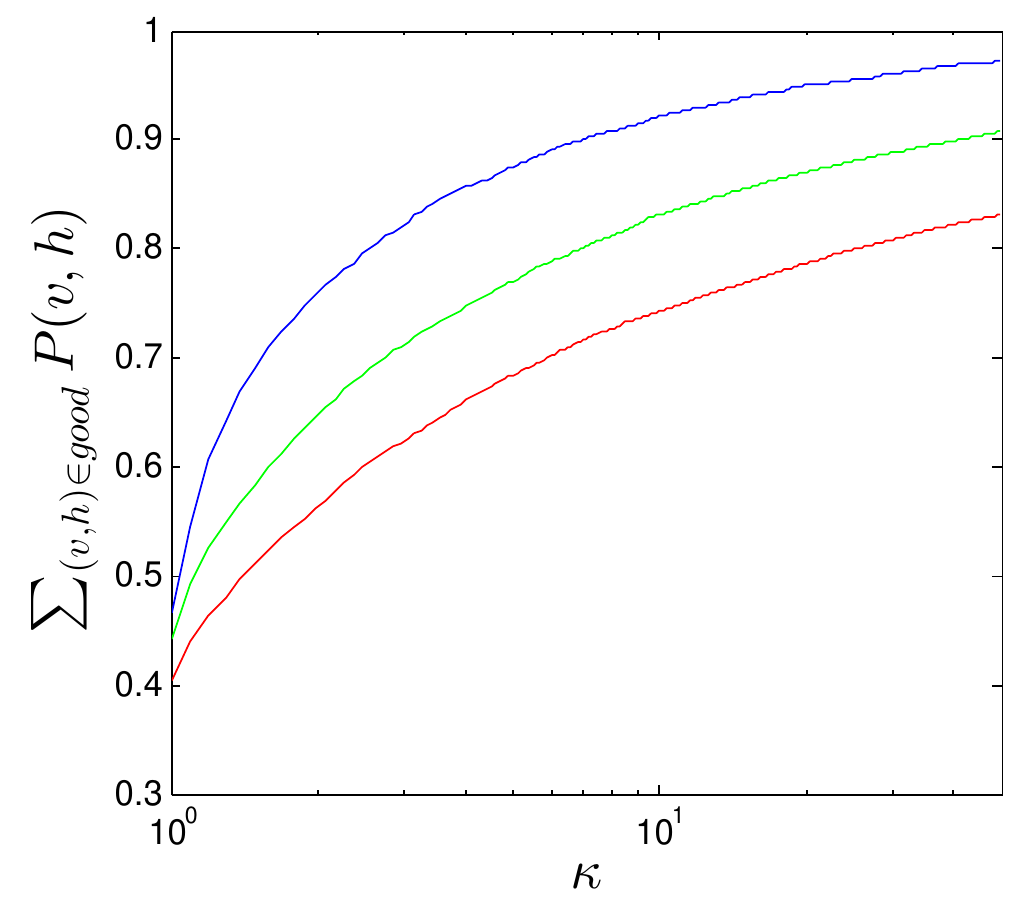}
(c) Standard deviation $=1$
\end{minipage}
\hspace{1mm}
\begin{minipage}{0.45\linewidth}
\begin{tabular}{|c|c|c|c|}
\hline
$\sigma(w_{i,j})$ & $n_v$ & Mean $ {\rm KL}(Q|| P)$&Mean $\ln(Z)$\\
\hline
0.25 & 8 & 0.0371&$6.673$ \\
0.25 & 14 & 0.1102&$12.020$\\
0.25 & 20 & 0.2193&$17.287$\\
\hline
0.5 & 8 & 0.1222&7.125\\
0.5 & 14 & 0.3100&13.701\\
0.5 & 20 & 0.5528&20.228\\
\hline
1 & 8 &0.3209&8.807\\
1 & 14 &0.5187&19.200\\
1 & 20 &0.7822&29.883\\
\hline
\end{tabular}\\
\vskip1.2em
(d)
\end{minipage}
\caption{Fraction of probability for which $\mathcal{P}(v,h) \le 1$ vs $\kappa$ for synthetic full Boltzmann machines on $n$ qubits where $n=8,14,20$ (from top to bottom) with weights randomly initialized according to a Gaussian distribution with zero mean and variance $\sigma^2$ and biases that are set similarly but with mean $0$ and unit variance..  The data shows that the mean--field approximation begins to rapidly break down for Boltzmann machines with that have large weights.  All data points are computed using $100$ random instances.\label{fig:kappaBM}}
\end{figure}

An important limitation of this approach is that the mean-field approximation tends to be much worse for Ising models on the complete graph than
it is for layered networks~\cite{Wai05}.  This means that the value of $\kappa$ needed may also be larger for these systems.  Although the results in~\cite{Wai05} show acceptable performance in cases of small edge weight, further work is needed to investigate the trade off
between model quality and training time for the quantum algorithm.

The remaining question is how does the success probability scale as a function of the standard deviation and number of units in the BM?  We examine this in~\fig{kappaBM} where we find qualitatively similar results to those seen in~\fig{kappa}.  The most notable difference is that we consider a wider range of visible units and larger standard deviations to illustrate how the algorithm can fail in cases where the mean--field approximation is no longer applicable.  This behavior is most notable in (c) of~\fig{kappaBM}, where values of $\kappa>50$ are needed for accurate state preparation.  In contrast, we see that in (a) that smaller weights tend to cause the state preparation algorithm to be much more efficient and $\kappa<20$ suffices to have perfect state preparation even for networks with $20$ hidden units.

Although the expectation values seem to indicate that  the success probability systematically shrinks as a function of $\sigma$, this is not necessarily true.  In our sampling we also find evidence that there are easy as well as hard instances for the state preparation.  This point can be seen in~\fig{95CI} where we plot a $95\%$ confidence interval and see, somewhat surprisingly, that many of the models conform to a mean--field distribution for the $\sigma=1$ data.  In fact, for small $\kappa$, the $95^{\rm th}$ percentile of $\sigma=1$ actually provides a more accurate approximation to the Gibbs state than the corresponding percentile for $\sigma=0.5$.  Conversely, $5^{\rm th}$ percentile for the $\sigma=1$ data has very poor fidelity and does not seem to scale qualitatively the same way with $\kappa$ as the rest of the data considered.  

This shows that although the situation here is qualitatively similar to that of the RBM, larger values of $\kappa$ will be required for the full Boltzmann machine to achieve the same fidelity as a RBM could achieve.  Small numerical studies, unfortunately, are inadequate to say conclusively whether the models that typically arise during training will conform to these easy cases or the hard cases.

The mean value of $\sum_{(v,h)\in {\rm good}} P(v,h)$ in~\fig{95CI} scales as $1-\kappa^{f(n,\sigma)}$ for some function $f(n,\sigma)$.  We find by fitting the data for $n=6,8,\ldots,20$ for $\sigma=0.25,0.5,0.75,1$ that $\sum_{(v,h)\in {\rm good}} P(v,h)-1 \in \kappa^{\Theta(-1/\sigma^{1.5}n)}$.  Hence if a value of $\kappa$ is sought that causes the error in the final Gibbs state to be $\delta$ then it suffices to take $\kappa \in \delta^{-\Theta(\sigma^{1.5} n)}$.  Thus these small numerical experiments suggest that state preparation for the full Boltzmann machine will likely be efficient and exact if $\sigma^2 \in o(n^{-4/3})$ and that it may be inefficient or approximate otherwise.

\begin{figure}
\includegraphics[width=\linewidth]{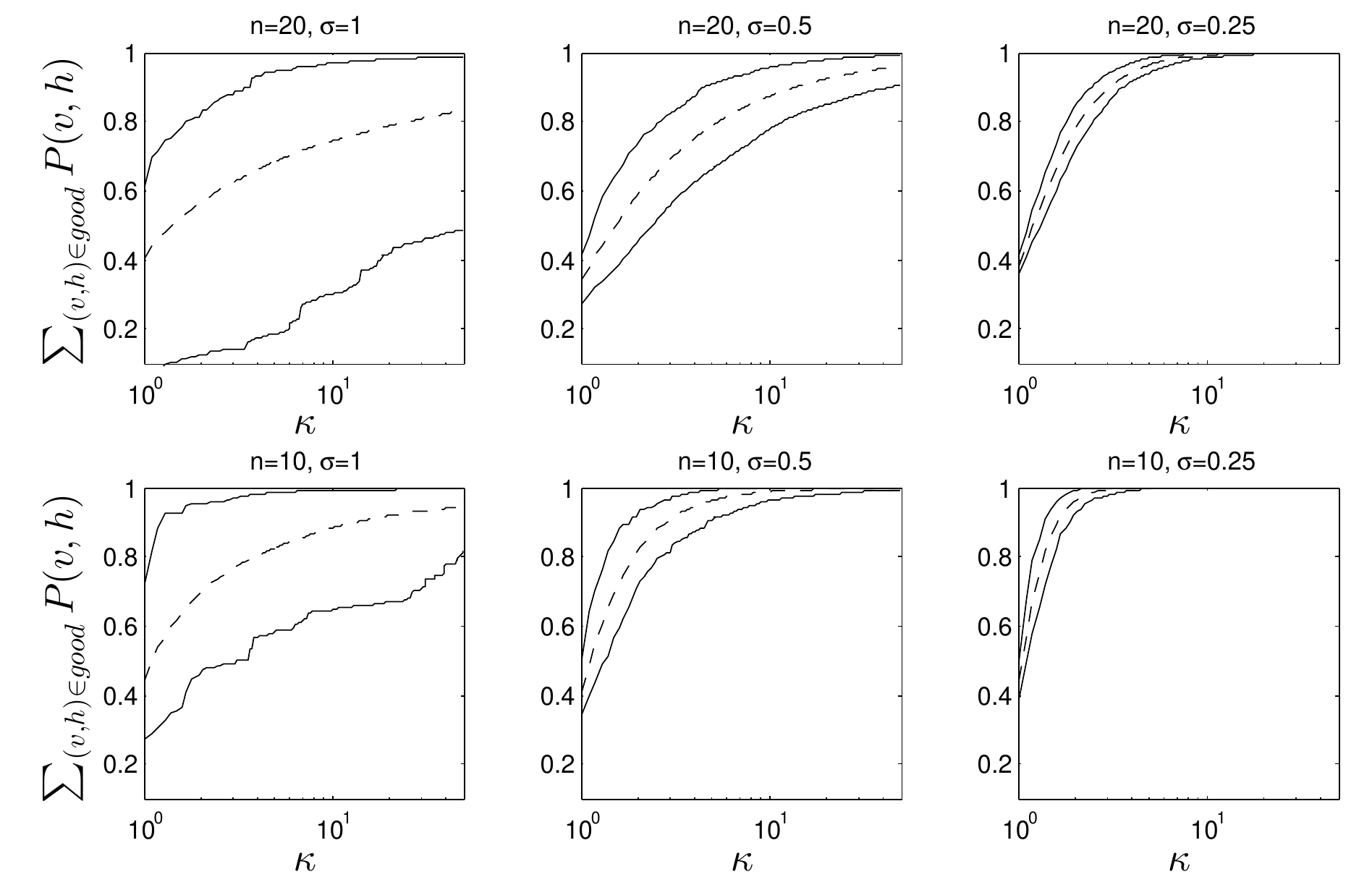}
\caption{Expectation values (dashed lines) and $95\%$ confidence intervals for the fraction of the probability distribution that cannot be prepared properly using a fixed value of $\kappa$ as a function of $\kappa$ for unrestricted Boltzmann machines. The data shows that  the distribution is not strongly concentrated about the mean for $\sigma\ge 1$. \label{fig:95CI}}
\end{figure}

\begin{figure}[t!]
\includegraphics[width=0.7\linewidth]{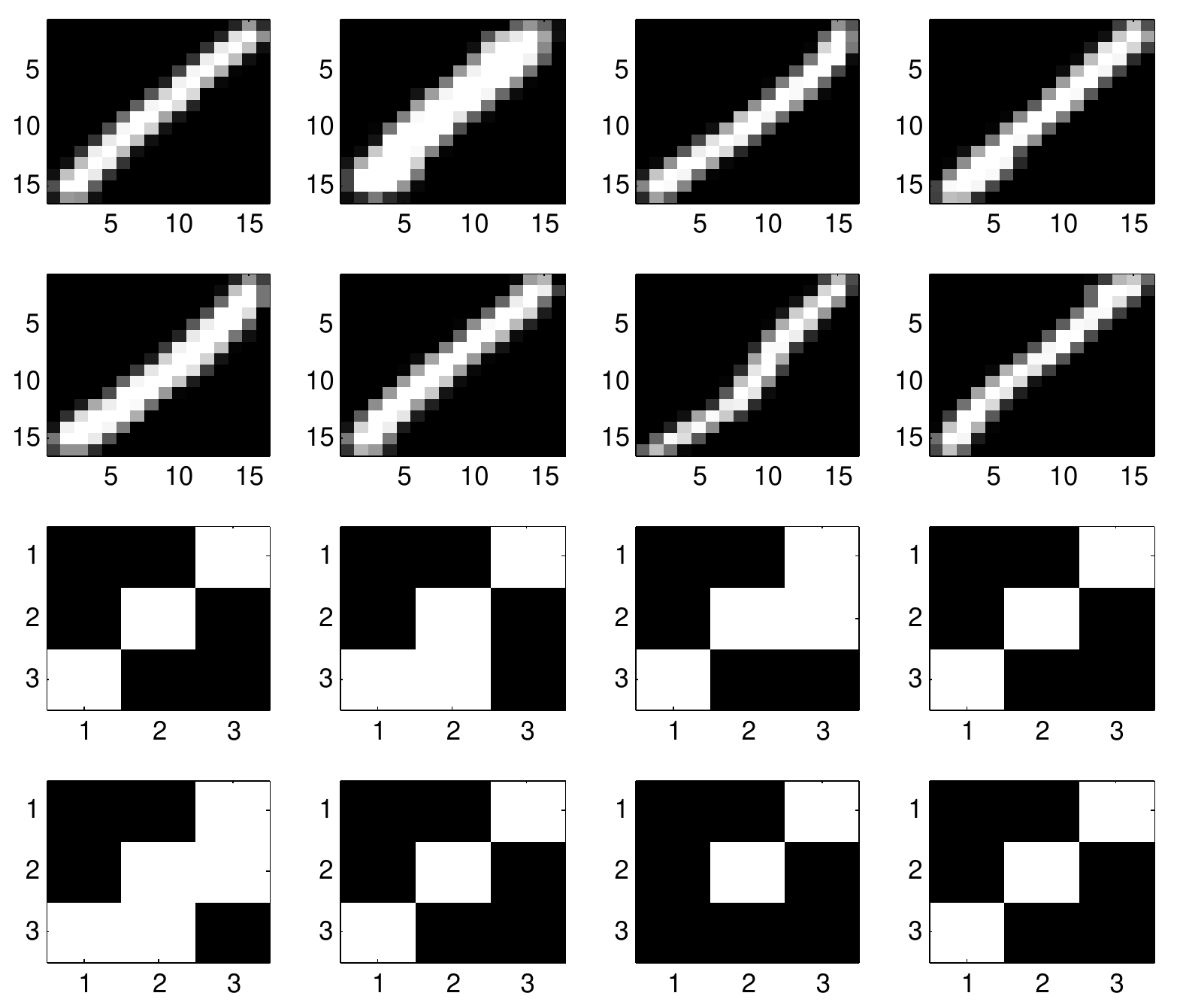}
\caption{Examples of original (top) and sub--sampled (bottom) MNIST ones digits.\label{fig:mnistexamp}}
\end{figure}

\section{Training using sub--sampled MNIST data}
An important criticism of the previous numerical results is that they only examine synthetic data sets.  In order to provide some insight about whether natural data sets are qualitatively different from the synthetic examples we now consider training examples that are taken from the MNIST database of handwritten digits.  The MNIST digits are $16\times 16$ greyscale images and as a result we cannot directly compute $O_{\rm ML}$ for them as computing $P(v,h)$ requires computing $e^{-E(v,h)}$ for a minimum of $2^{256}$ distinct configurations.  We instead look at a simpler problem that consists of $3\times 3$ coarse--grained versions of the original images.  As the resolution of these examples is low enough that the digits can be very difficult to distinguish, we focus on the training examples that are assigned the label of ``1'' in order to avoid confusion with other digits that may appear similar on a $9$ pixel grid.  The resultant $400$ training examples are found by dividing the image into thirds, computing expectation value of all the pixels within that third of the image and setting the corresponding pixel in the $3\times 3$ image to that value.  We then round the pixels to binary values using the mean value of the pixels as a threshold.  The results from this subsampling procedure are illustrated in~\fig{mnistexamp}.

\fig{realCDML} compares the quality of optima found via CD--1 and gradient ascent on $O_{\rm ML}$ for CD--ML experiments as a function of the number of hidden units ($n_h\in\{4,6,8,10\}$).  The expectation values were found by using $1000$ random restarts for the training procedure.  We observe that the relative differences between the locations of the optima found in these experiments vary by a few percent whereas the difference in $O_{\rm ML}$ varies by as much as half a percent.  These differences are comparable to those observed for the relative discrepancies in the trained models and quality of the resultant optima found for the synthetic training sets used in the main body.  

We see evidence that the discrepancies between contrastive divergence training and ML training grow approximately linearly with the number of hidden units in the graph.  We cannot say with confidence that this constitutes a linear scaling in the asymptotic regime because the discrepancies grow modestly with $n_h$ and polynomial or exponential scalings cannot be excluded.

\begin{figure}[t!]
\includegraphics[width=0.5\linewidth]{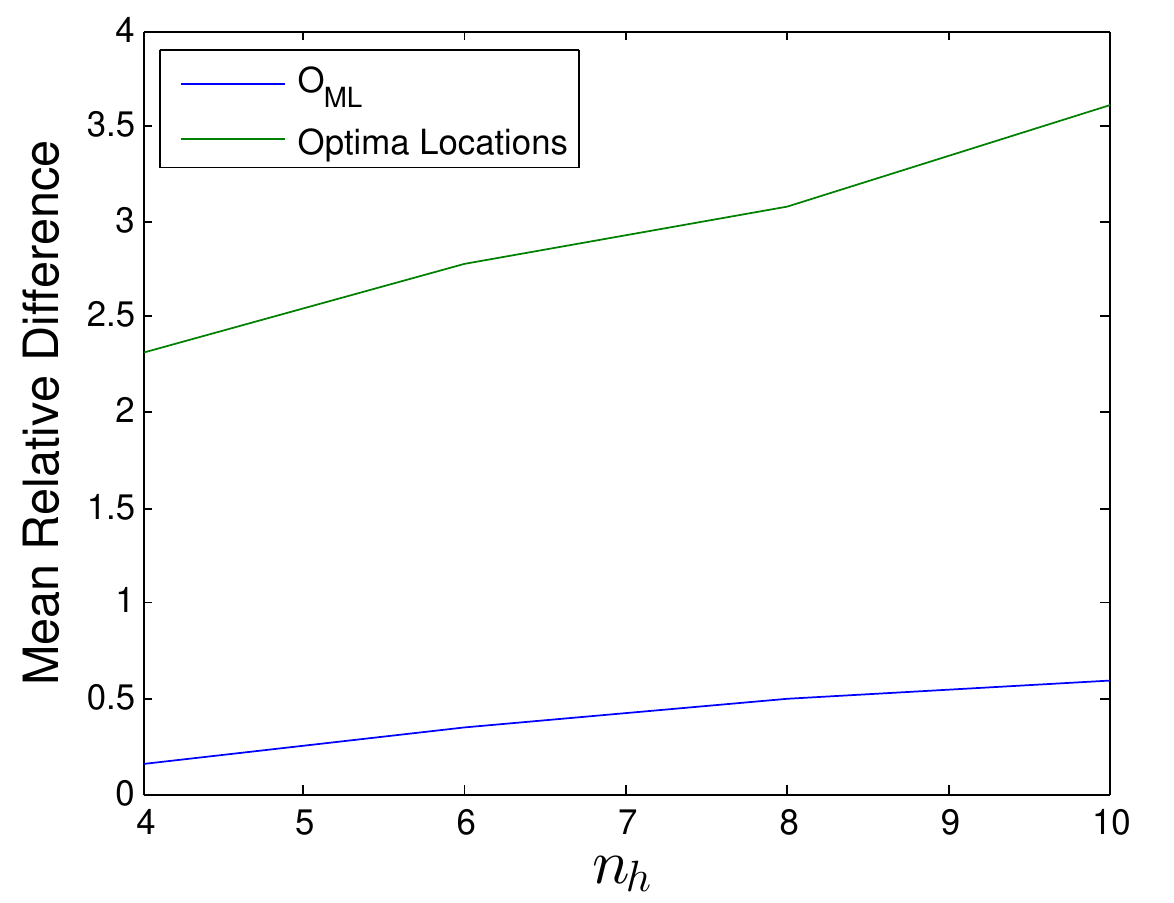}
\caption{Mean relative distance (measured in percent discrepancy) between ML optima and the corresponding relative difference in $O_{\rm ML}$ at the locations of these optima for CDML experiments on RBMs trained using MNIST examples.  A regularization constant of $\lambda=0.01$ is used in all these experiments.\label{fig:realCDML}}
\end{figure}

Next we examine the dependence on the quality of the Gibbs state prepared as a function of $\kappa$ and $n_h$ for models that are trained using ML optimization.  We choose these models as they are typical examples of models inferred towards the end of the training process, whereas the scaling of $\kappa$ found for random models typifies models that arise at the beginning of the training.  The results we obtain in~\fig{kappaReal} show qualitatively similar behavior to those observed for the synthetic examples.  We see that, in all cases the mean--field ansatz initially does a surprisingly good job of predicting the probability of a configuration: it under--estimates roughly $10-15\%$ of the probability mass.  Also a modest amount of hedging results in substantial improvements in the ability of the system to exactly prepare the Gibbs state, with a value of $\kappa$ that is less than $1000$ needed in the vast majority of the cases considered.  The value $\alpha=0.5$ is unlikely to be optimal for the cases considered.  In practice, scanning over $\alpha$ to find an appropriate value may be preferable to choosing any of the three values of $\alpha$ given in~\fig{kappaReal}.

\begin{figure}[t!]
\begin{minipage}{0.45\linewidth}
\includegraphics[width=\textwidth]{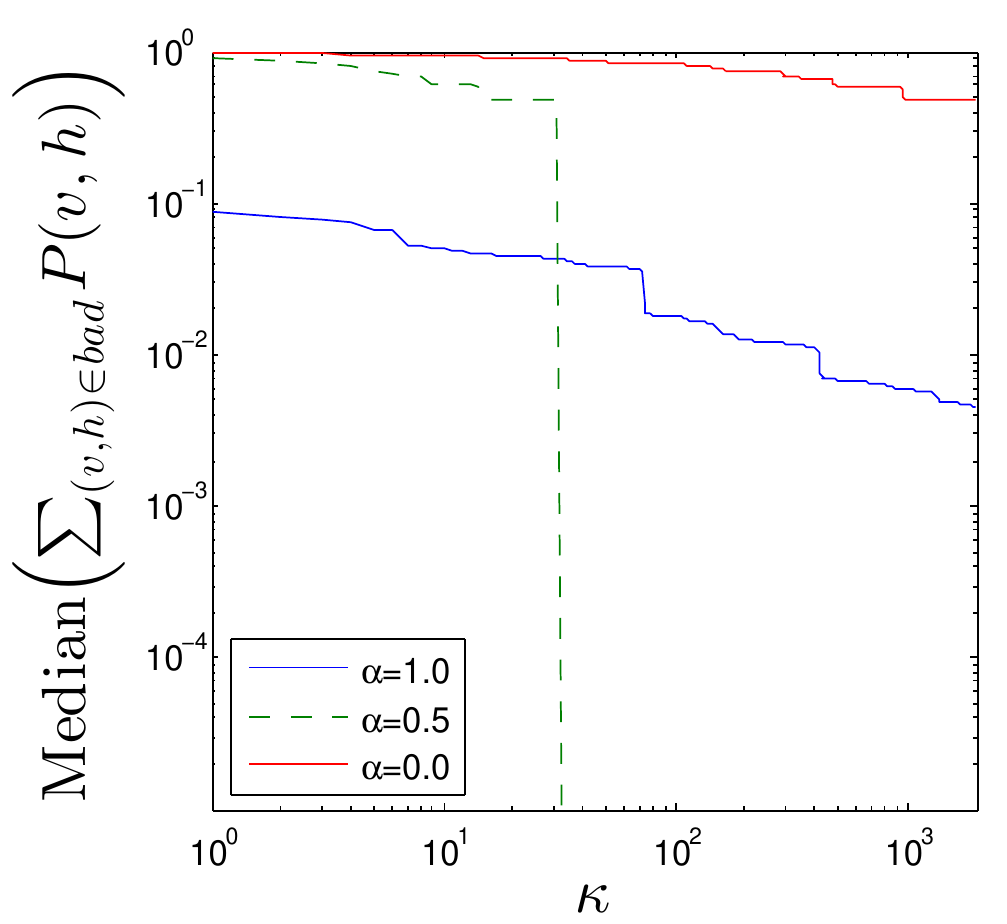}
(a) $n_h=4$
\end{minipage}
\hspace{1mm}
\begin{minipage}{0.45\linewidth}
\includegraphics[width=\linewidth]{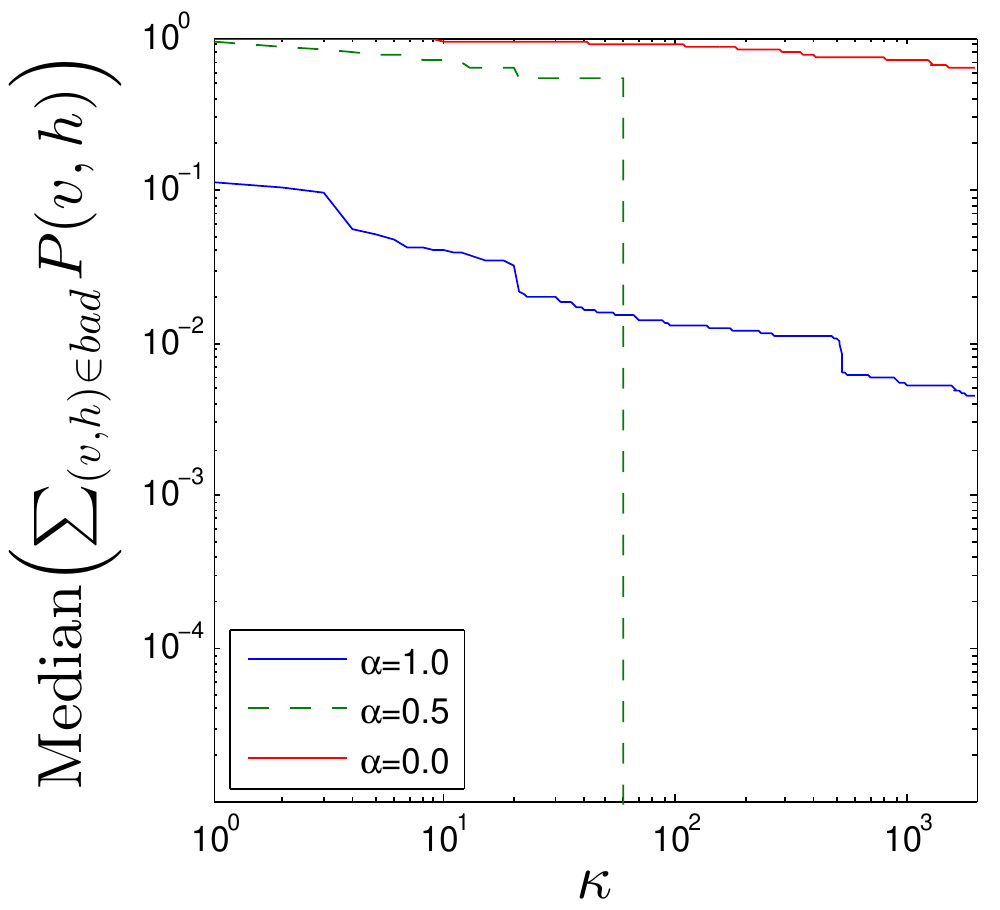}
(b) $n_h=6$
\end{minipage}

\begin{minipage}{0.45\linewidth}
\includegraphics[width=\linewidth]{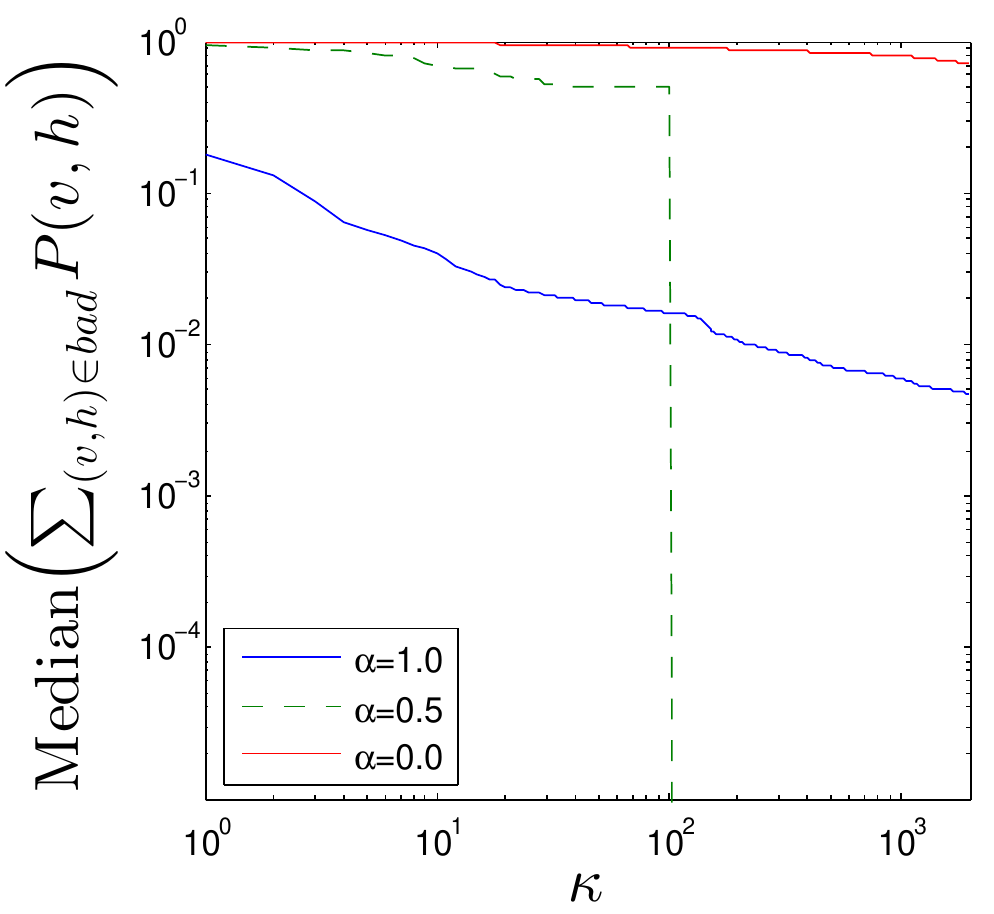}
(c) $n_h=8$
\end{minipage}
\hspace{1mm}
\begin{minipage}{0.45\linewidth}
\includegraphics[width=\linewidth]{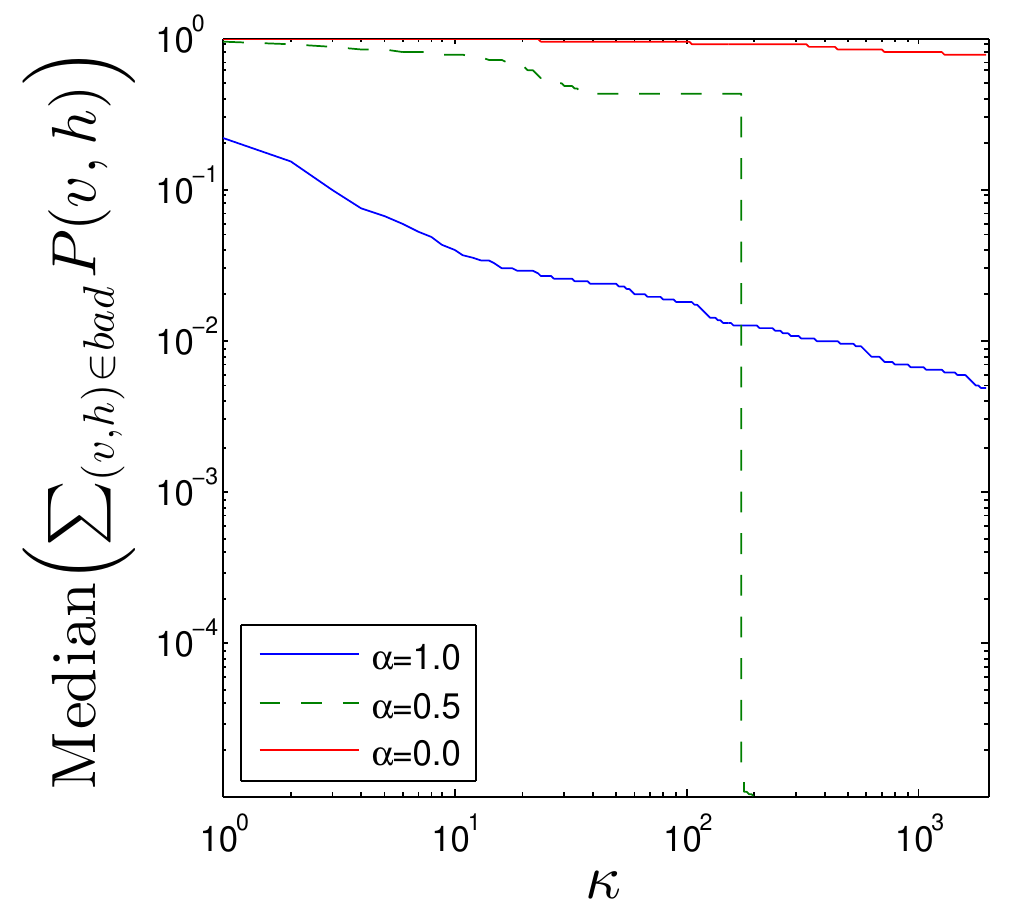}
(c) $n_h=10$
\end{minipage}
\caption{Fraction of probability for which $\mathcal{P}(v,h) \ge 1$ vs $\kappa$ for RBMs trained on truncated MNIST data with a varying number of hidden units and different hedging parameters.\label{fig:kappaReal}}
\end{figure}

The data in~\fig{kappaReal} shows that the mean--field approximation remains much more stable as we add more nodes to these graphical networks.  In particular, the median probability mass of the bad configurations is nearly a constant over all the data considered at $\kappa=2000$.  In contrast, we see evidence for slight variation in the median at $\kappa=1$.  The scaling of the value of $\kappa$ where the system transitions from imperfect state preparation to exact state preparation for $\alpha=0.5$ is unclear even for these small examples; the data is consistent with both a power--law scaling and an exponential scaling with the number of units.  Larger numerical experiments may be needed to distinguish these two possibilities, but in either case the scaling with $\alpha=0.5$ is qualitatively different than that observed for the mean--field initial state $(\alpha=1)$.  Most importantly the results show that nothing qualitatively changes when we transition from synthetic training data to subsampled MNIST training data.

Both the real and synthetic examples considered suffer from the fact that very strong correlations emerge in the model owing to the strong patterns that are present in our training data.  Such correlations are likely deleterious for the mean--field approximation and so structured mean--field approximations may lead to much better fidelity in the small examples that can be simulated using a classical computer.  The issue of how to choose an initial prior distribution for the true likelihoods of the configurations of the system remains an important issue in the field and provides an important way in which our quantum deep learning algorithms can be optimized to improve both the speed of training and the quality of the resultant models.

\section{Review of mean--field theory}\label{sec:meanfield}

The mean--field approximation is a variational approach that finds an uncorrelated distribution, $Q(v,h)$, that has minimal KL--divergence with the joint probability distribution $P(v,h)$ given by the Gibbs distribution.  The main benefit of using $Q$ instead of $P$ is that $\left\langle v_ih_j \right\rangle_{\rm model}$ and $\log(Z)$ can be efficiently estimated using mean--field approximations~\cite{Jor99}.  A secondary benefit is that the mean--field state can be efficiently prepared using single--qubit rotations.
More concretely, the mean--field approximation is a distribution such that
\begin{equation}
Q(v,h) = \left(\prod_{i} \mu_i^{v_i}(1-\mu_i)^{1-v_i}\right)\left(\prod_{j} \nu_j^{h_j}(1-\nu_j)^{1-h_j} \right),\label{eq:Qdef}
\end{equation}
where $\mu_i$ and $\nu_j$ are chosen to minimize ${\rm KL}(Q||P)$.  The parameters $\mu_i$ and $\nu_j$ are called mean--field parameters.

Using the properties of the Bernouli distribution, it is easy to see that
\begin{align}
{\rm KL}(Q|| P)  &= \sum_{v,h} -Q(v,h) \ln(P(v,h)) + Q(v,h) \ln(Q(v,h)),\nonumber\\
&= \sum_{v,h} Q(v,h)\left(\sum_i v_i b_i + \sum_j h_j d_j + \sum_{i,j}w_{i,j}v_ih_j+\ln Z\right) +Q(v,h)\ln(Q(v,h))\nonumber\\
&= \sum_i \mu_i b_i + \sum_j \nu_j d_j + \sum_{i,j}w_{i,j}\mu_i\nu_j  +\ln(Z) \nonumber\\
&\qquad\qquad + \sum_i \mu_i\ln(\mu_i) +  (1-\mu_i)\ln(1-\mu_i)+\sum_j \nu_j\ln(\nu_j) +  (1-\nu_j)\ln(1-\nu_j).
\end{align}
The optimal values of $\mu_i$ and $\nu_i$ can be found by differentiating this equation with respect to $\mu_i$ and $\nu_i$ and setting the result equal to zero.  The solution to this is
\begin{align}
\mu_i &= \sigma(-b_i -\sum_j w_{i,j} \nu_j)\nonumber\\
\nu_j &= \sigma(-d_j -\sum_i w_{i,j} \mu_i),\label{eq:update}
\end{align}
where $\sigma(x) = 1/(1+\exp(-x))$ is the sigmoid function.
These equations can be implicitly solved by fixed point iteration, which  involves initializing the $\mu_i$ and $\nu_j$ arbitrarily and iterate these equations until convergence is reached.  Convergence is guaranteed provided that the norm of the Jacobian of the map is bounded above by 1.  Solving the mean--field equations by fixed point iteration is analogous to Gibbs sampling with the difference being that here there are only a polynomial number of configurations to sample over and so the entire process is efficient. The generalization of this process to deep networks is straight forward and is discussed in~\cite{Ben09}.

Mean--field approximations to distributions such as $P(v,h)=\delta_{v,x} \exp^{-E(x,h)}/Z_x$ can be computed using the exact same methodology.  The only difference is that in such cases the visible units in the mean--field approximation is only taken over the hidden units.  Such approximations are needed to compute the expectations over the data that are needed to estimate the derivatives of $O_{\rm ML}$ in our algorithms.


It is also easy to see from the above argument that among all product distributions, $Q$ is the distribution that leads to the least error in the approximation to the log--partition function in~\eq{Zbd}.  This is because
\begin{equation}
\log(Z_{Q}) = \log(Z) -{\rm KL}(Q||P),\label{eq:Zmfbd}
\end{equation}
and the mean--field parameters found by solving~\eq{update} minimize the KL--divergence among all product distributions.  It is also interesting to note that all such approximations are lower bounds for the log--partition function because ${\rm KL}(Q||P)\ge 0$.

Experimentally, mean--field approximations can estimate the log-partition function within less than $1\%$ error~\cite{LK00} depending on the weight distribution and the geometry of the graph used.  We further show in~\sec{kappa} that the mean--field approximation to the partition function is sufficiently accurate for small restricted Boltzmann machines. Structured mean--field approximation methods~\cite{Xin02}, TAP~\cite{OW01} or AIS~\cite{SM08,SH09} can be used to reduce such errors if needed, albeit at a higher classical computational cost. 

These results also suggest the following result, which shows that the success probability of our state preparation method approaches $1$ in the limit where the strengths of the correlations in the model vanish.
\begin{corollary}
The success probability in~\lem{succ} approaches $1$  as $\max_{i,j}|w_{ij}| \rightarrow 0$.
\end{corollary}
\begin{proof}
The energy is a continuous function of $w$ and therefore $e^{-E(v,h)}/Z$ is also a continuous function of $w$.  Therefore, $\lim_{w\rightarrow 0} P(v,h) = e^{-\sum_i b_i v_i -\sum_j b_j h_j}/ \sum_{v,h}e^{-\sum_i b_i v_i -\sum_j b_j h_j}$.  Under such circumstances, $P(v,h)$ factorizes and so there exist $\tilde \mu_i$ and $\tilde \nu_j$ such that
\begin{equation}
\lim_{w\rightarrow 0}P(v,h) = \left(\prod_{i} \tilde \mu_i^{v_i}(1-\tilde \mu_i)^{1-v_i}\right)\left(\prod_{j} \tilde \nu_j^{h_j}(1-\tilde \nu_j)^{1-h_j} \right).\label{eq:Qdef2}
\end{equation}
Hence it follows from~\eq{Qdef} that there exists a mean--field solution such that ${\rm KL}(Q||\lim_{w \rightarrow 0}P)=0$.  Since the solution to~\eq{update} is unique when $w_{i,j}=0$ it follows that the mean--field solution found must be the global optima and hence there ${\rm KL}(Q||P)$ approaches $0$ as $\max_{i,j}|w_{i,j}|\rightarrow 0$.  Therefore~\eq{Zmfbd} implies that $Z_{Q} \rightarrow Z_{x,{Q}}$ in the limit.  Hence as $\max_{i,j}|w_{i,j}| \rightarrow 0$ we can take $\kappa=1$ and $Z /Z_{Q}=1$.  Therefore the success probability approaches $1$ if the optimal value of $\kappa$ is chosen.   The same argument also applies for $Z_x/Z_{x,{Q}}$.
\end{proof}

\section{Review of contrastive divergence training}\label{sec:CD}
The idea behind contrastive divergence is straighforward.  The model average in the expression for the gradient of the average log--likelihood given in the main body can be computed by sampling from the Gibbs distribution $P$.  This process is not tractable classically, so contrastive divergence samples from an approximation to the Gibbs distribution found by applying a finite number of rounds of Gibbs sampling.  The resultant samples drawn from the distribution are then, ideally, drawn from a distribution that is close to the true Gibbs distribution $P$.

Gibbs sampling proceeds as follows.  First the visible units are set to a training vector.  Then the hidden units are set to $1$ with probability $P(h_j=1|v)=\sigma(-d_j -\sum_i w_{i,j} v_i)$.  Once the hidden units are set, the visible units are reset to $1$ with probability $P(v_i=1|h)=\sigma(-b_i -\sum_j w_{i,j} h_j)$.  This process can then be repeated using the newly generated training vector in place of the original $v$.  As the number of rounds of Gibbs sampling increases, the distribution over resultant samples approaches that of the true Gibbs distribution.

The simplest contrastive divergence algorithm, CD$-1$, works by using only one round of Gibbs sampling to reset the visible units.  The probability that each of the hidden units is $1$ is then computed using $P(h_j=1|v)=\sigma(-d_j -\sum_i w_{i,j} v_i)$.  These probabilities are stored and the process of Gibbs sampling and probability computation is repeated for each training vector.  The probabilities necessary for the model average are then set, for each $w_{i,j}$, to be the average of all the probabilities ${\rm Pr}(v_i=1, h_j=1)$ computed in the prior samples.  Closer approximations to the Gibbs distribution can be found by using more steps of Gibbs sampling~\cite{CH05,BD07,Tie08}.  For example, CD$-10$ uses ten steps of Gibbs sampling rather than one and tends to give much better approximations to the true gradient of the objective function.

Contrastive divergence earns its name because it does not try to approximate the gradients of the ML-objective function; rather, CD$-n$ approximately optimizes the difference between the average log--likelihood after zero and $n$ rounds of Gibbs sampling: ${\rm KL}(p_0||p_\infty)-{\rm KL}(p_n||p_\infty)$~\cite{Hin02}.  Also as $n\rightarrow \infty$ the contrastive divergence objective becomes the average log--likelihood, which is $O_{\rm ML}$ in the absence of regularization.  This means that asymptotically CD$-n$ approximates the correct derivatives.  Although contrastive divergence approximates the gradient of the contrastive divergence objective function, the gradients yielded by the algorithm are not precisely the gradients of any objective function~\cite{ST10}.  Thus the analogy of contrastive divergence optimizing an objective function that is close to $O_{\rm ML}$ is inexact.

Although it is efficient, there are several drawbacks to contrastive divergence training.
The main drawback of CD is that it does not permit interactions between hidden and visible units.  This restricts the allowable class of graphical models.  Additionally, the training process can take hours to days~\cite{Tie08}.  Finally, the method does not directly apply to training deep networks.  In order to train deep restricted Boltzmann machines, layer-wise training is typically employed, which breaks the undirected structure of the network and potentially leads to sub--optimal models.  Parallelism can accelerate the training process for training deep RBMs using contrastive divergence, but only to a limited extent because of the sequential nature of the updates on the (now directed) graph.  Our work shows that quantum computing can circumvent such restrictions.

\section{Review of quantum computing}\label{sec:QC}
In quantum information processing (QIP), information is stored in a quantum bit, or {\it qubit}, which is analogous to a classical bit.
Whereas a classical bit has a state value $s\in \{0,1\}$, a qubit state $\ket{\psi}$ is actually a linear {\it superposition} of states:
\begin{equation}
\ket{\psi} = \alpha\ket{0} + \beta\ket{1},
\end{equation}
where the $\{0,1\}$ basis state vectors are represented in Dirac notation (ket vectors) as
$\ket{0} = \begin{bmatrix} 1 & 0 \end{bmatrix}^T$, and
$\ket{1} = \begin{bmatrix} 0 & 1 \end{bmatrix}^T$, respectively.
The {\it amplitudes} $\alpha$ and $\beta$ are complex numbers that satisfy the normalization condition: $|\alpha|^2 + |\beta|^2 = 1$.
Upon {\it measurement} of the quantum state $\ket{\psi}$, either state $\ket{0}$ or $\ket{1}$ is observed with probability $|\alpha|^2$ or $|\beta|^2$, respectively.
Note that a $n$-qubit quantum state is a $2^n \times 1$-dimensional state vector, where each entry represents the amplitude of the corresponding basis state.
Therefore, $n$ qubits live in a $2^n$-dimensional Hilbert space, and we can represent a superposition over $2^n$ states as:
\begin{equation}
\ket{\psi} = \sum_{i=0}^{2^n - 1} \alpha_i \ket{i},
\end{equation}
where $\alpha_i$ are complex amplitudes that satisfy the condition $\sum_i |\alpha_i|^2 = 1$, and $i$ is the binary representation of integer $i$.
Note, for example, that the state $\ket{0000}$ is equivalent to writing the tensor product of the four states: $\ket{0}\otimes\ket{0}\otimes\ket{0}\otimes\ket{0} = \ket{0}^{\otimes 4} = \begin{bmatrix} 1 & 0 & 0 & 0 & 0 & 0 & 0 & 0 \end{bmatrix}^T$.
The ability to represent a superposition over exponentially many states with only a linear number of qubits is one of the essential ingredients of a quantum algorithm --- an innate massive parallelism.

A quantum computation proceeds through the {\it unitary} evolution of a quantum state; in turn, quantum operations are necessarily {\it reversible}.
We refer to quantum unitary operations as quantum {\it gates}.
Note that measurement is not reversible;
it collapses the quantum state to the observed value, thereby erasing the knowledge of the amplitudes $\alpha$ and $\beta$.

An $n$-qubit quantum gate is a $2^n \times 2^n$ unitary matrix acting on an $n$-qubit quantum state.
For example, the {\it Hadamard} gate maps
$\ket{0} \rightarrow \frac{1}{\sqrt{2}}\left( \ket{0} + \ket{1}\right)$, and
$\ket{1} \rightarrow \frac{1}{\sqrt{2}}\left( \ket{0} - \ket{1}\right)$.
An $X$ gate, similar to a classical NOT gate, maps
$\ket{0}\rightarrow \ket{1}$, and
$\ket{1}\rightarrow \ket{0}$.
The identity gate is represented by $I$.
The two-qubit {\it controlled-NOT} gate, $CX$, maps $\ket{x,y}\rightarrow\ket{x, x\oplus y}$.
It is convenient to include a further gate, $T$, which is known as a $\pi/8$--gate and is needed to make the above quantum gate set complete.
The corresponding unitary matrices are given by:
\begin{equation}
H=\frac{1}{\sqrt{2}}\begin{bmatrix} 1 & 1 \\ 1 & -1 \end{bmatrix}, X=\begin{bmatrix} 0 & 1 \\ 1 & 0 \end{bmatrix}, I=\begin{bmatrix} 1 & 0 \\ 0 & 1 \end{bmatrix}, CX=\begin{bmatrix} 1 & 0 & 0 & 0 \\ 0 & 1 & 0 & 0 \\ 0 & 0 & 0 & 1 \\ 0 & 0 & 1 & 0 \end{bmatrix}, T=\begin{bmatrix}1 & 0 \\ 0& e^{i\pi/4} \end{bmatrix}.
\end{equation}

The single qubit rotation is an important operation for quantum computation.  The single qubit rotation, $R_y(2\theta)$, which under the isomorphism between SO(3) and SU(2) corresponds to a rotation of the state vector about the $y$--axis where the states $\ket{0}$ and $\ket{1}$ are computational basis states.  The gate is defined below.
\begin{align}
R_y(2\theta) =\begin{bmatrix}\cos{\theta} & -\sin(\theta) \\ \sin(\theta) &\cos(\theta) \end{bmatrix}.
\end{align}
Unlike the previous gates, single qubit rotations are not discrete.  They can, however, be approximated to within arbitrarily small error using a sequence of fundamental (discrete) quantum operations~\cite{KMM+13,RS14,BRS14}.



\begin{thebibliography}{10}

\bibitem{HOT06}
Geoffrey Hinton, Simon Osindero, and Yee-Whye Teh.
\newblock A fast learning algorithm for deep belief nets.
\newblock {\em Neural computation}, 18(7):1527--1554, 2006.

\bibitem{CW08}
Ronan Collobert and Jason Weston.
\newblock A unified architecture for natural language processing: Deep neural
  networks with multitask learning.
\newblock In {\em Proceedings of the 25th international conference on Machine
  learning}, pages 160--167. ACM, 2008.

\bibitem{Ben09}
Yoshua Bengio.
\newblock Learning deep architectures for ai.
\newblock {\em Foundations and trends{\textregistered} in Machine Learning},
  2(1):1--127, 2009.

\bibitem{LYK+10}
Yann LeCun, Koray Kavukcuoglu, and Cl{\'e}ment Farabet.
\newblock Convolutional networks and applications in vision.
\newblock In {\em Circuits and Systems (ISCAS), Proceedings of 2010 IEEE
  International Symposium on}, pages 253--256. IEEE, 2010.

\bibitem{Hin02}
Geoffrey~E Hinton.
\newblock Training products of experts by minimizing contrastive divergence.
\newblock {\em Neural computation}, 14(8):1771--1800, 2002.

\bibitem{SMH07}
Ruslan Salakhutdinov, Andriy Mnih, and Geoffrey Hinton.
\newblock Restricted boltzmann machines for collaborative filtering.
\newblock In {\em Proceedings of the 24th international conference on Machine
  learning}, pages 791--798. ACM, 2007.

\bibitem{Tie08}
Tijmen Tieleman.
\newblock Training restricted boltzmann machines using approximations to the
  likelihood gradient.
\newblock In {\em Proceedings of the 25th international conference on Machine
  learning}, pages 1064--1071. ACM, 2008.

\bibitem{SH09}
Ruslan Salakhutdinov and Geoffrey~E Hinton.
\newblock Deep boltzmann machines.
\newblock In {\em International Conference on Artificial Intelligence and
  Statistics}, pages 448--455, 2009.

\bibitem{ST10}
Ilya Sutskever and Tijmen Tieleman.
\newblock On the convergence properties of contrastive divergence.
\newblock In {\em International Conference on Artificial Intelligence and
  Statistics}, pages 789--795, 2010.

\bibitem{TH09}
Tijmen Tieleman and Geoffrey Hinton.
\newblock Using fast weights to improve persistent contrastive divergence.
\newblock In {\em Proceedings of the 26th Annual International Conference on
  Machine Learning}, pages 1033--1040. ACM, 2009.

\bibitem{BD07}
Yoshua Bengio and Olivier Delalleau.
\newblock Justifying and generalizing contrastive divergence.
\newblock {\em Neural Computation}, 21(6):1601--1621, 2009.

\bibitem{FI11}
Asja Fischer and Christian Igel.
\newblock Bounding the bias of contrastive divergence learning.
\newblock {\em Neural computation}, 23(3):664--673, 2011.

\bibitem{LB97}
Daniel~A Lidar and Ofer Biham.
\newblock Simulating ising spin glasses on a quantum computer.
\newblock {\em Physical Review E}, 56(3):3661, 1997.

\bibitem{TD98}
Barbara~M Terhal and David~P DiVincenzo.
\newblock The problem of equilibration and the computation of correlation
  functions on a quantum computer.
\newblock {\em arXiv preprint quant-ph/9810063}, 1998.

\bibitem{PW09}
David Poulin and Pawel Wocjan.
\newblock Sampling from the thermal quantum gibbs state and evaluating
  partition functions with a quantum computer.
\newblock {\em Physical review letters}, 103(22):220502, 2009.

\bibitem{DF11}
Misha Denil and Nando De~Freitas.
\newblock Toward the implementation of a quantum rbm.
\newblock In {\em NIPS Deep Learning and Unsupervised Feature Learning
  Workshop}, 2011.

\bibitem{ORR13}
Maris Ozols, Martin Roetteler, and J{\'e}r{\'e}mie Roland.
\newblock Quantum rejection sampling.
\newblock {\em ACM Transactions on Computation Theory (TOCT)}, 5(3):11, 2013.

\bibitem{Jor99}
Michael~I Jordan, Zoubin Ghahramani, Tommi~S Jaakkola, and Lawrence~K Saul.
\newblock An introduction to variational methods for graphical models.
\newblock {\em Machine learning}, 37(2):183--233, 1999.

\bibitem{WH02}
Max Welling and Geoffrey~E Hinton.
\newblock A new learning algorithm for mean field boltzmann machines.
\newblock In {\em Artificial Neural Networks—ICANN 2002}, pages 351--357.
  Springer, 2002.

\bibitem{Xin02}
Eric~P Xing, Michael~I Jordan, and Stuart Russell.
\newblock A generalized mean field algorithm for variational inference in
  exponential families.
\newblock In {\em Proceedings of the Nineteenth conference on Uncertainty in
  Artificial Intelligence}, pages 583--591. Morgan Kaufmann Publishers Inc.,
  2002.

\bibitem{BHM+00}
Gilles Brassard, Peter Hoyer, Michele Mosca, and Alain Tapp.
\newblock Quantum amplitude amplification and estimation.
\newblock {\em arXiv preprint quant-ph/0005055}, 2000.

\bibitem{ABG06}
Esma A{\"\i}meur, Gilles Brassard, and S{\'e}bastien Gambs.
\newblock Machine learning in a quantum world.
\newblock In {\em Advances in Artificial Intelligence}, pages 431--442.
  Springer, 2006.

\bibitem{LMR13}
Seth Lloyd, Masoud Mohseni, and Patrick Rebentrost.
\newblock Quantum algorithms for supervised and unsupervised machine learning.
\newblock {\em arXiv preprint arXiv:1307.0411}, 2013.

\bibitem{RML13}
Patrick Rebentrost, Masoud Mohseni, and Seth Lloyd.
\newblock Quantum support vector machine for big feature and big data
  classification.
\newblock {\em arXiv preprint arXiv:1307.0471}, 2013.

\bibitem{QKS15}
Nathan Wiebe, Ashish Kapoor, and Krysta Svore.
\newblock Quantum nearest-neighbor algorithms for machine learning.
\newblock {\em QIC}, 15:318--358, 2015.

\bibitem{NC00}
Michael~A Nielsen and Isaac~L Chuang.
\newblock {\em Quantum computation and quantum information}.
\newblock Cambridge university press, 2010.

\bibitem{GLM08}
Vittorio Giovannetti, Seth Lloyd, and Lorenzo Maccone.
\newblock Quantum random access memory.
\newblock {\em Physical review letters}, 100(16):160501, 2008.

\bibitem{WR14}
Nathan Wiebe and Martin Roetteler.
\newblock Quantum arithmetic and numerical analysis using repeat-until-success
  circuits.
\newblock {\em arXiv preprint arXiv:1406.2040}, 2014.

\bibitem{ABG07}
Esma A{\"\i}meur, Gilles Brassard, and S{\'e}bastien Gambs.
\newblock Quantum clustering algorithms.
\newblock In {\em Proceedings of the 24th international conference on machine
  learning}, pages 1--8. ACM, 2007.

\bibitem{HDY+12}
Geoffrey Hinton, Li~Deng, Dong Yu, George~E Dahl, Abdel-rahman Mohamed, Navdeep
  Jaitly, Andrew Senior, Vincent Vanhoucke, Patrick Nguyen, Tara~N Sainath,
  et~al.
\newblock Deep neural networks for acoustic modeling in speech recognition: The
  shared views of four research groups.
\newblock {\em Signal Processing Magazine, IEEE}, 29(6):82--97, 2012.

\bibitem{lecun1998mnist}
Yann LeCun and Corinna Cortes.
\newblock The mnist database of handwritten digits, 1998.

\bibitem{Gro96}
Lov~K Grover.
\newblock A fast quantum mechanical algorithm for database search.
\newblock In {\em Proceedings of the twenty-eighth annual ACM symposium on
  Theory of computing}, pages 212--219. ACM, 1996.

\bibitem{KMM+13}
Vadym Kliuchnikov, Dmitri Maslov, and Michele Mosca.
\newblock Fast and efficient exact synthesis of single-qubit unitaries
  generated by clifford and t gates.
\newblock {\em Quantum Information \& Computation}, 13(7-8):607--630, 2013.

\bibitem{RS14}
Neil~J Ross and Peter Selinger.
\newblock Optimal ancilla-free clifford+ t approximation of z-rotations.
\newblock {\em arXiv preprint arXiv:1403.2975}, 2014.

\bibitem{BRS14}
Alex Bocharov, Martin Roetteler, and Krysta~M Svore.
\newblock Efficient synthesis of universal repeat-until-success circuits.
\newblock {\em arXiv preprint arXiv:1404.5320}, 2014.

\bibitem{CW12}
Andrew~M Childs and Nathan Wiebe.
\newblock Hamiltonian simulation using linear combinations of unitary
  operations.
\newblock {\em arXiv preprint arXiv:1202.5822}, 2012.

\bibitem{DSB09}
Pinar Donmez, Krysta~M Svore, and Christopher~JC Burges.
\newblock On the local optimality of lambdarank.
\newblock In {\em Proceedings of the 32nd international ACM SIGIR conference on
  Research and development in information retrieval}, pages 460--467. ACM,
  2009.

\bibitem{CH05}
Miguel~A Carreira-Perpinan and Geoffrey~E Hinton.
\newblock On contrastive divergence learning.
\newblock In {\em Proceedings of the tenth international workshop on artificial
  intelligence and statistics}, pages 33--40. Citeseer, 2005.

\bibitem{Wai05}
Martin~J Wainwright, Tommi~S Jaakkola, and Alan~S Willsky.
\newblock A new class of upper bounds on the log partition function.
\newblock {\em Information Theory, IEEE Transactions on}, 51(7):2313--2335,
  2005.

\bibitem{OW01}
Manfred Opper and Ole Winther.
\newblock Tractable approximations for probabilistic models: The adaptive
  thouless-anderson-palmer mean field approach.
\newblock {\em Physical Review Letters}, 86(17):3695, 2001.

\bibitem{SM08}
Ruslan Salakhutdinov and Iain Murray.
\newblock On the quantitative analysis of deep belief networks.
\newblock In {\em Proceedings of the 25th international conference on Machine
  learning}, pages 872--879. ACM, 2008.

\bibitem{LK00}
Martijn~AR Leisink and Hilbert~J Kappen.
\newblock A tighter bound for graphical models.
\newblock In {\em NIPS}, volume~13, pages 266--272, 2000.

\end{thebibliography}

\end{document}